\documentclass[a4paper, 11pt]{article}
\usepackage[dvipsnames]{xcolor}
\usepackage{amsmath}
\usepackage{bbm}
\usepackage{amsthm}
\usepackage{amssymb}
\usepackage{dsfont}
\usepackage{comment}
\usepackage{authblk}
\usepackage[utf8]{inputenc}

\usepackage[active]{srcltx}
\usepackage{stackengine}
\usepackage[bookmarksnumbered=true]{hyperref}

\usepackage{geometry}
\newtheorem{definition}{Definition}[section]
\newtheorem{lemma}[definition]{Lemma}
\newtheorem{proposition}[definition]{Proposition}
\newtheorem{theorem}[definition]{Theorem}

\newtheorem{assumption}[definition]{Assumption}
\theoremstyle{definition}
\newtheorem{remark}[definition]{Remark}

\numberwithin{equation}{section}

\renewcommand{\epsilon}{\varepsilon}

\renewcommand{\phi}{\varphi}
\newcommand{\R}{\mathbb{R}}
\newcommand{\Ss}{\mathbb{S}}

\newcommand{\Z}{\mathbb{Z}}
\newcommand\eps\varepsilon

\newcommand\udarrow{{\uparrow\!\downarrow}}
\renewcommand\rho\varrho

\newcommand\ind{\mathds{1}}

\begin{document}

\title{The Huang--Yang formula for the low-density Fermi gas:  upper bound}

\author[1]{Emanuela L. Giacomelli}
\affil[1]{LMU Munich, Department of Mathematics, Theresienstr. 39, 80333 Munich, Germany}
\author[1]{Christian Hainzl}
\author[1]{Phan Thành Nam}

\author[2]{Robert Seiringer}
\affil[2]{Institute of Science and Technology Austria, Am Campus 1, 3400 Klosterneuburg, Austria}

\maketitle

\abstract{We study the ground state energy of a gas of spin $1/2$ fermions with repulsive short-range interactions. We derive an upper bound that agrees, at low density $\rho$,  with the Huang--Yang conjecture. The latter captures the first three terms in an asymptotic low-density expansion, and in particular the Huang--Yang correction term of order $\rho^{7/3}$. Our trial state is constructed using an adaptation of the bosonic Bogoliubov theory to the Fermi system,  where the correlation structure of fermionic particles is incorporated by quasi-bosonic Bogoliubov transformations. In the latter, it is important to consider a modified zero-energy scattering equation that takes into account the presence of the Fermi sea, in the spirit of the Bethe--Goldstone equation.}

\tableofcontents

\section{Introduction and main result}

Establishing asymptotic formulas for the ground state energy of dilute gases has been an active area in Mathematical Physics in recent years, partly motivated by the advances in the physics of cold atoms and molecules. In the bosonic case, the validity of the Lee--Huang--Yang formula \cite{LHY} for the first two terms in a low-density expansion was recently established \cite{BCS,FS-20,FS-23,YY}, following earlier work  on the first term in \cite{Dys57,LY98}. In this work, we are concerned with the analogous question for fermions, for which the Huang--Yang formula \cite{HY} gives the first {\em three} terms in an asymptotic expansion at low density. We shall establish the validity of this formula, at least as an upper bound. 

We consider a system of $N$ fermions with spin $1/2$ 
in a box $\Lambda:= [-L/2,L/2]^3\subset\mathbb{R}^3$ with periodic boundary conditions. The interaction between particles is described by the periodization $V:\Lambda\to \R$ of a potential $V_\infty:\R^3\to \R$, i.e., $V(x) = \sum_{z\in\Z^3} V_\infty(x+Lz)$, where 
we assume that $V_\infty$ is nonnegative, radial and compactly supported. 
The Hamiltonian describing the system is given by 
\begin{equation}\label{eq: hamiltonian}
  H_N := -\sum_{j=1}^N\Delta_{x_j} + \sum_{1\leq i< j\leq N} V(x_i - x_j), 
\end{equation}
acting on $\bigwedge^N L^2(\Lambda,\mathbb{C}^2)$, the  Hilbert space of antisymmetric square-integrable functions of $N$ space-spin variables $(x_i,\sigma_i)$, with $x_i\in \Lambda$ and  $\sigma_i \in \{\uparrow,\downarrow\}$.  Since $H_N$ is spin-independent, if $N= N_\uparrow + N_\downarrow$, then $H_N$ leaves invariant the subspace $\mathfrak{h}(N_\uparrow, N_\downarrow) \subset \bigwedge^N L^2(\Lambda,\mathbb{C}^2)$ consisting of $N$-body wave functions with exactly $N_{\sigma}$ particles of spin $\sigma\in \{\uparrow,\downarrow\}$. 

The Hamiltonian $H_N$ is bounded from below and can be defined as a self-adjoint operator by Friedrichs' method. The ground state energy in the subspace $\mathfrak{h}(N_\uparrow, N_\downarrow)$ is given by 
\begin{equation}\label{eq: def gs energy}
  E_L(N_\uparrow, N_\downarrow) = \inf_{\psi\in \mathfrak{h}(N_\uparrow, N_\downarrow)}\frac{\langle \psi, H_N\psi\rangle}{\langle \psi, \psi\rangle}, 
\end{equation}
and the ground state energy density in the thermodynamic limit is
\begin{equation}\label{eq: def gs energy density}
  e(\rho_\uparrow, \rho_\downarrow) = \lim_{\substack{L \rightarrow \infty \\ N_\sigma/L^3 \rightarrow \rho_\sigma, \sigma \in \{\uparrow, \downarrow\}}}\frac{E_L(N_\uparrow, N_\downarrow)}{L^3}.
\end{equation}
It is well-known that the thermodynamic limit exists and is independent of the boundary conditions \cite{Ruelle-99,Robinson-71}.

By using a pseudopotential method, Huang and Yang \cite{HY} predicted that in the low-density limit $\rho_\uparrow = \rho_\downarrow = \rho/2 \to 0$,  
\begin{equation}\label{eq:HY}
  e(\rho/2, \rho/2) = \frac{3}{5}(3\pi^2)^{\frac{2}{3}}\rho^{\frac{5}{3}} + 2\pi a \rho^2 + \frac{4}{35}(11 -2\log2) (9\pi)^{\frac{2}{3}} a^2 \rho^{7/3} + o(\rho^{\frac{7}{3}})_{\rho\to 0}
  \end{equation}
  where $a>0$ is the $s$-wave scattering length of the interaction potential $V_\infty$, defined by
  $
  8\pi a = \int_{\R^3} V_\infty (1-\varphi_\infty)
  $
  with $\varphi_\infty:\R^3\to \R$ the solution of the zero-scattering equation 
\begin{equation}\label{eq: zero energy scatt eq-intro}
  2\Delta\varphi_\infty + V_\infty(1-\varphi_\infty) = 0, \qquad \varphi_\infty(x) \rightarrow 0\quad \mbox{as}\quad |x|\rightarrow \infty.
\end{equation}
The Huang--Yang conjecture was derived from a consideration of hard-sphere interactions, but it is expected that the formula extends to a large class of repulsive and short-range interaction potentials. The formula \eqref{eq:HY} is of great interest in the theory of dilute quantum gases, since it exhibits a remarkable universality, 
with the energy depending on the interaction potential only via its scattering length, at least to the order considered. We refer to \cite{NNCS} for an experimental verification.

On the mathematical side,  for a large class of repulsive interactions, including hard spheres, the validity of the first two terms of \eqref{eq:HY} was proved in   \cite{LSS-05}; more precisely,  \cite[Thm.~1]{LSS-05} states that 
\begin{align} \label{eq:LSS}
e(\rho_\uparrow, \rho_\downarrow)= \frac{3}{5}(6\pi^2)^{\frac{2}{3}} \left(  \rho_\uparrow^{5/3} + \rho_\downarrow^{5/3}\right)  +8\pi a \rho_\uparrow\rho_\downarrow + o(\rho^2)_{\rho\to 0}. 
\end{align}
The leading order term comes from the kinetic energy of a non-interacting Fermi gas, while the second-order term provides an effective description of the interaction between components of different spins. (In the case of fermions without spin, this second term vanishes; see \cite{LS2,LS1} for the corresponding result in this case.) An analogue of \eqref{eq:LSS} for the thermodynamic pressure at positive temperature was  derived in \cite{S06}.

An alternative proof of \eqref{eq:LSS} was recently given in \cite{FGHP}, by applying elements of bosonic Bogoliubov theory to Fermi systems. The method in \cite{FGHP} was later extended in \cite{Gia1} to derive an optimal error estimate $O(\rho^{7/3})$ in the upper bound, and will also be of crucial importance in our analysis here. (See also \cite{Lau} for an alternative approach.) 
%
In fact, in the present paper we shall  prove an upper bound containing the precise Huang--Yang correction of order $\rho^{7/3}$. We impose the following conditions on the interaction potential.  
\begin{assumption}\label{asu: potential V}
  The interaction potential $V_\infty\in L^2(\mathbb{R}^3)$ is nonnegative, radial and compactly supported.
  \end{assumption}

In order to state our main result, we shall need the function $F:\R_+\to \R_+$ defined as 
\begin{align}\label{eq:def-F}
F(x) =  \frac {(6\pi^2)^{1/3}}{35} \Biggl(  & {16 x^{7/3}}  \ln x  -   48 \left(x^{7/3} + 1\right)  \ln (1+x^{1/3}) \nonumber \\
&\quad +  6 \left( 15  x^{1/3} - 4x^{2/3} + 33 x + 33 x^{4/3}  -4 x^{5/3} + 15 x^{2}\right) \nonumber\\
 &\quad+   {21} \left( 1 - 6   x^{2/3} +5  x +5   x^{4/3} -6 x^{5/3}
 + x^{7/3} \right) \ln\frac {|1- x^{1/3}|}{1+x^{1/3}} 
  \Biggl).
\end{align}
It is a continuous, increasing function, with $F(1) = \frac{48}{35}(11 -2\log2) (6 \pi^2)^{\frac{1}{3}}$.

\begin{theorem}[Upper bound]\label{thm: main up bd} Let $V_\infty$ be as in Assumption \ref{asu: potential V}, and let $a>0$ be its scattering length. In the low-density limit $\rho_\uparrow+ \rho_\downarrow = \rho \rightarrow 0$, the ground state energy density defined in \eqref{eq: def gs energy density} satisfies
\begin{equation}\label{eq: main thm 1}
  e(\rho_\uparrow, \rho_\downarrow) \leq  \frac{3}{5}(6\pi^2)^{\frac{2}{3}} \left(\rho_\uparrow^{\frac{5}{3}} + \rho_\downarrow^{\frac{5}{3}}\right) + 8\pi a \rho_\uparrow\rho_\downarrow + a^2 \rho_\uparrow^{\frac{7}{3}} F\left(\frac{\rho_\downarrow}{\rho_\uparrow}\right) + O(\rho^{\frac{7}{3} +\frac{1}{9}}),
\end{equation}
where $F$ is defined in \eqref{eq:def-F}. 
In particular, if $\rho_\uparrow = \rho_\downarrow = \rho/2 \to 0$, we have
\begin{equation}\label{eq: main thm 2}
  e(\rho/2, \rho/2) \leq \frac{3}{5}(3\pi^2)^{\frac{2}{3}}\rho^{\frac{5}{3}} + 2\pi a \rho^2 + \frac{4}{35}(11 -2\log2) (9 \pi )^{\frac{2}{3}} a^2 \rho^{7/3} + O(\rho^{\frac{7}{3} + \frac{1}{9}}).
\end{equation}
\end{theorem}


As we will see in the proof, the contribution $a^2 \rho_\uparrow^{\frac{7}{3}} F({\rho_\downarrow}/{\rho_\uparrow})$ in \eqref{eq: main thm 1} comes from the expression 
\begin{equation}\label{inte:F}
 \frac{ a^2}{8\pi^7}\int_{\mathbb{R}^3} dp \int_{|r| \le k_F^\uparrow }dr \int_{|r^\prime| \le k_F^\downarrow}dr^\prime 
 \left(\frac{1}{2|p|^2} - \frac{\ind_{\{|r+p|>k_{F}^\uparrow\}} \ind_{\{|r^\prime-p|>k_{F}^\downarrow\}}}{|r+p|^2 - |r|^2 + |r^\prime -p|^2 - |r^\prime|^2} \right)
\end{equation}
with $k_{F}^\sigma= (6\pi^2)^{\frac{1}{3}}\rho_\sigma^{1/3}$ the Fermi momenta of the two spin components. In particular, the roles of $\rho_\uparrow$ and $\rho_\downarrow$ are symmetric. In fact, $\rho_\uparrow^{\frac{7}{3}} F\left({\rho_\downarrow}/{\rho_\uparrow}\right)=\rho_\downarrow^{\frac{7}{3}} F\left({\rho_\uparrow}/{\rho_\downarrow}\right)$ since $F(x^{-1})= x^{-7/3} F(x)$ for $x>0$. 
The explicit form of $F$ in \eqref{eq:def-F} has appeared in the physics literature in \cite{Kanno} (see also \cite{ChaWoj,Pera}).

The Huang--Yang correction in \eqref{eq:HY} can be interpreted as the fermionic analogue of the Lee--Huang--Yang correction for the energy of a dilute Bose gas, in the sense that it captures the next to leading order contribution of the interactions, beyond the $\rho^2$ term. The validity of the Lee--Huang--Yang correction was proved in \cite{FS-20,FS-23} (lower bound) and \cite{YY,BCS} (upper bound) (see also \cite{HHNST-23,HHST-24,FJGMOT} for extensions to the free energy at positive temperature). In the analysis of the Bose gas, an important role is played by Bogoliubov's theory  \cite{Bogoliubov-47}, where the correlation between particles is captured by unitary transformations containing quadratic expressions in the creation and annihilation operators on Fock space. The solution of the zero-energy scattering equation \eqref{eq: zero energy scatt eq-intro} naturally enters in these transformations.

Our proof of Theorem \ref{thm: main up bd} is based on an adaptation of the bosonic Bogoliubov theory to Fermi systems. We will interpret suitable pairs of fermions as bosons, and then construct a trial state using the corresponding quasi-bosonic Bogoliubov transformations. This bosonization method goes back to  heuristic arguments in \cite{Sawada-57,SawBruFukBro-57} in the attempt to explain the random phase approximation proposed 
in the 1950s \cite{BohPin-51,BohPin-52,BohPin-53,Pines-53}. In the context of the high-density Fermi gas, the bosonization method has been recently made rigorous to study the correlation energy in the mean-field regime \cite{BNPSS-20,BNPSS-21,BPSS-21,ChrHaiNam-23a,CHN,ChrHaiNam-24}.

For the low-density Fermi gas, the bosonization method was used in \cite{FGHP} to give a new proof of \eqref{eq:LSS},  and was improved further in  \cite{Gia1} to obtain an optimal error estimate $O(\rho^{7/3})$ in the upper bound. However, to obtain the precise Huang--Yang correction of order $\rho^{7/3}$, we need to go substantially beyond the existing analysis. We will in fact introduce two quasi-bosonic Bogoliubov transformations to extract the correlation contribution of the particles. The first Bogoliubov transformation is closely related to the construction in \cite{Gia1}, and involves the solution to the  zero-energy scattering equation in \eqref{eq: zero energy scatt eq-intro}. In the present paper, this transformation carries  an additional momentum cut-off to ensure that all relevant error terms are of order $o(\rho^{7/3})$. The main new tool is the second Bogoliubov transformation, which contains the solution of a modified scattering equation (related to the {Bethe--Goldstone} equation \cite{BG}) taking into account the presence of the Fermi sea at low momentum. The main ingredients of the proof will be explained in the next section. 

Our analysis can be extended in a straightforward way to higher spin fermions or, equivalently, more than two species of spinless fermions. For simplicity, we stick to the case of spin $1/2$ here. We also remark that the method of the present paper can be used to derive a lower bound with the optimal error $O(\rho^{7/3})$; see \cite{Gia3} for details. 

\bigskip

{\bf Note added in proof.} After the completion of the present paper, the matching lower bound was established in \cite{GiaHaiNamSei-25}. Subsequently, the authors of \cite{CWZ-25} studied the free energy of the dilute Fermi gas for temperatures $T<\rho^{2/3+1/6}$ and proved an analogue of the Huang--Yang formula in this case.  Several ingredients of the present work, in particular the analysis of the Bethe--Goldstone equation explaining the Huang--Yang correction and factorization techniques yielding estimates uniform in volume, play an essential role in these subsequent developments.

\bigskip

{\bf Organization of the paper.} In the next section, we shall explain the main structure of the proof. In Sections~\ref{ss:Fock} and~\ref{sec: particle hole} we introduce the fermionic Fock space and the particle-hole transformation, which gives a convenient representation of the correlation Hamiltonian in Proposition~\ref{pro: fermionic transf}. In Section~\ref{sec: almost bos bog transfo} we construct the trial state in terms of two quasi-bosonic unitary transformations $T_1$ and $T_2$, defined in Definition~\ref{def:bosonic-transformations}. A heuristic discussion of the effect of the unitary transformations on the correlation Hamiltonian, and a sketch of the key estimates required, is then given in Section~\ref{sec:sub-strategy}. Section~\ref{sec:prel} contains some preliminary bounds that in particular will allow to reduce the correlation Hamiltonian to a simpler, effective one, as formulated in Proposition~\ref{lem:eff-Hamiltonian}. 
The main technical analysis needed in the proof of Theorem~\ref{thm: main up bd} is done in Sections~\ref{sec:T1} and~\ref{sec: T2}, respectively, where the effect of the two unitary transformations $T_1$ and $T_2$ on the various operators of relevance is investigated. The main results of these two sections are summarized in Propositions~\ref{lem:conjugation-T1} and~\ref{prop:T2}, respectively. Finally, in Section~\ref{sec: conclusions} we combine the results of the previous sections, and complete the proof of Theorem~\ref{thm: main up bd}. The appendix contains various auxiliary results needed in the proof.


\section{Main ingredients in the proof}\label{sec: strategy}

In this section we shall explain in detail the construction of the trial state and the main strategy of the proof of Theorem~\ref{thm: main up bd}.  We will briefly recall the Fock space formalism and the particle-hole transformation,  which will be convenient in order to focus on the excitations around the Fermi sea. Then we will introduce quasi-bosonic transformations, which are our main tool to impose the precise correlations among particles needed to capture the energy to the desired order. Finally we will explain heuristically some key estimates in the computation of the  energy, leading to the upper bound in Theorem \ref{thm: main up bd}.

\subsection{The fermionic Fock space}\label{ss:Fock}
 It will be convenient to work in second quantization. The fermionic Fock space is given by
$$\mathcal{F}_{\mathrm{f}} = \bigoplus_{n\geq 0} \mathcal{F}^{(n)}_{\mathrm{f}},\quad \mathcal{F}^{(n)}_{\mathrm{f}} = \bigwedge^n L^2(\Lambda, \mathbb{C}^2).$$ 
Any $\psi\in\mathcal{F}_{\mathrm{f}}$ is of the form $\psi = (\psi^{(0)}, \psi^{(1)}, \ldots, \psi^{(n)}, \ldots)$ with $\psi^{(n)}= \psi^{(n)}((x_1,\sigma_1), \ldots, (x_n,\sigma_n)) \in \mathcal{F}^{(n)}_{\mathrm{f}}$. The vacuum state will be denoted by $\Omega=(1,0,0,...)$.

For any $f\in L^2(\Lambda, \mathbb{C}^2)$, we denote by $a^\ast(f)$ and $a(f)$ the creation and annihilation operators, respectively. These operators satisfy the canonical commutation relations
\[
  \{a(f), a^\ast(g)\} = \langle f,g\rangle_{L^2(\Lambda;\mathbb{C}^2)}, \quad \{a(f), a(g)\} = \{a^\ast(f), a^\ast (g)\} = 0
\]
for any $f,g\in L^{2}(\Lambda; \mathbb{C}^2)$, where $\{A,B\}=AB+BA$. As a consequence, these operators are bounded, with  
\begin{equation}\label{bound:a}
\|a(f)\| =  \|a^\ast(f)\| = \|f\|_{L^2(\Lambda, \mathbb{C}^2)}.
\end{equation}
For $\sigma \in \{\uparrow, \downarrow\}$ and $k\in \Lambda^*:=(2\pi/L)\mathbb{Z}^3$, we denote $(\delta_\sigma f_k)(x,\sigma') = \delta_{\sigma,\sigma'}{L^{-3/2}} e^{i k\cdot x}$ and 
\begin{equation}\label{eq: def ak a*k}
\hat{a}_{k,\sigma} 
 = a(\delta_{\sigma} f_k) 
 = \frac 1{L^{3/2}} \int_\Lambda dx \, a_{x,\sigma} e^{- i k x}
\end{equation}
where 
we introduced the operator-valued distributions
$a_{x,\sigma} = a(\delta_{x,\sigma})$ with $\delta_{x,\sigma}(y,\sigma^\prime) = \delta_{\sigma,\sigma^\prime}\delta(x-y)$.
Furthermore, we denote by $\mathcal{N}_\sigma$, $\sigma \in\{ \uparrow, \downarrow\}$, the number operators counting the number of particles of spin $\sigma$, which can be written as
\[
 \mathcal{N}_\sigma = \sum_{k\in\Lambda^*} \hat{a}_{k,\sigma}^\ast \hat{a}_{k,\sigma}=  \int_\Lambda dx\, a_{x,\sigma}^\ast {a}_{x,\sigma},
\]
and the total number operator $\mathcal{N}=\sum_{\sigma \in\{\uparrow, \downarrow\}}\mathcal{N}_\sigma$. 


The Hamiltonian $H_N$ in \eqref{eq: hamiltonian} can be extended to the operator on Fock space 
\[
  \mathcal{H} = \sum_{\sigma\in\{\uparrow, \downarrow\}}\int_{\Lambda}dx\, \nabla_x a^\ast_{x,\sigma}\nabla_x a_{x,\sigma} + \frac{1}{2}\sum_{\sigma,\sigma^\prime \in \{\uparrow, \downarrow\}} \int_{\Lambda \times \Lambda} dxdy\, V(x-y) a^\ast_{x,\sigma}a^\ast_{y,\sigma^\prime}a_{y,\sigma^\prime}a_{x,\sigma} .
\]
That is, $H_N$ is the restriction of $\mathcal{H}$ to $\mathcal{F}^{(N)}_{\mathrm{f}} \subset \mathcal{F}_{\mathrm{f}}$. Equivalently, we can write
\[
  \mathcal{H} = \sum_{\sigma\in\{\uparrow, \downarrow\}} \sum_{k\in\Lambda^*} |k|^2 \hat{a}_{k,\sigma}^\ast \hat{a}_{k,\sigma}  + \frac{1}{2 L^3}\sum_{\sigma,\sigma^\prime \in \{\uparrow, \downarrow\}} \sum_{k,p,q\in\Lambda^*} \hat V(k) \hat{a}_{p+k,\sigma}^\ast \hat{a}_{q-k,\sigma'}^\ast  \hat{a}_{q,\sigma'} \hat{a}_{p,\sigma}
 \]
 where $\hat V$ are the Fourier coefficients of $V$. Here and in the following,  we use the convention 
\[
  \hat{f}(p) = \int_{\Lambda} dx\, f(x)e^{-ip\cdot x}\ , \qquad f(x) = \frac{1}{L^3}\sum_{p\in\Lambda^*}\hat{f}(p) e^{ip\cdot x}
\]
for the Fourier transform in a box $\Lambda$. 
(The Fourier transform in $\mathbb{R}^3$ will instead be denoted by $\mathcal{F}$ in the following.)

Finally, we can identity $\mathfrak{h}(N_\uparrow, N_\downarrow)$ with the subspace $\mathcal{F}_{\mathrm{f}}(N_\uparrow, N_\downarrow)\subset \mathcal{F}_{\mathrm{f}}$  where $\mathcal{N}_\sigma = N_\sigma$ for $\sigma\in\{\uparrow,\downarrow\}$ (which is left invariant by $\mathcal{H}$), and write the ground state energy in \eqref{eq: def gs energy} as
\[
  E_{L}(N_\uparrow,N_\downarrow) = \inf_{\psi\in\mathcal{F}_{\mathrm{f}}{(N_\uparrow, N_\downarrow)}}\frac{\langle \psi, \mathcal{H}\psi\rangle}{\langle \psi, \psi \rangle}.
\]

\subsection{The particle-hole transformation}\label{sec: particle hole}


For a non-interacting system, the ground state is given by the free Fermi gas state $\psi_{\mathrm{FFG}}$, which is simply the Slater determinant of the plane waves with momentum inside the Fermi sphere(s), namely 
$$
\psi_{\mathrm{FFG}} = \prod_{\sigma\in \{\uparrow,\downarrow\}} \prod_{k\in \mathcal{B}_F^\sigma} \hat a^*_{k,\sigma} \Omega
$$
where $\mathcal{B}_F^\sigma := \big\{ k\in (2\pi/L)\mathbb{Z}^3\, \,\,\vert\,\,\, |k| \leq k_F^\sigma\big\}$
is the Fermi ball of spin $\sigma$, with  
$k_F^\sigma = (6\pi^2)^{\frac{1}{3}}\rho_\sigma^{\frac{1}{3}} + o(1)_{L\rightarrow \infty}$.
\begin{equation}\label{eq: FFG energy}
  \lim_{L\to \infty}  \frac{E_{\mathrm{FFG}}}{L^3} 
  = \frac{3}{5}(6\pi^2)^{\frac{2}{3}}\left(\rho_\uparrow^{\frac{5}{3}} + \rho_{\downarrow}^{\frac{5}{3}}\right) +\hat{V}(0)\rho_\uparrow\rho_\downarrow + {O}(\rho^{\frac{8}{3}})
\end{equation}
for small $\rho$, where the shorthand notation $\lim_{L\to \infty}$ stands for the thermodynamic limit.


We will denote by $v$ the one-particle  reduced density matrix of $ \psi_{\mathrm{FFG}}$, which is an orthogonal projection with integral kernel given by
%
\begin{equation}\label{eq: def omega}
v_{\sigma,\sigma^\prime}(x;y) = \frac{\delta_{\sigma,\sigma^\prime}}{L^3}\sum_{k\in\mathcal{B}_F^\sigma}e^{ik\cdot (x-y)}.
\end{equation}
Note that $v_{\sigma,\sigma^\prime}(x;y)$ is real-valued. 
We shall also define $u = 1 - v$, with integral kernel
%
\begin{equation} \label{eq: def u,v}
  u_{\sigma, \sigma^\prime}(x;y) = \frac{\delta_{\sigma,\sigma^\prime}}{L^3} \sum_{k\notin\mathcal{B}_F^{\sigma}}e^{ik\cdot(x-y)}, 
\end{equation}
satisfying $u v = v u = 0$. 
We will write $u_\sigma=u_{\sigma,\sigma}$ and $v_\sigma=v_{\sigma,\sigma}$ for simplicity. Furthermore, we define $u_{x,\sigma}(\,\cdot\,) := u_\sigma(\,\cdot\,;x)$ and $v_{x,\sigma}(\,\cdot\,) := v_\sigma(\,\cdot\, ; x)$.

%

As in \cite{FGHP,Gia1} we factor out $E_{\rm FFG}$ by a unitary particle-hole transformation. 

\begin{definition}[Particle-hole transformation]\label{def: ferm bog}
Let $u$,$v: L^2(\Lambda;\mathbb{C}^2)\rightarrow L^2(\Lambda;\mathbb{C}^2)$ be given in \eqref{eq: def omega}--\eqref{eq: def u,v}. The particle-hole transformation $R: \mathcal{F}_{\mathrm{f}}\rightarrow\mathcal{F}_{\mathrm{f}}$ is a unitary operator such that the following properties hold: 
\begin{itemize}
  \item[(i)] The state $R\Omega$ is such that $(R\Omega)^{(n)} = 0$ whenever $n\neq N$ and $(R\Omega)^{(N)} = \Psi_{\mathrm{FFG}}$.
  \item[(ii)] 
  \begin{equation}\label{eq: prop II R}
    R^\ast a_{x,\sigma}^\ast R = a_\sigma^\ast(u_x) + a_\sigma(v_x),
  \end{equation}
  where
  \begin{equation}\label{recall_not}
    a_\sigma^\ast(u_x) = \int_{\Lambda} dy\, u_\sigma(y;x)a_{y,\sigma}^\ast, \qquad a_\sigma(v_x) = \int_{\Lambda} dy\,v_\sigma(y;x) a_{y,\sigma}.
  \end{equation}
\end{itemize}
\end{definition}

\begin{remark}
As explained in \cite{BPS-14}, the particle-hole transformation $R$ is a particular example of a fermionic Bogoliubov transformation, which allows  to focus on the excitations around the free Fermi gas. Equivalently, the transformation $R$ can be defined by 
\begin{equation}
  R^\ast \hat{a}_{k,\sigma} R = \begin{cases} \hat{a}_{k,\sigma} &\mbox{if}\,\,\, k\notin\mathcal{B}_F^\sigma, \\ \hat{a}^\ast_{-k,\sigma} &\mbox{if}\,\,\, k\in \mathcal{B}_F^\sigma, \end{cases}
\end{equation} 
which justifies the name \textit{particle-hole transformation}. 
\end{remark}
Note that we define the particle-hole transformation slightly differently from the one used in the previous work \cite{FGHP,Gia1}, in order to preserve translation-invariance. 
Since $R\Omega = \psi_{\rm FFG}$ we have $E_{\mathrm{FFG}}=  \langle R\Omega, \mathcal{H} R\Omega\rangle $. With the aid of \eqref{eq: prop II R} we can compute $R^* \mathcal{H} R$ explicitly.  Proceeding as in \cite{FGHP} (using \eqref{eq: prop II R} and normal-ordering all the terms), we obtain

%
%
%
%
%
%
%
%
%
\begin{proposition}[Conjugation of $\mathcal{H}$ by $R$]\label{pro: fermionic transf}
  Let $V_\infty$ be as in Assumption \ref{asu: potential V}. Let $\psi\in \mathcal{F}_{\mathrm{f}}$ be a normalized state satisfying $\mathcal{N}_\sigma\psi = N_\sigma\psi$ for $\sigma\in\{\uparrow,\downarrow\}$. 
  Then 
  \begin{equation} 
    \langle \psi, \mathcal{H}\psi\rangle = E_{\mathrm{FFG}} + \langle R^\ast\psi,   \mathcal{H}_{\mathrm{corr}} R^\ast\psi\rangle
      \end{equation}
   where $E_{\mathrm{FFG}}$ is the energy of the free Fermi gas introduced above \eqref{eq: FFG energy} and  $ \mathcal{H}_{\mathrm{corr}}= \mathbb{H}_0 + \mathbb{X} + \sum_{i=1}^4\mathbb{Q}_i$ is the correlation Hamiltonian given by 
%
\begin{eqnarray}\label{eq: def H-corr}
  \mathbb{H}_0 &=& \sum_\sigma\sum_k ||k|^2 -(k_F^\sigma)^2| \hat{a}_{k,\sigma}^\ast \hat{a}_{k,\sigma}, \nonumber\\
  \mathbb{X} &=& \sum_\sigma \int_{\Lambda^2} dxdy\, V(x-y)v_\sigma(x-y)\left(a^\ast_\sigma(u_x)a_\sigma(u_y) -a^\ast_\sigma(v_y)a_\sigma(v_x)\right)\nonumber\\
  \mathbb{Q}_1 &=& \sum_{\sigma,\sigma^\prime}\int_{\Lambda^2} dxdy\, V(x-y) a^\ast_\sigma(u_x)a^\ast_{\sigma}(v_x)a_{\sigma^\prime}(v_y)a_{\sigma^\prime}(u_y)\nonumber
  \\
  && + \frac{1}{2}\sum_{\sigma,\sigma^\prime}\int_{\Lambda^2} dxdy\, V(x-y) \left(a^\ast_\sigma(v_x)a^\ast_{\sigma^\prime}(v_y)a_{\sigma^\prime}(v_y)a_\sigma(v_x) -2a^\ast_\sigma(u_x)a^\ast_{\sigma^\prime}(v_y)a_{\sigma^\prime}(v_y)a_\sigma(u_x)\right)\nonumber,
 \\
  \mathbb{Q}_2 &=& \frac{1}{2}\sum_{\sigma,\sigma^\prime} \int_{\Lambda^2} dxdy\, V(x-y) a^\ast_\sigma(u_x)a^\ast_{\sigma^\prime}(u_y)a^\ast_{\sigma^\prime}(v_y)a^\ast_{\sigma}(v_x) + \mathrm{h.c.}\nonumber
  \\
   \mathbb{Q}_3 &=& \sum_{\sigma,\sigma^\prime} \int_{\Lambda^2} dxdy\, V(x-y)\left(  a^\ast_\sigma(u_x) a^\ast_{\sigma^\prime}(v_y) a^\ast_{\sigma}(v_x)a_{\sigma^\prime}(v_y)- a^\ast_\sigma(u_x) a^\ast_{\sigma^\prime}(u_y) a^\ast_{\sigma}(v_x)a_{\sigma^\prime}(u_y) \right) + \mathrm{h.c.}\nonumber,
  \\
    \mathbb{Q}_4 &=& \frac{1}{2}\sum_{\sigma,\sigma^\prime}\int_{\Lambda^2} dxdy \, V(x-y)  {a}^\ast_\sigma(u_x){a}^\ast_{\sigma^\prime}(u_y){a}_{\sigma^\prime}(u_y){a}_\sigma(u_x) . 
\end{eqnarray}
\end{proposition}
%

In the following, it will be convenient to further decompose $\mathbb{Q}_2 =  \mathbb{Q}_2^\parallel  +  \mathbb{Q}_2^\udarrow$ into a part $\mathbb{Q}_2^\parallel$ involving interactions of particles with the same spin, and $\mathbb{Q}_2^{\udarrow}$ involving interactions of particles of opposite spin, i.e.,
\begin{eqnarray}
 \mathbb{Q}_2^\parallel  &=&  \frac{1}{2}\sum_{\sigma} \int_{\Lambda^2} dxdy\, V(x-y) a^\ast_\sigma(u_x)a^\ast_{\sigma}(u_y)a^\ast_{\sigma}(v_y)a^\ast_{\sigma}(v_x) + \mathrm{h.c.}\nonumber \\
   \mathbb{Q}_2^\udarrow  &=& \frac{1}{2}\sum_{\sigma\neq\sigma^\prime} \int_{\Lambda^2} dxdy\, V(x-y) a^\ast_\sigma(u_x)a^\ast_{\sigma^\prime}(u_y)a^\ast_{\sigma^\prime}(v_y)a^\ast_{\sigma}(v_x) + \mathrm{h.c.} \label{q2ud}
  \end{eqnarray}
For our trial state, the main contribution to the energy will come from the effective correlation Hamiltonian
%
\begin{equation}\label{eq: def eff corr en strategy}
  \mathcal{H}_{\mathrm{corr}}^{\mathrm{eff}} := \mathbb{H}_0 + \mathbb{Q}_2^\udarrow + \mathbb{Q}_4
\end{equation}
(see Section \ref{sec: eff corr en} for more details).

\subsection{Trial state and quasi-bosonic Bogoliubov transformations}\label{sec: almost bos bog transfo}

The trial state we are going to use in order to prove Theorem \ref{thm: main up bd} is of the form 
\begin{equation}\label{eq: def trial state}
 \psi_{\mathrm{trial}}  := RT_1T_2\Omega,
\end{equation}
where $R$ is the particle hole transformation defined in Definition~\ref{def: ferm bog}, $\Omega\in \mathcal{F}_{\mathrm{f}} $ is the vacuum state, and  $T_1$, $T_2$ are certain unitary transformations that will be defined  below. 


The main idea of our approach is to describe the low energy excitations around the Fermi ball by pairs of fermionic particles that display an approximate bosonic behavior. To do that we introduce the quasi-bosonic annihilation operators 
\begin{equation}\label{eq: def bp bpr}
  b_{p,k,\sigma} 
  =\hat{u}_\sigma(p+k)\hat{v}_\sigma(k)\hat{a}_{p+k,\sigma}\hat{a}_{-k,\sigma} \qquad \text{and} \quad  b_{p,\sigma} = \sum_{k\in\Lambda^*}  b_{p,k,\sigma} 
\end{equation}
and their adjoints, the corresponding creation operators, 
where $\hat{u}_\sigma,\hat{v}_\sigma$ are the Fourier coefficients of the kernels introduced in \eqref{eq: def u,v}. In particular, 
\begin{equation}\label{eq: def u hat, v hat}
\hat{u}_\sigma(k) = \begin{cases}
  0 &\mbox{if}\,\,\, |k| \leq k_F^\sigma, \\ 1 &\mbox{if}\,\,\, |k| > k_F^\sigma,
  \end{cases}\qquad \hat{v}_\sigma(k) = \begin{cases} 1 &\mbox{if}\,\,\, |k| \leq k_F^\sigma, \\ 0 &\mbox{if} \,\,\, |k| > k_F^\sigma. \end{cases}
\end{equation}

\begin{remark}
Unlike in \cite{FGHP, Gia1},  we use here the sharp projections $\hat{u},  \hat{v}$ (instead of regularized ones), which will help to simplify and improve many error terms.
\end{remark}
The  unitary transformations $T_1$ and $T_2$ are defined as exponentials of expressions quadratic in the $b,b^\ast$ operators. The two transformations correspond  to two different regimes of the momentum $p$ of the bosonic quasi-particle: for high-momenta (with respect to $k_F^\sigma$) we use a ``less refined'' expression which helps to extract the leading term $8\pi a \rho_\uparrow \rho_\downarrow$ of the interaction energy, and at the same time renormalizes the interaction potential, while for low momenta we use a ``more refined" expression which allows to extract the correct energy of order $\rho^{7/3}$. 
In order to introduce the momentum splitting,  we 
choose two smooth functions $\R^3 \to \R$, denoted  by $\mathcal{F}(\chi_{<,\infty})$ and $\mathcal{F}(\chi_{>,\infty})$, respectively, satisfying
%
\begin{equation}\label{eq: def chi< chi>}
\mathcal{F}(\chi_{<,\infty}) + \mathcal{F}(\chi_{>,\infty})  =1, \quad 
 \mathcal{F}({\chi}_{<,\infty})(p) = \begin{cases} 
                1 &\mbox{if}\,\,\,|p| < 4\rho^{\frac{1}{3} - \gamma}, \\ 
                0 &\mbox{if}\,\,\, |p| > 5 \rho^{\frac{1}{3} - \gamma}.
              \end{cases}
\end{equation}
Here $0<\gamma < 1/3$ is a parameter which will be optimized over at the end; we shall in fact choose $\gamma = 1/9$. From  $\mathcal{F}(\chi_{<,\infty})$ we can construct the periodic function $\chi_{<}:\Lambda\to \mathbb{C}$ using the Fourier coefficients $\{\mathcal{F}(\chi_{<,\infty}(p)\}_{p\in \Lambda^*}$, which satisfies the uniform bound $\|\chi_{<}\|_{L^1(\Lambda)}\le C$ (see Lemma \ref{lem:cut-off-chi}). In the following we will also use the notation $\widehat{\chi}_>(p) = 1-\widehat{\chi}_<(p)$. 

The operators $T_1$ and $T_2$ will implement the desired correlation structure between particles by using two different expressions both related to the zero-energy scattering equation \eqref{eq: zero energy scatt eq-intro}.  


\begin{definition}[Scattering solutions]\label{def:scattering-phi-eta} We  define the periodic function $\varphi:\Lambda\to \mathbb{C}$ by 
$\hat \varphi(p)=\mathcal{F}(\varphi_\infty)(p)$ for $0\ne p\in \Lambda^*$, $\hat{\varphi}(0) = 0$,  
where $\varphi_\infty:\R^3\to \R$ is the solution of the zero-scattering equation \eqref{eq: zero energy scatt eq-intro}. Moreover, for $\varepsilon = \rho^{2/3 +\delta}$ with a parameter $\delta >0$, we  define 
\begin{equation}\label{eq: def eta epsilon}
  \hat{\eta}_{r,r^\prime}^\varepsilon (p) := \frac{8\pi a}{ \lambda_{r,p} +  \lambda_{r',-p} + 2\varepsilon},\quad \lambda_{r,p} :=  |r+p|^2 - |r|^2 
\end{equation}
where $a$ denotes the scattering length of $V_\infty$. 
\end{definition}

\begin{remark}
We will see in Appendix~\ref{app: bounds phi} that $\varphi$  is well-defined. 
In fact $\|\varphi\|_{L^\infty(\Lambda)}\le C$ uniformly in $L$ but  $\|\varphi\|_{L^2(\Lambda)}$ diverges as $L\to \infty$ (see Remark \ref{rm:varphi}). 
\end{remark}

\begin{definition}[Quasi-bosonic transformations]\label{def:bosonic-transformations} For a parameter $\lambda\in \mathbb{R}$, we define
\begin{align}\label{eq: def T1 unitary}
T_{1;\lambda} &=  \exp \left(\lambda(B_1 - B_1^\ast)\right),\quad B_1=  \frac{1}{L^3}\sum_{p\in \Lambda^*}\hat{\varphi}(p)\widehat{\chi}_>(p)\, {b}_{p,\uparrow}{b}_{-p,\downarrow}, \\
\label{eq: def T2 unitary}
T_{2;\lambda} &=  \exp \left(\lambda(B_2 - B_2^\ast)\right),\quad B_2 =  \frac{1}{L^3}\sum_{p,r,r^\prime \in \Lambda^*}\hat{\eta}_{r,r^\prime}^\varepsilon(p)\widehat{\chi}_<(p) b_{p,r,\uparrow}b_{-p,r^\prime,\downarrow},
\end{align}
where $\hat{\varphi}$ and $\hat{\eta}^\varepsilon_{r,r^\prime}$ are given in Definition \ref{def:scattering-phi-eta}. We also write $T_1$ and $T_2$ in place of $T_{1;1}$ and $T_{2;1}$, respectively. 
\end{definition}



The cut-off $\widehat \chi_>(p)$ helps to regularize  $B_1$ since the function $\tilde \varphi:\Lambda\to \R$ with $\widehat {\tilde \varphi}(p)=\hat \varphi (p)\hat \chi_{>}(p)$ is uniformly bounded in $L^1(\Lambda)\cap L^\infty(\Lambda)$ (see Lemma \ref{lem: bounds phi}) as $L\to \infty$. In the expression for $B_2$ in \eqref{eq: def T2 unitary}, we only use $\hat{\eta}_{r,r^\prime}^{\varepsilon}(p)$ for momenta $r\in\mathcal{B}_F^\uparrow$, $r^\prime\in\mathcal{B}_F^\downarrow$, $r+p\notin\mathcal{B}_F^\uparrow$, $r^\prime - p\notin\mathcal{B}_F^\downarrow$, where  $\lambda_{r,p} = |r+p|^2 - |r|^2 >0$ and $\lambda_{r',-p} = |r'-p|^2 - |r'|^2 >0.$ The parameter $\varepsilon = \rho^{2/3 +\delta}$ in the denominator of $\hat{\eta}_{r,r^\prime}^{\varepsilon}(p)$  is introduced for technical reasons to avoid logarithmic divergences otherwise encountered in integrals. At the end we shall choose $\delta$ to satisfy $16/63\leq \delta \leq 8/9$.

%
%
%
%
%

\begin{remark}[Relation to Bethe–Goldstone equation]
From our analysis, it may seem more natural to consider a trial state of the form
$\widetilde{\psi}_{\mathrm{trial}} := RT\Omega$, with the unitary $T$ given by 
\begin{equation}\label{eq: def T}
  T = \exp \left(\frac{1}{L^3}\sum_{p,r,r^\prime} \hat\varphi_{r,r^\prime}(p) b_{p,r,\uparrow} b_{-p,r^\prime, \downarrow} -\mathrm{h.c.}\right), 
\end{equation}
where $\hat{\varphi}_{r,r^\prime}$ is chosen to satisfy the equation
\begin{equation} \label{BetheGold1}
  2 ( |p|^2 + p\cdot (r-r^\prime)) \hat\phi_{r,r^\prime}(p)  = \mathcal{F}(V_\infty (1- \phi_{r,r^\prime}))(p)
\end{equation}
(in the infinity volume limit). Note that \eqref{BetheGold1} reduces to \eqref{eq: zero energy scatt eq-intro} for $r=r'=0$. 
%
%
Equation \eqref{BetheGold1} can be interpreted as describing the scattering (at zero energy) of two particles at relative momentum $p$, with initial momenta $r,r'$ within the Fermi ball which is fully occupied. From the definition of 
$b_{p,r,\uparrow}$ and $b_{-p,r^\prime, \downarrow}$, the function $\varphi_{r,r'}(p)$ appearing in \eqref{eq: def T} is naturally supported on the set where $|r+p|\geq k_F^\uparrow, |r'-p| \geq k_F^\downarrow$. Note that  the effective kinetic energy 
$$
2|p|^2 + 2p\cdot (r-r^\prime) = \lambda_{r,p} + \lambda_{r',-p} 
$$ 
describes the difference between the kinetic energy of the two particles with initial momenta $r,r'$ and excited momenta $r+p, r'-p$, and is  positive. 
Equation \eqref{BetheGold1} plays a prominent role in the physics literature. With the help of 
$ G_{r,r'}(p) :=  (\lambda_{r,p} + \lambda_{r^\prime,-p})\hat{\phi}_{r,r^\prime}(p)$, it 
can be rewritten in the form 
\begin{equation}\label{BetheGold2}
G_{r,r'}(p) = \mathcal{F}(V_\infty)(p) - \int_{|r+q| \geq k_F^\uparrow, |r'-q| \geq k_F^\downarrow} dq\,   \frac {\mathcal{F}(V_\infty)(p-q)}{ \lambda_{r,q} + \lambda_{r',-q}} G_{r,r'}(q) ,
\end{equation}
which is known as the {\em Bethe--Goldstone} equation, see \cite[Equation 2.2]{BG}.

For technical reasons, we find it simpler not to work with 
 $\widetilde{\psi}_{\mathrm{trial}} = RT\Omega$, but rather split the unitary into two parts,  $T_1$ and $T_2$. Note that the $\hat{\varphi}$ in the definition of $T_{1}$ depends only on the momentum $p$; in this sense $T_1$ is ``less refined'' but it is sufficient to renormalize the interaction in $\mathbb{Q}_2^\udarrow$ into a softer one with integral given by $8\pi a$. To be precise, the transformation $T_1$ allows us to replace the original potential $V$ in $\mathbb{Q}_2^\udarrow$ by the softer potential $8\pi a \chi_<$, which is supported in small momenta ($|p|\le 5\rho^{1/3-\gamma}\ll 1$), and consequently the scattering length of $8\pi a\chi_<$ is well approximated by $(8\pi)^{-1}$ times its integral. On the other hand, the transformation $T_2$ is ``more refined" since the $\hat{\eta}^{\varepsilon}_{r,r^\prime}$ in $T_{2}$ depends on $p,r,r^\prime$, but it is also simpler than $T$ since the function $\hat\eta_{r,r'}^\eps$ is explicit.

Note that on the support of $\chi_>(p)$, the term $p \cdot (r-r')$ is subleading compared to $p^2$, hence it is natural to replace ${\phi}_{r,r^\prime}$ by the usual scattering solution $\varphi$ in $T_1$. On the other hand, on the support of $\chi_<$ 
%
and for $r\in \mathcal{B}_F^{\uparrow}, r'\in \mathcal{B}_F^{\downarrow}$, $\mathcal{F}(V_\infty (1- \phi_{r,r^\prime}))(p)$ on the right-hand side of \eqref{BetheGold1} can be replaced by $8\pi a$ to leading order, naturally reducing $\hat{\phi}_{r,r^\prime}$   to $\hat{\eta}^{\varepsilon}_{r,r^\prime}$ in this regime.
%
\end{remark}

\begin{remark}[Configuration space representation]\label{rem: T1 and T2 configuration space}
In the calculation below it will be convenient to estimate the error terms in configuration space. For this reason, we will often write 
\begin{equation}\label{eq: def tilde phi T1}
 B_1 = \int dzdz^\prime\, \tilde \varphi(z-z^\prime)a_\uparrow(u_z)a_\uparrow(v_z)a_\downarrow(u_{z^\prime})a_\downarrow(v_{z^\prime})
\end{equation}
where 
\begin{equation}\label{eq: def tilde phi}
\widehat {\tilde \varphi}(p)=\mathcal{F}(\varphi_\infty)(p)  \widehat{\chi}_>(p) .
\end{equation}
Useful bounds on the function $\widetilde\varphi$ will be given in Lemma \ref{lem: bounds phi}. 

Writing the unitary $T_2$ in configuration space is a bit more complicated, since $\hat{\eta}_{r,r'}^{\varepsilon}(p)$ depends on all the momenta involved in the definition of the transformation. 
First of all we notice that in the definition of $B_2$ in \eqref{eq: def T2 unitary}, the coefficients  $\hat u_\uparrow(r+p)$ and $\hat u_\downarrow(r'-p)$ appearing in \eqref{eq: def bp bpr} can be replaced by 
$\hat{u}_\uparrow^<(r+p)$ and $\hat{u}_\downarrow^<(r^\prime - p)$ with
\begin{equation}\label{eq: u<}
  \hat{u}^<_\sigma(k) := \begin{cases} 1 &\mbox{if}\,\,\, k_F^\sigma < |k| \leq 6\rho^{\frac{1}{3} -\gamma}, 
  \\
  0 &\mbox{if}\,\,\, |k| < k_F^\sigma, \,\,\,\mbox{or}\,\,\,|k| > 6\rho^{\frac{1}{3} -\gamma}, \end{cases}
\end{equation} 
since $|p|\leq 5 \rho^{1/3-\gamma}$ and both $r$ and $r'$ are bounded by (a constant times) $\rho^{1/3}$, hence $|r+p|, |r'-p| < 6 \rho^{1/3-\gamma}$ for $\rho$ small. 
Moreover, since both $\lambda_{r,p}$ and $\lambda_{r',-p}$ appearing in the dominator in 
%
the definition of $\hat{\eta}^\varepsilon_{r,r^\prime}$ in \eqref{eq: def eta epsilon} are positive, we can write
\[
  \hat{\eta}^{\varepsilon}_{r,r^\prime}(p) = 8\pi a \int_0^{\infty} dt\, e^{-(|r+p|^2 - |r|^2 + |r^\prime - p|^2 - |r^\prime|^2 + 2\varepsilon)t}
\]
for $r+p\notin \mathcal{B}_F^{\uparrow}$,  $r^\prime - p\notin \mathcal{B}_F^{\downarrow}$, $r\in \mathcal{B}_F^{\uparrow}$ and $r^{\prime}\in \mathcal{B}_F^{\downarrow}$.
Therefore, introducing the functions
\begin{equation}\label{eq: def ut vt}
  \hat{v}^t_\sigma(\,\cdot\,) := e^{t|\,\cdot\,|^2}\hat{v}_\sigma(\,\cdot\,), \qquad \hat{u}^t_\sigma(\,\cdot\,) := e^{-t|\,\cdot\, |^2}\hat{u}^<_{\sigma}(\,\cdot\,),
\end{equation}
with $\hat{u}_\sigma^<$ and $\hat{v}_\sigma$  defined  in \eqref{eq: u<} and \eqref{eq: def omega}, respectively, we can write $B_2$ in configuration space as 
\begin{equation}\label{def:B2}
  B_2 = 8\pi a\int_0^\infty dt\, e^{-2t\varepsilon}\int dzdz^\prime \chi_<(z-z^\prime) a_\uparrow(u_z^t)a_\uparrow(v_z^t)a_\downarrow(u^t_{z^\prime})a_\downarrow(v^t_{z^\prime}).
\end{equation}
A big advantage of choosing the explicit function \eqref{eq: def eta epsilon} (as opposed to the solution of \eqref{BetheGold1}) is to allow for this representation in configuration space. 
\end{remark}

\begin{remark} The transformation $T_1$ is related to the unitary transformation used in \cite[Definition 4.3]{Gia1} (see also Remark 4.4 and Remark 4.5 in \cite{Gia1}). 
In \cite[Definition 4.3]{Gia1}, $T_1$ contains a cut-off around $|x|\simeq \rho^{-1/3}$ in configuration space, which is essentially equivalent to  a cut-off around  $|p| \simeq \rho^{1/3}$ in momentum space. 
%
In the present paper we choose a weaker localization with a momentum cut-off around $|p| \simeq \rho^{1/3 -\gamma}$ in order to obtain better error estimates of order $o(\rho^{7/3})$. 
\end{remark}

We conclude this section by showing that the trial state introduced in \eqref{eq: def trial state} is admissible, i.e., is a state with $N_\uparrow$ particles with spin $\uparrow$ and $N_\downarrow$ particles with spin $\downarrow$. 
To see this, note that $R\Omega$ is clearly admissible, and $RB_jR$ commutes with $\mathcal{N_\sigma}$ for $j\in\{1,2\}$ and $\sigma\in\{\uparrow,\downarrow\}$, which can be deduced from the definitions in a straightforward way. Hence
\[
  \mathcal{N}_\sigma \psi_{\mathrm{trial}} = \mathcal{N}_\sigma R {T}_1 {T}_2 \Omega = R T_1 T_2 R \mathcal{N}_\sigma R \Omega  = N_\sigma \psi_{\mathrm{trial}}.
\]

\subsection{Key estimates}\label{sec:sub-strategy}

Using the trial state \eqref{eq: def trial state} and Proposition~\ref{pro: fermionic transf}, we obtain as an upper bound to the ground 
state energy
\[
  E_{L}(N_\uparrow, N_\downarrow) \leq E_{\mathrm{FFG}} +  \langle \Omega,T_2^* T_1^* \mathcal{H}_{\mathrm{corr}} T_1T_2\Omega\rangle  \simeq  E_{\mathrm{FFG}} +  \langle \Omega,T_2^* T_1^* \mathcal{H}_{\mathrm{corr}}^{\rm eff} T_1T_2\Omega\rangle
\]
where $\mathcal{H}_{\mathrm{corr}}^{\mathrm{eff}} = \mathbb{H}_0 + \mathbb{Q}_2^\udarrow + \mathbb{Q}_4$ was defined in \eqref{eq: def eff corr en strategy}. 
%
%
%
%
We shall now explain some key estimates in the computation of  $\langle \Omega,T_2^* T_1^* \mathcal{H}_{\mathrm{corr}}^{\rm eff} T_1T_2\Omega\rangle$. In the first step, the unitary operator $T_1$ introduced in \eqref{eq: def T1 unitary} is responsible for extracting the leading order of the interaction energy given by $8\pi a \rho_\uparrow \rho_\downarrow$. Moreover, via $T_1$ we also renormalize the interaction, effectively replacing $\mathbb{Q}_2^\udarrow + \mathbb{Q}_4$ by 
\begin{equation}\label{eq: definition Q2<}
  {\mathbb{Q}}_{2;<} = \frac{8\pi a}{L^3}\sum_{p}\widehat{\chi}_<(p) b_{p,\uparrow}b_{-p,\downarrow} + \mathrm{h.c.}.
\end{equation}
At this point the correlation operator to leading order reduces to $\mathbb{H}_0 + {\mathbb{Q}}_{2;<}.$
To extract the next order correction to the correlation energy we then conjugate $\mathbb{H}_0 + {\mathbb{Q}}_{2;<}$ by the unitary operator $T_2$ in the second step. 
%

\bigskip
{\textbf{Step 1:}} 
We conjugate the effective correlation operator with the unitary $T_1$ by writing
   \begin{eqnarray*}
     T_1^*  \mathcal{H}_{\mathrm{corr}}^{\mathrm{eff}}T_1   &=&  \mathcal{H}^{\mathrm{eff}}_{\mathrm{corr}}  + \int_0^1 d\lambda \, \partial_\lambda \, T_{1;\lambda}^*  \mathcal{H}_{\mathrm{corr}}^{\mathrm{eff}} T_{1;\lambda} 
    \\
    &=&  \mathcal{H}^{\mathrm{eff}}_{\mathrm{corr}}  +  \int_0^1 d\lambda \, T_{1;\lambda}^* [\mathcal{H}_{\mathrm{corr}}^{\mathrm{eff}} , B_1 - B_1^\ast] T_{1;\lambda},
   \end{eqnarray*}
   recalling the notation $T_{1;\lambda} = \exp(\lambda (B_1 -B_1^\ast))$ introduced in \eqref{eq: def T1 unitary}. To leading order, we shall see that
    \[
    [\mathbb{H}_0 + \mathbb{Q}_4, B_1 - B_1^\ast] \simeq - \frac{1}{L^3}\sum_{p}\big(2|p|^2\hat{\varphi}(p)\widehat{\chi}_>(p) + \hat{V}\ast_{\sum} (\widehat{\varphi}\widehat{\chi}_>)(p) \big) b_{p,\uparrow} b_{-p,\downarrow} + \mathrm{h.c.}
  \]
   where $\ast_{_{\sum}}$ denotes  the ``discrete convolution'' for functions defined on  momentum space, namely  
   \[
  (\hat{f}\ast_{_{\sum}}\hat{g})(p)  = \frac{1}{L^3}\sum_{q\in\Lambda^*} \hat{f}(p-q)\hat{g}(q) = \widehat {(fg)}(p). 
\]
This explains the choice of the function $\hat{\varphi}\hat \chi_{>}$ appearing in  $B_1$: 
%
%
%
From the zero-energy scattering equation \eqref{eq: zero energy scatt eq-intro} we deduce  that $\hat{V}\ast_{\sum} \widehat{\varphi}(p) \simeq \hat{V}(p)  -  2|p|^2\hat{\varphi}(p)$. 
  The operator $\mathbb{Q}_2^\udarrow$ in \eqref{eq: def H-corr} can equivalently be written as
  \[
  \mathbb{Q}_2^\udarrow = \frac{1}{L^3}\sum_{p}\hat{V}(p)b_{p,\uparrow}b_{-p,\downarrow} + \mathrm{h.c}.
\]
 Since $2|p|^2 \hat\varphi(p) \simeq 8\pi a$ for small $p$,   
  we conclude that 
  \begin{equation}\label{eq: scatt eq Q2< + tilde Q2<}
  \mathbb{Q}_2^\udarrow + [\mathbb{H}_0 + \mathbb{Q}_4, B_1 - B_1^\ast] \simeq  {\mathbb{Q}}_{2;<}  + \widetilde{\mathbb{Q}}_{2;<}
  \end{equation}
   where ${\mathbb{Q}}_{2;<}$ is defined  in \eqref{eq: definition Q2<} and 
\begin{equation}\label{def:q2tl}
    \widetilde{\mathbb{Q}}_{2;<} = \frac{1}{L^3}\sum_{p}\hat{V}\ast_{_{\sum}} (\hat\varphi \widehat{\chi}_<)(p)b_{p,\uparrow}b_{-p,\downarrow} + \mathrm{h.c.}.
\end{equation}

%

We  emphasize that $\mathbb{Q}_{2;<}$ is a renormalized version of $\mathbb{Q}_2^\udarrow$, where after taking the 2-body scattering process into account, we have the softer interaction potential $8\pi a \widehat{\chi}_<$ instead of $\hat{V}$. The operator $\widetilde{\mathbb{Q}}_{2;<}$ will be important only in the next order of the Duhamel expansion, where it leads via a commutator to a term relevant for the correct Huang--Yang correction of order $\rho^{7/3}$, as will be demonstrated below. In fact, applying \eqref{eq: scatt eq Q2< + tilde Q2<} we can extend the above Duhamel expansion as
  \begin{align}\label{eq: conjugation under T1}
  &   T_1^* (\mathbb{H}_0 +  \mathbb{Q}_4 + \mathbb{Q}_2^\udarrow)T_1 =  \mathbb{H}_0 + \mathbb{Q}_4  + \mathbb{Q}_2^\udarrow+ [\mathbb{H}_0 + \mathbb{Q}_4, B_1 - B_1^\ast]     \nonumber \\
    & + \int_0^1d\lambda \,  (1-\lambda) T_{1;\lambda}^*  [[\mathbb{H}_0 + \mathbb{Q}_4, B_1 -B_1^\ast],B_1 - B_1^\ast] T_{1;\lambda} \nonumber
    + \int_0^1 d\lambda\,  T_{1;\lambda}^*  [\mathbb{Q}_2^\udarrow,B_1 - B_1^\ast] T_{1;\lambda} \\
    &\simeq  \mathbb{H}_0 + \mathbb{Q}_4  + \mathbb{Q}_{2;<} + \widetilde{\mathbb{Q}}_{2;<} + \int_0^1d\lambda \,  T_{1;\lambda}^* [   \lambda \mathbb{Q}_2^\udarrow + (1-\lambda)\big(  {\mathbb{Q}}_{2;<} + \widetilde{\mathbb{Q}}_{2;<}\big)  ,B_1 -B_1^\ast] T_{1;\lambda}  .
  \end{align}
  Computing the commutator in the last integrand and normal-ordering the expression leads to constant term, which is the leading contribution. The integration over $\lambda$ then just gives a factor $1/2$. We shall see that 
  $$
  \frac 12  [  \mathbb{Q}_2^\udarrow +  {\mathbb{Q}}_{2;<} + \widetilde{\mathbb{Q}}_{2;<}  ,B_1 -B_1^\ast] \simeq
  -\rho_\uparrow\rho_\downarrow\sum_{p\in\Lambda^*}\hat{V}(p)\hat{\varphi}(p)  + (8\pi a)^2\rho_\uparrow \rho_\downarrow \sum_{0\ne p \in\Lambda^*} \frac{\widehat{\chi}_<(p)}{2|p|^2}.
  $$
 The first term on the right-hand side naturally combines with $\hat V(0) \rho_\uparrow \rho_\downarrow$ from \eqref{eq: FFG energy} to yield the scattering length, since 
 \[
  \hat{V}(0) -\frac{1}{L^3}\sum_{p\in\Lambda^*} \hat{V}(p)\hat{\varphi}(p) =  (\hat{V} - \hat{V}\ast_{\sum} \hat{\varphi})(0)  \simeq \int_{\R^3} V_\infty (1-\varphi_\infty) = 8\pi a.
\]
Therefore, we obtain the upper bound 
\begin{multline}\label{eq: second step app part I}
   \frac{E_L(N_\uparrow, N_\downarrow)}{L^3}\leq \frac{3}{5}(6\pi^2)^{\frac{2}{3}} \left(\rho_\uparrow^{\frac{5}{3}} + \rho_\downarrow^{\frac{5}{3}}\right) + 8\pi a \rho_\uparrow \rho_\downarrow +  \frac{(8\pi a)^2\rho_\uparrow\rho_\downarrow}{L^3}\sum_{0\neq p \in\Lambda^*} \frac{(\widehat{\chi}_<(p))^2}{2|p|^2}
  \\
  + \frac{1}{L^3}\langle T_{2}\Omega, (\mathbb{H}_0 + \mathbb{Q}_4 + {\mathbb{Q}}_{2;<}  )T_2\Omega\rangle + o(\rho^{\frac{7}{3}}).
\end{multline}
The rigorous justification of this statement is the content of Proposition \ref{lem:conjugation-T1}.

\bigskip
{\textbf{Step 2:}} To complete the proof, we still have to compute the last term in \eqref{eq: second step app part I}. 
We will see that  $\mathbb{Q}_4$ 
is negligible, i.e., that $L^{-3}\langle T_2\Omega, \mathbb{Q}_4  T_2\Omega\rangle = o(\rho^{7/3})$ (see Proposition \ref{pro: prop est Q4 T2}). 
It therefore remains  to conjugate the final effective correlation Hamiltonian $\mathbb{H}_0 + {\mathbb{Q}}_{2;<}$ under $T_2$.
     As above we can write
\begin{equation}\label{eq: conj T2 explanation}
       T_2^* (\mathbb{H}_0 + {\mathbb{Q}}_{2;<}) T_2  
      = \mathbb{H}_0 + T_2^*  {\mathbb{Q}}_{2;<}T_2 + \int_0^1 d\lambda\,  T_{2;\lambda}^*  [\mathbb{H}_0, B_2 - B_2^\ast] T_{2;\lambda} 
\end{equation}
where $T_{2;\lambda} = \exp(\lambda (B_2 - B_2^\ast))$.
From the definition of the unitary operator $T_2$ (see \eqref{eq: def T2 unitary} and \eqref{eq: def eta epsilon}), we find that 
\begin{equation}\label{tov}
  [\mathbb{H}_0, B_2 - B_2^\ast] \simeq -{\mathbb{Q}}_{2;<}.
\end{equation}
As a consequence, inserting  \eqref{tov} in \eqref{eq: conj T2 explanation}  
we obtain
$$
  T_2^*  (\mathbb{H}_0 + {\mathbb{Q}}_{2;<}) T_2   \simeq   \mathbb{H}_0 + T_2^*  {\mathbb{Q}}_{2;<}T_2  - \int_0^1 d\lambda\,  T_{2;\lambda}^* {\mathbb{Q}}_{2;<} T_{2;\lambda} 
  =  \mathbb{H}_0 + \int_0^1 \lambda d\lambda \,  \, T_{2;\lambda}^*  [{\mathbb{Q}}_{2;<}, B_2 - B_2^\ast] T_{2;\lambda} .
$$
Computing the commutator and normal-ordering leads to a constant contribution in the last expression, which is the leading term and given by
\begin{align}\label{eq: secon step appr part II}
 \frac 12  [{\mathbb{Q}}_{2;<}, B_2 - B_2^\ast] 
    \simeq - \frac{(8\pi a)^2}{L^6}\hspace{-0.3cm}\sum_{p,r,r^\prime \in\Lambda^*}  \frac{\hat{u}_\uparrow(r+p) \hat{u}_\downarrow(r^\prime - p) \hat{v}_\uparrow(r) \hat{v}_\downarrow(r^\prime)}{|r+p|^2 - |r|^2 + |r^\prime - p|^2 - |r^\prime|^2 + 2\varepsilon}   (\widehat{\chi}_<(p))^2,
\end{align}
where $ \hat{u}_\sigma$ and $ \hat{v}_\sigma$ are given in \eqref{eq: def u hat, v hat}.
The rigorous justification of this formula  is  contained in Proposition \ref{pro: T2 Q2<}.
Combining both transformations we find that the energy per unit volume satisfies
\begin{multline*}
  \frac{E_{L}(N_\uparrow, N_\downarrow)}{L^3}\leq \frac{3}{5}(6\pi^2)^{\frac{2}{3}} \left(\rho_\uparrow^{\frac{5}{3}} + \rho_\downarrow^{\frac{5}{3}}\right) + 8\pi a\rho_\uparrow\rho_\downarrow 
  \\
  - \frac{(8\pi a)^2}{L^{9}}\sum_{p,r,r^\prime \in\Lambda^*} \left(  \frac{\hat{u}_\uparrow(r+p) \hat{u}_\downarrow(r^\prime - p) \hat{v}_\uparrow(r) \hat{v}_\downarrow(r^\prime)}{|r+p|^2 - |r|^2 + |r^\prime - p|^2 - |r^\prime|^2 +2\varepsilon} - \frac{1}{2|p|^2}\right) (\widehat{\chi}_<(p))^2
  + o(\rho^{\frac{7}{3}}).
\end{multline*}
To conclude we can remove the $\varepsilon$ and the momentum cut-off in the third term (see Lemma \ref{lem: gap constant}) to get, in the thermodynamic limit, 
\begin{multline*}
  e(\rho_\uparrow, \rho_\downarrow)  \leq \frac{3}{5}(6\pi^2)^{\frac{2}{3}} \left(\rho_\uparrow^{\frac{5}{3}} + \rho_\downarrow^{\frac{5}{3}}\right) + 8\pi a\rho_\uparrow\rho_\downarrow 
  \\
 - \frac{(8\pi a)^2}{(2\pi)^9}\int_{\mathbb{R}^3} dp \int_{\mathcal{B}_F^\uparrow}dr \int_{\mathcal{B}_F^\downarrow}dr^\prime \left(\frac{\chi_{\mathcal{B}_F^\uparrow}^c(r+p) \chi_{\mathcal{B}_F^\downarrow}^c(r^\prime - p)}{|r+p|^2 - |r|^2 + |r^\prime -p|^2 - |r^\prime|^2} -\frac{1}{2|p|^2}\right)
 + o(\rho^{\frac{7}{3}}).
\end{multline*}
The evaluation of the last integral yields  the expression $a^2 \rho_\uparrow^{\frac{7}{3}} F({\rho_\downarrow}/{\rho_\uparrow})$ in Theorem \ref{thm: main up bd}, as will be demonstrated in Appendix \ref{app: function g}.

\section{Preliminary bounds and effective correlation operator}\label{sec:prel}
In this section, we collect some useful preliminary bounds that we will use frequently in our analysis. We  also identify the effective correlation operator from which we will extract the desired expression for the energy in the upper bound in Theorem \ref{thm: main up bd}. As usual, we will denote by $C$ generic constants in the following, which may have different values at different occurrences. 

\subsection{Some operator bounds}\label{sec: operator bounds}
In this subsection we collect some general operator inequalities. 
For $g\in L^1(\Lambda)$, define 
\begin{equation}\label{eq: def b phi}
  b_\sigma(g) := \int_{\Lambda} dz \, g(z) a_{\sigma}(u_{z})a_{\sigma}(v_{z}) 
\end{equation}
where we recall the notation \eqref{recall_not}. 
Similarly, we define 
\begin{equation}\label{eq: def b tilde}
  {b}_{j, \sigma}(g) = \int_{\Lambda} dz\, g(z) a_{\sigma}(u_{z})a_{\sigma}(\partial_j v_{z}), 
  \qquad j = 1,2,3.
\end{equation}
\begin{lemma}[Bounds for $b_\sigma(\widetilde\varphi_z)$, $b_{j,\sigma}(\widetilde\varphi_z)$]\label{lem: bounds b b*}
For $\widetilde{\varphi}$ defined in \eqref{eq: def tilde phi} let $\widetilde \varphi_z(z') = \varphi(z-z')$. 
Then
\begin{equation}\label{eq: est b phi, b phi j}
  \|b_\sigma(\widetilde{\varphi}_{z})\|\leq C\rho^{\frac{1}{3} +\frac{\gamma}{2}},\qquad \|{b}_{j,\sigma}(\widetilde{\varphi}_z)\| \leq C\rho^{\frac{2}{3} +\frac{\gamma}{2}}, \qquad j=1,2,3
\end{equation}
uniformly in $z\in \Lambda$. 
\end{lemma}
\begin{proof}
 The proof is analogous to the one of \cite[Lemma 4.8]{Gia1}. 
 We have 
  \begin{eqnarray}\label{eq: b phi z momentum space}
    b_\sigma(\widetilde{\varphi}_z) &=& \int_{\Lambda}dz^\prime \,\widetilde{\varphi}(z-z^\prime)a_\sigma(u_{z^\prime})a_\sigma(v_{z^\prime}) = \frac{1}{L^3}\sum_{p,k}\hat{\varphi}(p)\widehat\chi_>(p) e^{ip\cdot z}\hat{u}_\sigma(k+p)\hat{v}_\sigma(k)\hat{a}_{k+p,\sigma}\hat{a}_{-k,\sigma}\nonumber
    \\
    &=& \frac{1}{L^3}\sum_{p,k}\hat{\varphi}(p)\widehat\chi_>(p) e^{ip\cdot z}\hat{v}_\sigma(k)\mathbbm{1}_{\{|k+p|\geq 3\rho^{1/3 - \gamma}\}}\hat{a}_{k+p,\sigma}\hat{a}_{-k,\sigma},
  \end{eqnarray}
where we used that the conditions $|p| \geq 4\rho^{1/3 -\gamma}$ and $|k| < k_{F}^\sigma$ imply that $|k+p|\geq 3\rho^{1/3 - \gamma}$ for small $\rho$, and hence $\hat{u}_\sigma(k+p) = 1$. 
We  multiply and divide  the expression on the right hand side of \eqref{eq: b phi z momentum space} by $|k+p|^2$ and decompose $b_\sigma(\widetilde{\varphi}_z)$ as a sum of three terms. Writing them in configuration space, we obtain
  \begin{eqnarray*}
    b_\sigma(\widetilde{\varphi}_z) &=& -\int dz^\prime\, \Delta\widetilde{\varphi}(z-z^\prime)a_\sigma(\widetilde{u}_{z^\prime}) a_\sigma(v_{z^\prime}) + 2\sum_{j=1}^3\int dz^\prime\, \partial_j\widetilde{\varphi}(z-z^\prime)a_\sigma(\widetilde{u}_{z^\prime})a_\sigma(\partial_jv_{z^\prime})
    \\
    && - \int dz^\prime\, \widetilde{\varphi}(z-z^\prime) a_\sigma(\widetilde{u}_{z^\prime})a_\sigma(\Delta v_{z^\prime}),
  \end{eqnarray*}
where 
\begin{equation}\label{eq: def tilde u}
  \widetilde{u}_{x}(y) :=\frac{1}{L^3}\sum_{p\in\Lambda^*}\frac{1}{|p|^2}\mathbbm{1}_{\{|p| \geq 3\rho^{1/3-\gamma}\}}e^{ip\cdot (y-x)}.
\end{equation}
Using \eqref{bound:a} we thus obtain
\[
  \|b_\sigma(\widetilde{\varphi}_z)\| \leq  \|{v}_\sigma\|_2 \|{\widetilde{u}}\|_2 \|\Delta\widetilde{\varphi}\|_1 +  2 \|\nabla v_\sigma\|_2 \|{\widetilde{u}}\|_2\|\nabla\widetilde{\varphi}\|_1 +  \|\Delta{v}_\sigma\|_2 \|{\widetilde{u}}\|_2\|\widetilde{\varphi}\|_1 \leq C\rho^{\frac{1}{3} +\frac{\gamma}{2}}, 
\]
where in last step we used \eqref{eq: L1 norms phi} to bound the various norms involving $\widetilde\varphi$, as well as $\|\widetilde u\|_2 \leq C \rho^{\gamma/2-1/6}$ and $\|\partial^n v_\sigma\|_2 \leq C \rho^{1/2 + n/3}$.


For the proof of the second estimate in \eqref{eq: est b phi, b phi j} we can proceed in the same way. The extra factor $\rho^{1/3}$ arises from the additional derivative hitting $v_\sigma$.  
\end{proof}


\begin{lemma}\label{lem: b* phi tilde vs b phi tilde}
For any $g\in L^1(\Lambda)\cap L^2(\Lambda)$, we have 
\begin{equation}\label{eq: bound b ast g} 
  \|b_\sigma^\ast(g)\psi\|^2\leq \|b_\sigma(g)\psi\|^2 + C\rho \|g\|_2^2 \qquad \forall \psi \in\mathcal{F},\ \|\psi\| =1.
\end{equation}
\end{lemma}
\begin{proof}
The proof of \eqref{eq: bound b ast g} can be done as in \cite[Proposition 5.1, Eq. (5.10)]{FGHP}. The idea is to write 
 \[
    \|b^\ast_\sigma(g)\psi\|^2 = \|b_\sigma(g)\psi\|^2 - \langle \psi, [b^\ast_\sigma(g), b_\sigma(g)]\psi\rangle,
 \] 
 and to use the positivity of the (normal-ordered) quadratic terms in the commutator above, obtaining 
 \[
    -  [b^\ast_\sigma(g_z), b_\sigma(g_z)] \leq \frac{1}{L^6}\sum_{k,r} |\hat{g}(k)|^2 \hat{v}_\sigma(r) \hat{u}_\sigma(r+k).
 \]
Bounding $\hat{u}_\sigma$ by $1$, the result follows.
%
 \end{proof}


Next, we turn to the number operator $\mathcal{N}$ on Fock space, which, after the particle-hole transformation $R$, measures the number of excitations above the Fermi sea. We have the following estimates concerning its behavior  under the quasi-bosonic transformations $T_{1;\lambda}$ and $T_{2;\lambda}$ defined in Section \ref{sec: almost bos bog transfo}. 
 
\begin{proposition}[Estimate for $\mathcal{N}$ -- Part I]\label{pro: est N T1} Let $\lambda \in [0,1]$ and $0<\gamma< 1/3$. For any normalized $\psi\in\mathcal{F}_{\mathrm{f}}$, we have 
\begin{equation}\label{eq: est number operator T1}
  \langle T_{1;\lambda}\psi, \mathcal{N} T_{1;\lambda} \psi\rangle \leq C\langle \psi, \mathcal{N}\psi\rangle + CL^3\rho^{\frac{5}{3} + \gamma}.
\end{equation}
\end{proposition}
\begin{proof}
The proof is very similar to the one in \cite[Proposition 4.11]{Gia1}. We compute 
\begin{eqnarray*}
  \partial_\lambda  T_{1;\lambda}^\ast\mathcal{N} T_{1;\lambda}  &=&  T_{1;\lambda}^\ast [\mathcal{N}, B_1 - B_1^\ast] T_{1;\lambda} = - 4 T_{1;\lambda}^\ast \left(  B_1 + B_1^\ast \right) T_{1;\lambda}
  \\
  &=& 4  T_{1;\lambda}^\ast \left(\int dz\, b_\downarrow(\widetilde\varphi_z) a_\uparrow(v_z)a_\uparrow(u_z)+ \mathrm{h.c.}\right) T_{1;\lambda},
\end{eqnarray*}
where we  used the definition \ref{eq: def b phi} and the form of $B_1$ in \eqref{eq: def tilde phi T1}. 
Using that $\|a_\uparrow(v_z)\|= \|v_\uparrow\|_2 \leq \rho^{1/2}$ as well as the bound \eqref{eq: est b phi, b phi j} from Lemma \ref{lem: bounds b b*}, we obtain with the help of the Cauchy--Schwarz inequality 
$$
  |\partial_\lambda \langle T_{1;\lambda}\psi, \mathcal{N}T_{1;\lambda}\psi\rangle | \leq C\rho^{\frac{5}{6}  +\frac{\gamma}{2}} \int dz\, \|a_\uparrow(u_z)T_{1;\lambda}\psi\| \leq CL^{\frac{3}{2}}\rho^{\frac{5}{6} +\frac{\gamma}{2}}\|\mathcal{N}^{\frac{1}{2}}T_{1;\lambda}\psi\|.
$$
The bound \eqref{eq: est number operator T1} then follows from Gr\"onwall's Lemma.
\end{proof}

Before considering the analogous problem for $T_{2;\lambda}$ we shall prove the following Lemma. 
\begin{lemma}[Integration over $t$]\label{lem:t}
Let $0<\gamma<1/3$ and $0<\delta\leq  8\gamma$.  
For $v^t_\sigma$ and $u^t_\sigma$ defined in \eqref{eq: def ut vt}, we have
\begin{equation}\label{fac1}
 \int_0^\infty dt \, e^{-2t\varepsilon} e^{2t (k_F^\sigma)^2} \|{u}^{t}_\sigma\|^2_2 \leq C \rho^{\frac 13-\gamma}
\end{equation}
and
\begin{equation}\label{fac2}
\int_0^\infty dt \, e^{-2t\varepsilon} e^{-2t (k_F^\sigma)^2} \|{v}^{t}_\sigma\|^2_2  \leq C\rho^{\frac{1}{3} - \frac{\delta}{8}}.
\end{equation}
\end{lemma}

The proof shows that the factor $8$ can be replaced by any other number. The choice $8$ guarantees that the final error bound is independent of it. 
In the following we will often use 
\begin{equation}\label{eq: est int t uv final}
  \int dt \, e^{-2t\varepsilon} \|{u}^{t}_\sigma\|_2 \|{v}^t_\sigma\|_2  \leq   C\rho^{\frac{1}{3} - \frac{\gamma}{2} - \frac{\delta}{16}}
\end{equation}
which follows from Lemma~\ref{lem:t} by the Cauchy--Schwarz inequality.

\begin{proof} 
We start with proving \eqref{fac1}. Writing the $L^2$-norm as a sum in momentum space, we have
$$
 \int_0^\infty dt \, e^{-2t\varepsilon} e^{2t (k_F^\sigma)^2} \|{u}^{t}_\sigma\|^2_2 = \frac{1}{2 L^3}\sum_k  \frac{ |\hat{u}_\sigma^<(k)|^2}{  |k|^2 - (k_F^\sigma)^2 + \varepsilon}.
$$
Since $\hat u_\sigma^<(k)$ is supported on $k_F^\sigma < |k| \leq   6\rho^{\frac{1}{3} -\gamma}$ (see \eqref{eq: u<}), the contribution to the sum from $|k| \geq 2 k_F^\sigma$ is bounded by $C \rho^{1/3-\gamma}$. For the remaining part, we bound the sum by the corresponding integral and obtain 
\begin{eqnarray*}
  \frac{1}{2L^3}\sum_{k_F^\uparrow < |k| < 2k_F^\sigma}  \frac{1}{(|k| - k_F^\sigma)(|k| + k_F^\sigma) + \varepsilon} &\leq& C\int_{{k_F^\sigma} \leq |k| \leq {2k_F^\sigma}} dk\frac{1}{(|k| - k_F^\sigma) (|k| +k_F^\sigma) + \varepsilon}
  \\
  &\leq& Ck_F^\sigma \log \left(1 + \frac{(k_F^\sigma)^2}{\varepsilon}\right) .
\end{eqnarray*}
The last factor grows only logarithmically with $(k_F^\sigma)^2/\varepsilon \lesssim \rho^{-\delta}$, hence we can bound it by $C\rho^{-\delta/8}$. This term is negligible compared to the one of order $\rho^{1/3-\gamma}$ above, and we arrive at \eqref{fac1}.
For \eqref{fac2} we have analogously 
$$
 \int_0^\infty dt \, e^{-2t\varepsilon} e^{-2t (k_F^\uparrow)^2} \|{v}^{t}_\uparrow\|^2_2 = \frac{1}{2 L^3}\sum_k  \frac{ |\hat{v}_\uparrow(k)|^2}{ (k_F^\uparrow)^2  - |k|^2  + \varepsilon}.
$$
The sum can be bounded very similarly as above, with the result that
$$
  \frac{1}{2L^3}\sum_k \frac{|\hat{v}_\sigma(k)|^2}{(k_F^\sigma)^2 - |k|^2 +\varepsilon} \leq C \int_{0 \leq |k| \leq {k_F^\sigma}} d k \frac{1}{(k_F^\sigma - |k|)(k_F^\sigma + |k|) + \varepsilon}
  \leq Ck_F^\sigma \log \left(1 + \frac{(k_F^\sigma)^2}{\varepsilon}\right)\leq C\rho^{\frac{1}{3} - \frac{\delta}{8}}.
$$
\end{proof}

\begin{proposition}[Estimate for $\mathcal{N}$ -- Part II]\label{pro: est number op}
 Let $0<\delta\leq 8\gamma$ for some $0< \gamma <1/3$, and $\lambda\in [0,1]$. For any normalized $\psi \in \mathcal{F}_{\mathrm{f}}$, we have 
\begin{equation}\label{eq: est N T2psi}
 |\partial_\lambda \langle T_{2;\lambda}\psi, \mathcal{N} T_{2;\lambda}\psi\rangle| \leq 
C L^{\frac{3}{2}} \rho^{\frac{5}{6} - \frac{\gamma}{2} - \frac{\delta}{16}} \|\mathcal{N}^{\frac{1}{2}}T_{2;\lambda}\psi\|.
\end{equation}
Furthermore, 
\begin{equation}\label{eq: est number operator T1 final}
  \langle T_{1;\lambda}T_2\Omega, \mathcal{N} T_{1;\lambda}T_2\Omega\rangle \leq  C L^3\rho^{\frac{5}{3} -\gamma - \frac{\delta}{8}}.
\end{equation}
\end{proposition}

\begin{proof}
As in Proposition \ref{pro: est N T1}, we compute
$$
    \partial_\lambda T_{2;\lambda}^\ast\mathcal{N} T_{2,\lambda} = T_{2;\lambda}^\ast [\mathcal{N}, B_2 - B_2^\ast] T_{2,\lambda} = - 4 T_{2;\lambda}^\ast \left(  B_2 + B_2^\ast\right) T_{2,\lambda}.
$$
In particular, inserting the expression \eqref{def:B2} for $B_2$, we have
  \[
    \partial_\lambda T_{2;\lambda}^\ast \mathcal{N}T_{2;\lambda} = 4 T^\ast_{2;\lambda}\left(8\pi a\int_0^\infty dt\, e^{-2t\varepsilon}\int dxdy\, \chi_<(x-y)  a_\uparrow(u^{t}_x)a_\uparrow(v^t_x)a_\downarrow(v^{t}_y)a_\downarrow(u^t_y) + \mathrm{h.c.}\right)T_{2;\lambda},
  \]              
  where ${u}^{t}_\sigma$ and ${v}^{t}_\sigma$ are defined in \eqref{eq: def ut vt}.
  We can then estimate
  \begin{eqnarray}\label{eq: est der N T2}
    |\partial_\lambda \langle T_{2;\lambda}\psi, \mathcal{N}T_{2;\lambda}\psi\rangle|  &\leq& C\int dt\, e^{-2t\varepsilon}\|{u}^{t}_\uparrow\|_2\|{v}^t_\uparrow\|_2\|{v}^t_\downarrow\|_2\int dxdy\, |\chi_<(x-y)|
    \|a_\downarrow(u^{t}_y)T_{2;\lambda}\psi\|\nonumber
    \\
    &\leq& C L^{\frac{3}{2}}\rho^{\frac{1}{2}} \|\mathcal{N}^{\frac{1}{2}}T_{2;\lambda}\psi\| \int dt\, e^{-2t\varepsilon} \|\hat{u}^{t}_\uparrow\|_2 \|\hat{v}^t_\uparrow\|_2 ,
  \end{eqnarray}
  where we used that $\|\chi_<\|_1 \leq C$, $\|v^t_\downarrow\|_2 \leq C \rho^{1/2}  e^{t (k_F^\downarrow)^2}$,  and that 
\begin{equation}\label{eq: from aut to Number op}
 \int dy\, \|a_\sigma(u^{t}_y) \psi\|^2 = \sum_{k}|\hat{u}_\sigma^t(k)|^2 \langle \psi, \hat{a}_{k,\sigma}^\ast \hat{a}_{k,\sigma} \psi\rangle \leq \sum_{|k| > k_F^\sigma} e^{-2t|k|^2}\langle  \psi, \hat{a}_{k,\downarrow}^\ast \hat{a}_{k,\downarrow} \psi\rangle \leq e^{-2t(k_F^\sigma)^2}\langle \psi, \mathcal{N}\psi\rangle,
\end{equation}
along with the Cauchy--Schwarz inequality. The desired bound \eqref{eq: est N T2psi} then follows from \eqref{eq: est int t uv final}. 
  
Gr\"onwall's Lemma readily implies \eqref{eq: est number operator T1 final} for $\lambda=0$. The general case then follows directly from  \eqref{eq: est number operator T1}. 
\end{proof}

\subsection{Simplified correlation operator} \label{sec: eff corr en}
In this section, we collect some bounds already discussed in \cite{FGHP, Gia1}, which allow to reduce the correlation operator  
$
  \mathcal{H}_{\mathrm{corr}} = \mathbb{H}_0 + \mathbb{X} + \sum_{i=1}^4\mathbb{Q}_i
$
defined in \eqref{eq: def H-corr} to a simplified one as far as the energy of our trial state $\psi_{\mathrm{trial}} = RT_1T_2\Omega$ is concerned. We shall see that we can ignore the contributions of $\mathbb{X}$, $\mathbb{Q}_1$, $\mathbb{Q}_2^\parallel$ and $\mathbb{Q}_3$.


\subsubsection*{Estimates for $\mathbb{X}$ and $\mathbb{Q}_1$}\label{sec: operators X and Q1}
 Proceeding as in the proof of  \cite[Proposition 3.3]{FGHP}, it is easy to see that 
\begin{equation}\label{eq: est X Q1}
  |\langle \psi, \mathbb{X}\psi\rangle | \leq C\rho\langle \psi, \mathcal{N}\psi\rangle, \qquad |\langle \psi, \mathbb{Q}_1 \psi\rangle |\leq C\rho\langle \psi,\mathcal{N}\psi\rangle, \qquad \forall\psi\in \mathcal{F_{\rm f}}.
\end{equation}
This estimate    guarantees that  $\langle T_1T_2\Omega, (\mathbb{X} + \mathbb{Q}_1) T_1T_2\Omega\rangle = L^3 o(\rho^{7/3})$, as a direct consequence of the bound for the number operator in \eqref{eq: est number operator T1 final}, as long as $\gamma+\delta/8 < 1/3$. 

\subsubsection*{Estimate for $\mathbb{Q}_3$}\label{sec: operator Q3} We note that the state $T_1 T_2\Omega$ is such that in the $n$-particle sector of the Fock space   $(T_1 T_2\Omega)^{(n)} = 0$ unless $n= 4k$ for $k\in\mathbb{N}$. This is a consequence of the fact that both $B_1$ and $B_2$ either create or annihilate $4$ particles. 
Since $\mathbb{Q}_3$ changes the particle number by $\pm 2$, we immediately conclude that
$
  \langle T_1T_2\Omega, \mathbb{Q}_3 T_1T_2\Omega\rangle  = 0
$.

\subsubsection*{Estimate for $\mathbb{Q}_2^\parallel$} 
We shall now argue, similarly as above, that also  $\langle T_{1}T_2\Omega,{\mathbb{Q}}_2^\parallel  T_1T_2\Omega\rangle  = 0$. This follows from the fact that $T_1T_2\Omega$ and ${\mathbb{Q}}_2^{\parallel} T_1T_2\Omega$ are orthogonal, since $\mathbb{Q}_2^\parallel$ creates or annihilates $4$ particles of the same spin, while in the state $T_1T_2\Omega$ the number of particles in each spin component is the same, since both $T_1$ and $T_2$ create particle (and annihilate) pairs of opposite spin simultaneously. 
%
%

\bigskip

In summary, we obtain the following simplification of the correlation Hamiltonian. 
\begin{proposition}[Effective correlation operator] \label{lem:eff-Hamiltonian} We have
$$
\langle T_1T_2 \Omega, \mathcal{H}_{\mathrm{corr}} T_1T_2 \Omega\rangle \le
\langle T_1T_2 \Omega, \mathcal{H}_{\mathrm{corr}}^{\mathrm{eff}}   T_1T_2 \Omega\rangle + C L^3\rho^{\frac{8}{3} -\gamma - \frac{\delta}{8}}
$$
where $\mathcal{H}_{\mathrm{corr}}^{\mathrm{eff}}= \mathbb{H}_0 + \mathbb{Q}_2^\udarrow + \mathbb{Q}_4$ (with the various terms defined in \eqref{eq: def H-corr}).
\end{proposition}

%
%

\section{First quasi-bosonic Bogoliubov transformation $T_1$}\label{sec:T1}
In this section we perform the first step  discussed in Section \ref{sec:sub-strategy}, i.e., we conjugate the effective correlation operator $\mathcal{H}_{\mathrm{corr}}^{\mathrm{eff}}= \mathbb{H}_0 + \mathbb{Q}_2^\udarrow + \mathbb{Q}_4$ in Proposition \ref{lem:eff-Hamiltonian} with the unitary $T_1$ introduced in \eqref{eq: def T1 unitary}. The main result of this section is the following rigorous version of \eqref{eq: second step app part I}. Throughout this  section, we shall denote by $o(1)_{L\to\infty}$ terms that vanish in the thermodynamic limit (and likewise $o(L^\alpha)_{L\to\infty}$ for those that vanish after dividing by $L^\alpha$). We shall not specify their $\rho$-dependence since we take $L\to \infty$ before considering the low density limit.

\begin{proposition}[Conjugation by $T_1$] \label{lem:conjugation-T1} 
For any $\gamma\in (0,\frac 16)$, $\delta \in (0,8\gamma]$ with $2\gamma +  \frac{\delta}{16} \leq \frac{1}{3}$ we have
\begin{multline}
 \frac{1}{L^3} \langle   T_1 T_2 \Omega, \mathcal{H}^{\rm eff}_{\rm corr}  T_1 T_2 \Omega \rangle \leq  
\rho_\uparrow\rho_\downarrow \biggl( 8\pi a - \hat V(0)   + \frac{1}{L^3}\sum_{0\ne p \in\Lambda^*} \frac{(8\pi a)^2(\widehat{\chi}_<(p))^2}{2|p|^2} \biggl)
  \\
  +\frac{1}{L^3}\langle T_2\Omega, (\mathbb{H}_0 + {\mathbb{Q}}_{2;<}+ \mathcal{E}_{\mathbb{H}_0} + C\rho^{\frac 23 + 2\gamma - \frac \delta 8} \mathbb{H}_0+ C\mathbb{Q}_4) T_2\Omega\rangle 
  + C\rho^{\frac{8}{3}  - 2\gamma} + o(1)_{L\to \infty}\label{eq: final T1 up bd}
\end{multline}
with ${\mathbb{Q}}_{2;<}$ defined in \eqref{eq: definition Q2<} and
\begin{equation}
  \mathcal{E}_{\mathbb{H}_0} = -\frac{2}{L^3}\sum_{k,s,s^\prime \in\Lambda^*} k\cdot(s-s^\prime)\hat{\varphi}(k)\widehat{\chi}_>(k)\, b_{k,s,\uparrow}b_{-k,s^\prime, \downarrow} + \mathrm{h.c.}. \label{eq: error kin energy}
\end{equation}
\end{proposition}

The operator $\mathbb{Q}_{2;<}$ can be interpreted as a renormalized version of $\mathbb{Q}_2^\udarrow$, where the effect of the 2-body scattering process have been taken into account, leading to the softer interaction potential $8\pi a \widehat{\chi}_<$ instead of $\hat{V}$. This operator, together with the kinetic energy operator $\mathbb{H}_0$, allows us to extract the full constant term of order $\rho^{7/3}$ after conjugation with $T_2$.  
The error term $\langle T_2\Omega, (\mathcal{E}_{\mathbb{H}_0} + C\rho^{\frac 23+2\gamma-\frac \delta 8} \mathbb{H}_0+ C\mathbb{Q}_4) T_2\Omega\rangle$ will be estimated in the next section.   

To prove Proposition \ref{lem:conjugation-T1}, we first consider the conjugation of the operators $\mathbb{H}_0$, $\mathbb{Q}_4$ and $\mathbb{Q}_2^\udarrow$ with respect to the unitary $T_1$ separately  in Sections~\ref{sec: conj H0 T1}--\ref{sec: conj Q2 T1}. Then in Section \ref{sec: scat eq T1} we use the zero-energy scattering equation \eqref{eq: zero energy scatt eq-intro} to justify \eqref{eq: scatt eq Q2< + tilde Q2<}, effectively leading to ${\mathbb{Q}}_{2;<}$ as a renormalization of $\mathbb{Q}_2^\udarrow$.  In Section \ref{sec: prop est H0 Q4} we prove  propagation bounds 
for the operators $\mathbb{H}_0$ and $\mathbb{Q}_4$, which help to control some of the error terms. Finally, in Section \ref{sec: conj undert T1} we collect all the estimates and complete the proof of Proposition~\ref{lem:conjugation-T1}.

In the following we will often use the expression \eqref{eq: def tilde phi T1} for $T_1$ in configuration space,  together with several estimates on $\tilde \phi$ from Lemma \ref{lem: bounds phi}. 


\subsection{Conjugation of $\mathbb{H}_0$}\label{sec: conj H0 T1}

First, let us consider the conjugation of $\mathbb{H}_0$ by $T_1$. Our goal is to extract a (quasi-bosonic) quadratic term which helps to renormalize the interaction potential in $\mathbb{Q}_2^\udarrow$. 

\begin{proposition}[Conjugation of $\mathbb{H}_0$ by $T_1$]\label{pro: H0} Let $\lambda\in [0,1]$ and  $\gamma\in (0,1/3)$.  Under the same assumptions as in Theorem \ref{thm: main up bd}  
\begin{equation} \label{eq: prop H0}
 \partial_\lambda T_{1;\lambda}^\ast \mathbb{H}_0 T_{1;\lambda} =  T^{\ast}_{1;\lambda}(\mathbb{T}_1 + \mathcal{E}_{\mathbb{H}_0})T_{1;\lambda},
\end{equation}
where $  \mathcal{E}_{\mathbb{H}_0}$ is given in \eqref{eq: error kin energy}, and
\begin{equation}\label{eq: def T1} 
  \mathbb{T}_1 = -\frac{2}{L^3}\sum_{k\in\Lambda^*} |k|^2 \hat{\varphi}(k)\widehat{\chi}_>(k)\,b_{k,\uparrow}b_{-k, \downarrow} + \mathrm{h.c.}.
\end{equation}
%
%
Furthermore, for any $\psi \in \mathcal{F}_{\mathrm{f}}$ 
\begin{equation}\label{eq: est err H0}
  \left| \partial_\lambda \langle T_{1;\lambda}\psi, \mathcal{E}_{\mathbb{H}_0} T_{1;\lambda}\psi\rangle\right| \leq C\rho^{1+\gamma}\|\mathbb{H}_0^{\frac{1}{2}}T_{1;\lambda}\psi\|\|\mathcal{N}^{\frac{1}{2}}T_{1;\lambda}\psi\| + C\rho^{\frac 43+\gamma}\langle T_{1;\lambda}\psi, \mathcal{N} T_{1;\lambda}\psi\rangle .
\end{equation}
\end{proposition}

\begin{remark}\label{rem: Duhamel err kin energy}
  In our analysis, $\mathcal{E}_{\mathbb{H}_0}$ gives rise to an error term. In order to estimate it,  in Section \ref{sec: conj undert T1} (see \eqref{eq: duhamel H0 trial state}--\eqref{eq: est err H0 up bd}),  we shall apply Duhamel's formula again, in the form
  \[
  T^{\ast}_{1;\lambda}\mathcal{E}_{\mathbb{H}_0}T_{1;\lambda} =  \mathcal{E}_{\mathbb{H}_0} + \int_0^\lambda d\lambda^\prime \, \partial_{\lambda^\prime} T^\ast_{1;\lambda^\prime}\mathcal{E}_{\mathbb{H}_0}T_{1;\lambda^\prime},
  \]
and  use \eqref{eq: est err H0} to estimate the last term above.
\end{remark}

\begin{remark}[Comparison with \cite{Gia1}]
The proof of Proposition \ref{pro: H0} is similar to that in \cite[Proposition 5.1 and Lemma 5.2]{Gia1}. However, here we do not need to insert a smooth cut-off in the projection inside or outside the Fermi ball, which simplifies and improves many estimates. 
\end{remark}

\begin{proof}[Proof of Proposition \ref{pro: H0}] 
The equality in \eqref{eq: prop H0} is a straightforward consequence of
\[
  \partial_\lambda T_{1;\lambda}^\ast \mathbb{H}_0 T_{1;\lambda}  =  T_{1;\lambda}^\ast [\mathbb{H}_0 , B_1 - B_1^\ast]T_{1;\lambda} ,
\]
inserting the definitions of $\mathbb{H}_0 $ and $B_1$ in \eqref{eq: def H-corr} and \eqref{eq: def T1 unitary}, respectively, and computing the commutator. We shall skip the details. 
To prove \eqref{eq: est err H0}, we compute 
\[
 \partial_\lambda T_{1;\lambda}^\ast \mathcal{E}_{\mathbb{H}_0} T_{1;\lambda}   = T_{1;\lambda}^\ast  [\mathcal{E}_{\mathbb{H}_0}, B_1 -B_1^\ast]T_{1;\lambda} = T_{1;\lambda}^\ast  [\mathcal{E}_{\mathbb{H}_0}, B_1]T_{1;\lambda} + \mathrm{h.c.}
\]
For simplicity, we only consider the first term involving $k\cdot s$ in the definition  \eqref{eq: error kin energy} of $\mathcal{E}_{\mathbb{H}_0}$, i.e., 
\begin{equation}\label{def:E1}
  \mathcal{E}_{\mathbb{H}_0;1} := -\frac{2}{L^3}\sum_{k,s,s^\prime \in\Lambda^*} k\cdot s\,\hat{\varphi}(k)\widehat{\chi}_>(k)\, b_{k,s,\uparrow}b_{-k,s^\prime, \downarrow} + \mathrm{h.c.},
\end{equation}
the estimate for the other term involving $k\cdot s'$ is analogous. For the calculation of the commutator $ [\mathcal{E}_{\mathbb{H}_0;1}, B_1]$, we need to compute  
\begin{align}\label{eq: comm H0 T1}
  &[\hat{a}_{-s^\prime, \downarrow}^\ast\hat{a}_{s^\prime - k,\downarrow}^\ast \hat{a}^\ast_{-s,\uparrow}\hat{a}^\ast_{s+k,\uparrow}, \hat{a}_{r+p,\uparrow}\hat{a}_{-r,\uparrow}\hat{a}_{r^\prime - p, \downarrow}\hat{a}_{-r^\prime, \downarrow}]
  \\
  &= \hat{a}_{s+k,\uparrow}^\ast \hat{a}_{s^\prime - k, \downarrow}^\ast \hat{a}_{r+p,\uparrow}\hat{a}_{r^\prime - p, \downarrow}(\delta_{s,r}\delta_{s^\prime, r^\prime}-\delta_{s,r}\hat{a}^\ast_{-s^\prime, \downarrow}\hat{a}_{-r^\prime, \downarrow} - \delta_{s^\prime, r^\prime}\hat{a}_{-s,\uparrow}^\ast\hat{a}_{-r,\uparrow} )\nonumber
  \\
  &\quad +  (  \delta_{s+k,r+p}\hat{a}^\ast_{s^\prime - k, \sigma^\prime}\hat{a}_{r^\prime - p, \downarrow} + \delta_{s^\prime - k, r^\prime - p} \hat{a}^\ast_{s+k,\uparrow} \hat{a}_{r+p, \uparrow} - \delta_{s^\prime - k, r^\prime -p}\delta_{r+p,s+k} )\hat{a}_{-r,\uparrow}\hat{a}_{-r^\prime, \downarrow}\hat{a}^\ast_{-s^\prime, \downarrow}\hat{a}^\ast_{-s,\uparrow}\nonumber
\end{align}
Correspondingly, we have $[\mathcal{E}_{\mathbb{H}_0;1}, B_1] = \sum_{j=1}^6 \mathrm{I}_j$, 
with
 \begin{multline}\label{eq: def Ia}
  \mathrm{I}_1 = \frac{1}{L^6}\sum_{k,p,r,r^\prime} k\cdot r \hat{\varphi}(k)\widehat{\chi}_>(k) \hat{v}_\uparrow(r)\hat{\varphi}(p)\widehat{\chi}_>(p)\hat{v}_\downarrow(r^\prime)\hat{u}_\uparrow(r+k)\hat{u}_\downarrow(r^\prime - k)\hat{u}_\uparrow(r+p)\hat{u}_\downarrow(r^\prime - p)  
\\
  \times \hat{a}_{r+k,\uparrow}^\ast\hat{a}_{r^\prime - k, \downarrow}^\ast \hat{a}_{r+p,\uparrow}\hat{a}_{r^\prime - p, \downarrow},
\end{multline}
\begin{multline}\label{eq: def Ib}
  \mathrm{I}_2 = -\frac{1}{L^6}\sum_{k,p,r,r^\prime,s}k\cdot r\hat{\varphi}(k)\widehat{\chi}_>(k)\hat{v}_\uparrow(r)\hat{\varphi}(p)\widehat{\chi}_>(p) \hat{v}_\downarrow(r^\prime)\hat{v}_\downarrow(s)\hat{u}_\uparrow(r+k)\hat{u}_\downarrow(s-k)\hat{u}_\uparrow(r+p)\hat{u}_\downarrow(r^\prime - p)
  \\
  \times  \hat{a}^\ast_{r+k,\uparrow}\hat{a}^\ast_{s-k,\downarrow}\hat{a}^\ast_{-s,\downarrow}\hat{a}_{-r^\prime, \downarrow}\hat{a}_{r+p,\uparrow}\hat{a}_{r^\prime - p, \downarrow}, 
\end{multline}
\begin{multline}\label{eq: def Ic}
  \mathrm{I}_3 = -\frac{1}{L^6}\sum_{k,p,r,r^\prime,s} k \cdot s \hat{\varphi}(k)\widehat{\chi}_>(k)\hat{v}_\uparrow(s)\hat{\varphi}(p)\widehat{\chi}_>(p)\hat{v}_\downarrow(r^\prime)\hat{v}_\uparrow(r)\hat{u}_\uparrow(s+k)\hat{u}_\downarrow(r^\prime - k)\hat{u}_\uparrow(p+r)\hat{u}_\downarrow(r^\prime - p)  
\\
  \times \hat{a}_{s+k,\uparrow}^\ast\hat{a}_{r^\prime - k, \downarrow}^\ast \hat{a}_{-s,\uparrow}^\ast \hat{a}_{-r,\uparrow} \hat{a}_{r+p,\uparrow}\hat{a}_{r^\prime - p, \downarrow},
\end{multline}
\begin{multline}\label{eq: term II H0 T1}
  \mathrm{I}_4 = \frac{1}{L^6}\sum_{k,p, r,r^\prime, s} k\cdot(  r +p - k )  \hat{\varphi}(k)\widehat{\chi}_>(k)\hat{\varphi}(p)\widehat{\chi}_>(p)\hat{u}_\uparrow(r+p)\hat{u}_\downarrow(s - k)\hat{u}_\downarrow(r^\prime - p)  
  \\
  \times \hat{v}_\uparrow(r+p-k)\hat{v}_\downarrow(s)\hat{v}_\uparrow(r)\hat{v}_\downarrow(r^\prime)  \hat{a}_{s - k,\downarrow}^\ast\hat{a}_{r^\prime - p,\downarrow}\hat{a}_{-r,\uparrow}\hat{a}_{-r^\prime, \downarrow}\hat{a}_{-s, \downarrow}^\ast\hat{a}_{-r-p+k,\uparrow}^\ast,
\end{multline}
\begin{multline}\label{eq: term IIb H0 T1}
  \mathrm{I}_5 = \frac{1}{L^6}\sum_{k,p, r,r^\prime, s} k\cdot s \hat{\varphi}(k)\widehat{\chi}_>(k)\hat{\varphi}(p)\widehat{\chi}_>(p)\hat{u}_\uparrow(r+p)\hat{u}_\uparrow(s + k)\hat{u}_\downarrow(r^\prime - p)  
  \\
  \times  \hat{v}_\uparrow(s)\hat{v}_\downarrow(r^\prime)\hat{v}_\uparrow(r)\hat{v}_\downarrow(r^\prime +k-p) \hat{a}_{s +k,\uparrow}^\ast\hat{a}_{r+ p,\uparrow}\hat{a}_{-r,\uparrow}\hat{a}_{-r^\prime, \downarrow}\hat{a}_{-r^\prime -k+p, \downarrow}^\ast\hat{a}_{-s,\uparrow}^\ast,
\end{multline}
and 
\begin{multline}\label{eq: term III H0 T1}
  \mathrm{I_6} = -\frac{1}{L^6} \sum_{k,p, r,r^\prime} k\cdot( r+p-k) \hat{\varphi}(k)\widehat{\chi}_>(k)\hat{\varphi}(p)\widehat{\chi}_>(p)\hat{u}_\uparrow(r+p)\hat{u}_\downarrow(r^\prime - p)\hat{v}_\uparrow(r+p-k)\hat{v}_\downarrow(r^\prime - p + k)  
  \\
  \times \hat{v}_\uparrow(r)\hat{v}_\downarrow(r^\prime)\, \hat{a}_{-r,\uparrow}\hat{a}_{-r^\prime, \downarrow}\hat{a}_{-r^\prime + p-k, \downarrow}^\ast\hat{a}_{-r-p+k,\uparrow}^\ast.
  \end{multline}

In the following, we shall estimate all these terms. 
We shall repeatedly use the bound
\begin{equation}\label{eq: est a v H0}
  \|a_\sigma(\partial^n v_{x})\| \leq \||\cdot|^n\hat{v}_\sigma\|_{2} \leq C\rho^{\frac{1}{2}+ \frac{n}{3}}
\end{equation}
for the fermionic creation and annihilation operators. 
Moreover, we shall employ the bounds in Lemma \ref{lem: bounds b b*}  in Lemma \ref{lem: bounds phi}.
We will estimate the expectation value of all the terms in a general  state $\psi\in\mathcal{F}_{\mathrm{f}}$. 
We first consider the error terms coming from the second line on the right hand side in \eqref{eq: comm H0 T1}, i.e., $\mathrm{I}_1$, $\mathrm{I}_2$ and $\mathrm{I}_3$. It is convenient to divide them all into two parts, which corresponds to an integration by parts in configuration space: this allows us to get a factor $\rho^{1/3}$ when the derivative hits a $v_\sigma$, and to use $\mathbb{H}_0$ instead of $\mathcal{N}$ to estimate the error terms when the derivative hits a $u_\sigma$. We start by estimating the term $\mathrm{I}_1$ in \eqref{eq: def Ia}. 
 Writing $k\cdot r = - r^2 + r\cdot (r+k)$, 
 we split $\mathrm{I}_1$ into two terms, which we denote by $\mathrm{I}_{1;a}$ and $\mathrm{I}_{1;b}$, respectively. In $\mathrm{I}_{1;b}$ it is convenient to replace $\hat{u}_\uparrow(r+k)$ by $\hat{\nu}^>_\uparrow(r+ k)$ with $\hat{\nu}^>$ defined as 
 \begin{equation}\label{eq: def nu >}
 \hat{\nu}^>(k) = \begin{cases} 1&\mbox{if}\,\,\, |k| > 3\rho^{\frac{1}{3} - \gamma}, \\ 0 &\mbox{if}\,\,\, |k| \leq 3\rho^{\frac{1}{3} - \gamma}. \end{cases}
 \end{equation}
 This can be done since $|k| \geq 4\rho^{1/3 - \gamma}$ and $|r| \leq k_F^\uparrow$ for all the terms in the sum. The same argument applies to $I_2$ and $I_3$. 
 
 To estimate $\mathrm{I}_1$, we rewrite both $\mathrm{I}_{1;a}$ and $\mathrm{I}_{1;b}$ in configuration space. We have 
$$
   \mathrm{I}_{1;a} 
  = \int dxdydzdz^\prime\, \widetilde{\varphi}(x-y)\widetilde{\varphi}(z-z^\prime) \left(\Delta v_\uparrow\right)(x;z)v_\downarrow(y; z^\prime)  a^\ast_\uparrow(u_x)a^\ast_\downarrow (u_y) a_\uparrow(u_z)a_\downarrow(u_{z^\prime})  .
$$
For fixed $y$ and $z$, we can use \eqref{eq: est a v H0} and $0\leq \hat{u}_\uparrow\leq 1$ to bound 
\[
  \left\| \int dx\, \widetilde{\varphi}(x-y) \left( \Delta v_\uparrow\right)(x;z) a_\uparrow(u_x)\right\| 
  \leq \|\widetilde{\varphi}_y  \Delta v_{z,\uparrow}\|_2,
\]
where we used the notation $\widetilde{\varphi}_{y}(\,\cdot\,) = \widetilde{\varphi}(y-\,\cdot\,)$, $\Delta v_{z,\uparrow}(\,\cdot\,) = \Delta v_\uparrow(\,\cdot\,;z)$, as in Lemma~\ref{lem: bounds b b*} and the discussion after \eqref{eq: def u,v}. Similarly, 
\begin{equation}\label{eq: est conv L2}
  \left\| \int dz^\prime\, \widetilde{\varphi}(z-z^\prime) v_\downarrow(y;z^\prime) a_\downarrow(u_{z^\prime})\right\|\leq \|\widetilde{\varphi}_{z} v_{y, \downarrow}\|_2.
\end{equation}
With the aid of the Cauchy--Schwarz's inequality, we hence find that 
\begin{multline*}
  |\langle \psi, \mathrm{I}_{1;a}\psi\rangle| \leq \int dy dz\, \|\widetilde{\varphi}_y  \Delta v_{z,\uparrow}\|_2 \|\widetilde{\varphi}_{z} v_{y, \downarrow}\|_2 \|a_\downarrow(u_y)\psi\| \| a_\uparrow(u_{z})\psi\| 
\\
  \leq \left( \int dydzdw\, | \widetilde{\varphi}(w-y)|^2|\Delta v_\uparrow(z;w)|^2  \|a_\downarrow(u_y)\psi\|^2\right)^{\frac{1}{2}}\left(\int dydzdw\,  |\widetilde{\varphi}(z-w)|^2 |v_\downarrow(y;w)|^2 \| a_\uparrow(u_{z})\psi\|^2   \right)^{\frac{1}{2}}.
\end{multline*}
Thus, using also Lemma \ref{lem: bounds phi}, we obtain
\begin{equation}\label{eq: est Ia1}
  |\langle \psi, \mathrm{I}_{1;a}\psi\rangle| \leq   \|\widetilde{\varphi}\|_2^2\|\Delta v_\uparrow\|_2\|v_\downarrow\|_2  \langle \psi, \mathcal{N}\psi\rangle \leq 
  C\rho^{\frac{4}{3} + \gamma}\langle \psi, \mathcal{N}\psi\rangle.
\end{equation}
The estimate of the term $\mathrm{I}_{1;b}$ is similar. We  have 
$$
   \mathrm{I}_{1;b} = \sum_{\ell =1}^3\int dxdydzdz^\prime\, \widetilde{\varphi}(x-y)\widetilde{\varphi}(z-z^\prime) \partial_\ell v_\uparrow(x;z)v_\downarrow(y; z^\prime)  a^\ast_\uparrow(\partial_\ell \nu_x^>)a^\ast_\downarrow (u_y) a_\uparrow(u_z)a_\downarrow(u_{z^\prime})
$$
where $\nu_x^>(\cdot) = \nu^>(\cdot - x)$.
Proceeding as above, we can bound 
\begin{equation}\label{eq: est conv L2 partial ell}
  \left\| \int dz\, \widetilde{\varphi}(z-z^\prime) \partial_\ell v_\uparrow(x;z) a_\uparrow(u_z)\right\|\leq \|\widetilde{\varphi}_{z^\prime} \partial_\ell v_{x, \uparrow}\|_2.
\end{equation}
Using also \eqref{eq: est conv L2} and that 
\begin{equation}\label{eq: est partial u H0} 
\sum_{\ell =1}^3 \int dx\, \|a_\uparrow(\partial_\ell \nu_x^>)\psi\|^2 = \sum_{k}|k|^2 |\hat{\nu}^>_\uparrow(k)|^2 \|\hat{a}_{k,\uparrow}\psi\|^2
 \leq C\langle \psi, \mathbb{H}_0 \psi\rangle,
\end{equation}
we obtain the bound 
\begin{equation}\label{eq: est Ia2}
  |\langle\psi, \mathrm{I}_{1;b}\psi \rangle| \leq C \|\widetilde{\varphi}\|_2^2\|\nabla v_\uparrow\|_2\|v_\downarrow\|_2  \|\mathbb{H}_0^{\frac{1}{2}}\psi\| \|\mathcal{N}^{\frac{1}{2}}\psi\|\leq  C\rho^{1 + \gamma}\|\mathbb{H}_0^{\frac{1}{2}}\psi\| \|\mathcal{N}^{\frac{1}{2}}\psi\|.
\end{equation}

We now estimate the term $\mathrm{I}_2$ defined in \eqref{eq: def Ib}.
Similarly as above, we write $k\cdot r = - r^2 + r\cdot(r+k)$ 
and correspondingly we split $\mathrm{I}_{2}$ as a sum of two different terms,  $\mathrm{I}_{2;a}$ and $\mathrm{I}_{2;b}$. In order to bound $\mathrm{I}_{2;a}$, we  rewrite it partially in configuration space. 
Recalling the definition of $b_\sigma$  in \eqref{eq: def b phi}), we find
$$
  \langle\psi, \mathrm{I}_{2;a}\psi\rangle 
  = \frac{1}{L^3}\sum_r |r|^2 \hat{v}_\uparrow(r)\left\|  \int dx\, e^{-ir\cdot x} b_\downarrow(\widetilde{\varphi}_x)a_\uparrow(u_x) \psi\right\|^2.
$$
We can bound $|r|^2\hat{v}_\uparrow(r) \leq C\rho^{2/3}$ for all $r\in\Lambda^*$. 
From the bounds in Lemma \ref{lem: bounds b b*}, we find
\begin{align*}
  \frac{1}{L^3}\sum_r\left\|\int dx\, e^{-ir\cdot x} b_\downarrow(\widetilde{\varphi}_x)a_\uparrow(u_x) \psi\right\|^2 & = \int dx\, \| b_\downarrow(\widetilde{\varphi}_x) a_\uparrow(u_x)\psi\|^2  \nonumber
  \\
  &\leq \int dx\, \| b_\downarrow(\widetilde{\varphi}_x)\|^2 \|a_\uparrow(u_x)\psi\|^2 \leq C\rho^{\frac{2}{3} +\gamma}\langle \psi,\mathcal{N}\psi\rangle,
\end{align*}
and hence we obtain
\begin{equation}\label{eq: est Ib1}
  |\langle \psi, \mathrm{I}_{2;a}\psi\rangle| \leq C\rho^{\frac{4}{3} +\gamma}\langle \psi, \mathcal{N}\psi\rangle.
\end{equation}
 The estimate of $\mathrm{I}_{2;b}$ can be done analogously as the one of $\mathrm{I}_{2;a}$, using in addition \eqref{eq: est partial u H0}. We omit the details, and  directly write the final estimate as 
\begin{equation}\label{eq: est Ib2}
  |\langle \psi, \mathrm{I}_{2;b}\psi\rangle| \leq C\rho^{1+\gamma}\|\mathbb{H}_0^{\frac{1}{2}}\psi\|\|\mathcal{N}^{\frac{1}{2}}\psi\|.
\end{equation}

Next we consider $\mathrm{I}_3$  in \eqref{eq: def Ic}. Again we split the term in two, $\mathrm{I}_{3;a}$ and $\mathrm{I}_{3;b}$, by writing $k= r' - (r'-k)$.
%
The term $\mathrm{I}_{3;a}$ can then be written as
$$
  \mathrm{I}_{3;a}  =  -\sum_{\ell=1}^3 \sum_{r^\prime} r^\prime_\ell \hat{v}_\downarrow(r^\prime)\int dydz^\prime\, e^{i r^\prime \cdot (y-z^\prime)}   a^\ast_\downarrow(u_y){b}^\ast_{\ell,\uparrow}(\widetilde{\varphi}_y) b_\uparrow(\widetilde{\varphi}_{z^\prime}) a_\downarrow(u_{z^\prime})
$$
%
where we used the operators introduced in \eqref{eq: def b phi} and \eqref{eq: def b tilde}. Using the bounds in Lemma \ref{lem: bounds b b*} and that $|r^\prime_\ell \hat{v}_\downarrow(r^\prime)| \leq C\rho^{1/3}$, we get with the aid of the Cauchy--Schwarz inequality
\begin{align} \nonumber
  &|\langle\psi, \mathrm{I}_{3;a}\psi\rangle| \\ &\leq  C\rho^{\frac{1}{3}}\sum_{\ell = 1}^3\left(\frac{1}{L^3}\sum_{r^\prime}\left\| \int dy\, e^{-ir^\prime\cdot y}\,{b}_{\ell,\uparrow}(\widetilde{\varphi}_y)a_\downarrow(u_y)\psi \right\|^2\right)^{\frac{1}{2}} 
  \left(\frac{1}{L^3}\sum_{r^\prime}\left\|\int dz^\prime e^{-ir^\prime\cdot z^\prime} \, b_\uparrow(\widetilde{\varphi}_{z^\prime}) a_\downarrow( u_{z^\prime})\psi \right\|^2\right)^{\frac{1}{2}} \nonumber
  \\
  &= C\rho^{\frac{1}{3}}\sum_{\ell = 1}^3\left(\int dy\, \|{b}_{\ell,\uparrow}(\widetilde{\varphi}_y)a_\downarrow(u_y)\psi\|^2 \right)^{\frac{1}{2}} \left( \int dz^\prime\, \|{b}_{\uparrow}(\widetilde{\varphi}_{z^\prime})a_\downarrow(u_{z^\prime})\psi\|^2 \right)^{\frac{1}{2}}\leq C\rho^{\frac{4}{3}+\gamma}\langle \psi, \mathcal{N}\psi\rangle. \label{eq: est Ic1}
\end{align}
The analysis of  $\mathrm{I}_{3;b}$
 can be done similarly as the one for $\mathrm{I}_{3;a}$, using also the inequality in \eqref{eq: est partial u H0}. The result is
 \begin{equation}\label{eq: est Ic2}
  |\langle \psi, \mathrm{I}_{3;b}\psi\rangle| \leq C\rho^{1+\gamma}\|\mathbb{H}_0^{\frac{1}{2}} \psi\|\|\mathcal{N}^{\frac{1}{2}}\psi\|.
\end{equation}

The next error term  we estimate is $\mathrm{I}_4$ in \eqref{eq: term II H0 T1}. As above, it is convenient to rewrite it in configuration space. Since the constraints $|p| \geq 4\rho^{1/3 - \gamma}$ and $|r| \leq k_F^\uparrow$ imply that $|r+p|> 3\rho^{1/3 -\gamma}$, we have that $\hat{u}_\uparrow(r+p) = 1$ in \eqref{eq: term II H0 T1}, and hence
\[
   \mathrm{I_4}  =  \sum_{\ell = 1}^3\int dxdydz\, \partial_\ell\widetilde{\varphi}(x-y)\widetilde{\varphi}(x-z)  a^\ast_\downarrow(u_y) a_\uparrow(v_x) a^\ast_\downarrow(v_{z}) a_\downarrow(v_{y}) a^\ast_\uparrow(\partial_\ell v_x)   a_\downarrow(u_{z}) .
\]
Using Lemma \ref{lem: bounds phi} and  \eqref{eq: est a v H0}, we can bound 
\begin{eqnarray}\label{eq: est II H0}
  |\langle \psi, \mathrm{I}_4\psi\rangle| &\leq& C\rho^{2 + \frac{1}{3}}\sum_{\ell = 1}^3\int dxdydz\, |\partial_\ell\widetilde{\varphi}(x-y)||\widetilde{\varphi}(x-z)| \|a_\downarrow(u_y)\psi\|\|a_\downarrow(u_z)\psi\| \nonumber
  \\
  &\leq& C \rho^{2+\frac{1}{3}}\sum_{\ell = 1}^3\|\partial_\ell\widetilde{\varphi}\|_1 \|\widetilde{\varphi}\|_1 \langle \psi, \mathcal{N}\psi\rangle \leq C\rho^{\frac{4}{3} +3\gamma}\langle \psi, \mathcal{N}\psi\rangle.
\end{eqnarray}
The estimate for $\mathrm{I}_5$ can be done similarly, we omit the details. Also for this term we obtain
\begin{equation}\label{I5}
  |\langle \psi, \mathrm{I}_5\psi\rangle|\leq C\rho^{\frac{4}{3} + 3\gamma}\langle \psi, \mathcal{N}\psi\rangle.
\end{equation}

Finally we consider the last term $\mathrm{I}_6$ in \eqref{eq: term III H0 T1}. Similarly as above, we have 
$\hat{u}_\uparrow(r+p) = \hat{u}_\downarrow(r^\prime -p) = 1$ for all the terms in the sum. To estimate it, we shall write
$$
   -\hat{a}_{-r,\uparrow}\hat{a}_{-r^\prime, \downarrow}\hat{a}^\ast_{-r^\prime + p - k,\downarrow}\hat{a}^\ast_{-r-p+k, \uparrow} \\ = \hat{a}^\ast_{-r^\prime +p -k,\downarrow}\hat{a}_{-r,\uparrow}\hat{a}^\ast_{-r-p+k,\uparrow}\hat{a}_{-r^\prime, \downarrow} - \delta_{p,k}\hat{a}_{-r,\uparrow}\hat{a}_{-r,\uparrow}^\ast  
$$
and correspondingly split $\mathrm{I}_6$ in two terms, denoted by $\mathrm{I}_{6;a}$ and $\mathrm{I}_{6;b}$. The last one vanishes, 
$$
  \mathrm{I}_{6;b}   =  - \frac{1}{L^6} \sum_{p, r,r^\prime} p\cdot r \left| \hat{\varphi}(p)\widehat{\chi}_>(p) \right|^2  \hat{v}_\downarrow(r^\prime) \hat{v}_\uparrow(r)  \hat{a}_{-r,\uparrow} \hat{a}_{-r,\uparrow}^\ast  =0
$$
since the sum over $p$ is zero  by symmetry. 
%
The term $\mathrm{I}_{6;a}$ we can write in configuration space as
\begin{equation}\label{eq: term IIIa H0 T1}
  \mathrm{I}_{6;a}
= \sum_{\ell=1}^3\int dxdy\, \widetilde{\varphi}(x-y) \partial_\ell\widetilde{\varphi}(x-y)\widetilde{\varphi}(x-y) a^\ast_\downarrow(v_y) a_\uparrow(v_x)  a_\uparrow^\ast(\partial_\ell v_{x})a_\downarrow (v_y) .
\end{equation}
We use $\widetilde\varphi\partial_\ell \widetilde\varphi = \frac 12 \partial_\ell \widetilde\varphi^2$ and integrate by parts in $x$. The derivate then hits either $  a_\uparrow(v_x)$ or $a_\uparrow^\ast(\partial_\ell v_{x})$, and using \eqref{eq: est a v H0} we can bound
\begin{equation}\label{eq: est IIIa H0}
  |\langle \psi, \mathrm{I}_{6;a}\psi\rangle| \leq C\rho^{1 +\frac{2}{3}}\int dxdy\, |\widetilde{\varphi}(x-y)|^2 \|a_\downarrow(v_y)\psi\|^2
 \leq  C\rho^{\frac 43+\gamma}\langle \psi, \mathcal{N}\psi\rangle,
\end{equation}
where we used Lemma \ref{lem: bounds phi} in the last step. 
Collecting all the estimates, this concludes the proof.
\end{proof}

\subsection{Conjugation of $\mathbb{Q}_4$}\label{sec: conj Q4 T1}
We now conjugate $\mathbb{Q}_4$ by $T_1$ (see \eqref{eq: def H-corr} and \eqref{eq: def T1 unitary} for their definition).

\begin{proposition}[Conjugation of $\mathbb{Q}_4$ by $T_1$]\label{pro: Q4}
Let $\lambda\in [0,1]$ and $\gamma\in (0,1/3)$.  Under the same assumptions of Theorem \ref{thm: main up bd} 
\begin{equation} \label{eq: conj Q4 T1}
  \partial_\lambda T_{1;\lambda}^\ast \mathbb{Q}_4 T_{1;\lambda} =  T^{\ast}_{1;\lambda}(\mathbb{T}_2 + \mathcal{E}_{\mathbb{Q}_4})T_{1;\lambda},
\end{equation}
where 
\begin{equation}\label{eq: def T2}
  \mathbb{T}_2 :=  -\frac{1}{L^3}\sum_{k\in \Lambda^*}\hat{V}\ast_{_{\sum}} (\hat{\varphi}\widehat{\chi}_>)(k)\, b_{k,\uparrow}b_{-k,\downarrow} + \mathrm{h.c.}
\end{equation}
and $\mathcal{E}_{\mathbb{Q}_4}$ is such that for any $\psi \in \mathcal{F}_{\mathrm{f}}$
\begin{equation}\label{eq: est err Q4 after  T1}
  |\langle \psi, \mathcal{E}_{\mathbb{Q}_{4}} \psi\rangle| \leq  C\rho^{\frac{5}{6} +\frac{\gamma}{2}}\|\mathcal{N}^{\frac{1}{2}}\psi\|\|\mathbb{Q}^{\frac{1}{2}}_4 \psi\| .
\end{equation}
\end{proposition}
\begin{proof} 
We start by computing 
\[
 \partial_\lambda  T^\ast_{1;\lambda} \mathbb{Q}_4 T_{1; \lambda}  
 =  T^\ast_{1;\lambda}[\mathbb{Q}_4, B_1] T_{1;\lambda} + \mathrm{h.c.}
\]
with
\[
  [\mathbb{Q}_4, B_1] = \frac{1}{2L^6}\sum_{\sigma,\sigma^\prime}\sum_{k,p,s,s^\prime}\hat{V}(k)\hat{\varphi}(p)\widehat{\chi}_>(p)\hat{u}_\sigma(s+k)\hat{u}_{\sigma^\prime}(s^\prime -k)\hat{u}_\sigma(s)\hat{u}_{\sigma^\prime}(s^\prime)\, [\hat{a}^\ast_{s+k,\sigma}\hat{a}^\ast_{s^\prime - k,\sigma^\prime}\hat{a}_{s^\prime, \sigma^\prime}\hat{a}_{s,\sigma}, b_{p,\uparrow}b_{-p,\downarrow}].
\]
Recall from \eqref{eq: def bp bpr} that $b_{p,\sigma} = \sum_{r} \hat{u}_\sigma(p+r)\hat{v}_\sigma(r)\hat{a}_{p+r,\sigma}\hat{a}_{-r,\sigma}$.
Hence we need to compute
\begin{align*}
  &[\hat{a}^\ast_{s+k,\sigma}\hat{a}^\ast_{s^\prime - k,\sigma^\prime}\hat{a}_{s^\prime, \sigma^\prime}\hat{a}_{s,\sigma}, 
  \hat{a}_{p+r,\uparrow}\hat{a}_{-r,\uparrow} \hat{a}_{-p+r',\downarrow}\hat{a}_{-r',\downarrow}
] = - \delta_{s +k, r+p}\delta_{\sigma, \uparrow}\hat{a}^\ast_{s^\prime -k,\sigma^\prime}\hat{a}_{r^\prime - p,\downarrow}\hat{a}_{-r^\prime, \downarrow}\hat{a}_{-r,\uparrow}\hat{a}_{s^\prime, \sigma^\prime}\hat{a}_{s,\sigma} 
  \\
  & +\delta_{s^\prime - k, r+p}\delta_{\sigma^\prime,\uparrow}\hat{a}_{s+k,\sigma}^\ast\hat{a}_{r^\prime - p, \downarrow}\hat{a}_{-r^\prime, \downarrow}\hat{a}_{-r,\uparrow}\hat{a}_{s^\prime, \sigma^\prime}\hat{a}_{s,\sigma} - \delta_{s^\prime -k, r^\prime -p}\delta_{\sigma^\prime, \downarrow}\hat{a}^\ast_{s+k,\sigma}\hat{a}_{r+p,\uparrow}\hat{a}_{-r^\prime, \downarrow}\hat{a}_{-r,\uparrow}\hat{a}_{s^\prime, \sigma^\prime}\hat{a}_{s,\sigma} 
  \\
  & + \delta_{s+k,r^\prime - p}\delta_{\sigma,\downarrow}\hat{a}^\ast_{s^\prime - k, \sigma^\prime}\hat{a}_{r+p,\uparrow}\hat{a}_{-r^\prime, \downarrow}\hat{a}_{-r,\uparrow}\hat{a}_{s^\prime, \sigma^\prime}\hat{a}_{s,\sigma} +\delta_{s+k,r+p}\delta_{s^\prime - k, r^\prime -p}\delta_{\sigma,\uparrow}\delta_{\sigma^\prime,\downarrow}\hat{a}_{-r^\prime, \downarrow}\hat{a}_{-r,\uparrow}\hat{a}_{s^\prime, \sigma^\prime}\hat{a}_{s,\sigma}
  \\
  & - \delta_{s+k,r^\prime - p}\delta_{s^\prime -k,r+p}\delta_{\sigma,\downarrow}\delta_{\sigma^\prime, \uparrow}\hat{a}_{-r^\prime, \downarrow}\hat{a}_{-r,\uparrow}\hat{a}_{s^\prime, \sigma^\prime}\hat{a}_{s,\sigma}.
\end{align*}
Using that $\hat{V}(k) = \hat{V}(-k)$, $\hat{\varphi}(p)\widehat{\chi}_>(p) = \hat{\varphi}(-p)\widehat{\chi}_>(-p)$ and performing a suitable change of variables, we can combine the first two terms in the commutator above, as well as the third and the forth and the last two. Therefore, we find that 
\[
  [\mathbb{Q}_4, B_1] = \mathrm{I}_1 + \mathrm{I}_2 + \mathrm{I}_3,
\]
with
\begin{multline}\label{eq: def I1 Q4 T1}
   \mathrm{I}_1 =  \frac{1}{L^6}\sum_{\sigma}\sum_{k,p,r,r^\prime, s} \hat{V}(k)\hat{\varphi}(p)\widehat{\chi}_>(p) \hat{u}_\uparrow(r+p)\hat{u}_\downarrow(r^\prime - p)\hat{u}_\uparrow(r+p-k)\hat{u}_{\sigma}(s -k)\hat{u}_{\sigma}(s)\hat{v}_\uparrow(r)\hat{v}_\downarrow(r^\prime) 
    \\
    \times \hat{a}^\ast_{s - k, \sigma}\hat{a}_{r^\prime - p, \downarrow}\hat{a}_{-r,\uparrow}\hat{a}_{-r^\prime, \downarrow}\hat{a}_{s, \sigma}\hat{a}_{r+p-k,\uparrow},
\end{multline}
\begin{multline*}
   \mathrm{I}_2 =  \frac{1}{L^6}\sum_{\sigma}\sum_{k,p,r,r^\prime, s} \hat{V}(k)\hat{\varphi}(p)\widehat{\chi}_>(p) \hat{u}_\uparrow(r+p)\hat{u}_\downarrow(r^\prime - p)\hat{u}_\downarrow(r^\prime -p-k)\hat{u}_{\sigma}(s -k)\hat{u}_{\sigma}(s)\hat{v}_\uparrow(r)\hat{v}_\downarrow(r^\prime) 
    \\
    \times \hat{a}^\ast_{s - k, \sigma}\hat{a}_{r + p, \uparrow}\hat{a}_{-r^\prime,\downarrow}\hat{a}_{-r, \uparrow}\hat{a}_{s, \sigma}\hat{a}_{r^\prime -p-k,\downarrow},
\end{multline*}
\begin{multline*}
  \mathrm{I}_3 = - \frac{1}{L^6}\sum_{k,q, r,r^\prime}\hat{V}(k)\hat{\varphi}(p)\widehat{\chi}_{>}(p)\hat{u}_\uparrow(r+p)\hat{u}_\downarrow(r^\prime - p) \hat{u}_\uparrow(r+p+k)\hat{u}_\downarrow(r^\prime -p-k)\hat{v}_\uparrow(r)\hat{v}_\downarrow(r^\prime)
  \\
   \times \hat{a}_{r+p+k, \uparrow}\hat{a}_{-r,\uparrow}\hat{a}_{r^\prime -p-k, \downarrow}\hat{a}_{-r^\prime, \downarrow}.
 \end{multline*}  %
In $\mathrm{I}_3$ we have $\hat{u}_\uparrow(r+p) = \hat{u}_\downarrow(r^\prime -p) =1$ due to the constraints on $p,r,r^\prime$ and, with another change of variables, we find that $\mathrm{I}_3 +\mathrm{I}_3^\ast = \mathbb{T}_2$ defined in \eqref{eq: def T2}.
%
In other words, $ \mathcal{E}_{\mathbb{Q}_4} = \mathrm{I}_1 + \mathrm{I}_2  + \mathrm{h.c.}$.  
 
In the following we only consider  $\mathrm{I}_1$, the bound on $\mathrm{I}_2$ works in exactly the same way. Using again that $\hat{u}_\uparrow(r+p) = 1$ in \eqref{eq: def I1 Q4 T1}, we can write in configuration space 
\begin{equation}\label{eq: err Q1 position}
   \mathrm{I}_1 
=-\sum_\sigma\int dxdy\, V(x-y) a^\ast_{\sigma}(u_y)b_{\downarrow}(\widetilde{\varphi}_x) a_\uparrow(v_x)a_\sigma (u_y)a_{\uparrow}(u_x) ,
\end{equation}
where we recall the definition of the operator $b_\sigma$  in \eqref{eq: def b phi}. Using $\|a_{\downarrow}(v_x)\| = \|{v}\|_2 \leq \rho^{1/2}$,  Lemma \ref{lem: bounds b b*} as well as the Cauchy--Schwarz inequality, we obtain
\begin{equation}\label{eq: est err term Q4 T1}
  |\langle \psi, \mathrm{I}_1\psi\rangle | \leq  C \rho^{\frac{5}{6}  +\frac{\gamma}{2}} \sum_{\sigma}  \int dxdy\, V(x-y)\|a_{\sigma}(u_y)\psi\|\|a_\sigma (u_y)a_{\uparrow}(u_x)\psi\|\leq  C\rho^{\frac{5}{6} +\frac{\gamma}{2}}\|\mathcal{N}^{\frac{1}{2}}\psi\|\|\mathbb{Q}_4^{\frac{1}{2}}\psi\|.
\end{equation}
This yields  \eqref{eq: est err Q4 after  T1}. 
\end{proof}


\subsection{Conjugation of $\mathbb{Q}_2^\udarrow$}\label{sec: conj Q2 T1}
In the next proposition we conjugate the operator $\mathbb{Q}_2^\udarrow$ introduced in \eqref{q2ud} with respect to the unitary  operator $T_1$ defined in \eqref{eq: def T1 unitary}. 
\begin{proposition}[Conjugation of $\mathbb{Q}_2^\udarrow$ by $T_1$]\label{pro: Q2} Let $\lambda\in [0,1]$ and $\gamma\in (0,1/3)$. Under the same assumptions of Theorem \ref{thm: main up bd}
\[
  \partial_\lambda T_{1;\lambda}^\ast \mathbb{Q}_2^\udarrow T_{1;\lambda} = - 2\rho_\uparrow\rho_\downarrow \sum_{p\in \Lambda^*}\hat{V}(p)\hat{\varphi}(p)\widehat{\chi}_>(p) + T^\ast_{1;\lambda}\mathcal{E}_{\mathbb{Q}_2^\udarrow}T_{1;\lambda},
\]
with $\mathcal{E}_{\mathbb{Q}_2^\udarrow}$ such that for any $\psi \in \mathcal{F}_{\mathrm{f}}$
\begin{equation}\label{eq: est err Q2}
  |\langle \psi, \mathcal{E}_{\mathbb{Q}_2^\udarrow} \psi\rangle|\leq C\rho^{\frac{5}{6} +\frac{\gamma}{2}}\|\mathbb{Q}_{4}^{\frac{1}{2}}\psi\|\|\mathcal{N}^{\frac{1}{2}}\psi\|  + C\rho\langle \psi, \mathcal{N}\psi\rangle.
\end{equation}

\end{proposition}
\begin{proof}
As in the previous propositions, we have 
  \begin{equation}\label{eq: derivative lambda Q2}
    \partial_\lambda T_{1;\lambda}^\ast \mathbb{Q}_2^\udarrow T_{1;\lambda} 
    =  T_{1;\lambda}^\ast [\mathbb{Q}_2^\udarrow, B_1]T_{1;\lambda} + \mathrm{h.c.}.
  \end{equation}
 For the computation of $[\mathbb{Q}_2^\udarrow, B_1]$ we shall use that
 $$
 \mathbb{Q}_2^\udarrow = \frac 1{L^3}\sum_{k,s,s'} \hat V(k) \hat{v}_\uparrow(s^\prime) \hat{u}_\uparrow(s^\prime-k)  \hat{v}_\downarrow(s) \hat{u}_\downarrow(s+k) \hat{a}_{-s^\prime, \uparrow}^\ast\hat{a}_{s^\prime - k,\uparrow}^\ast\hat{a}_{-s,\downarrow}^\ast\hat{a}_{s+k,\downarrow}^\ast + \textrm{h.c.}
 $$
 and hence we need 
 \begin{align}\nonumber
&[\hat{a}_{-s^\prime, \uparrow}^\ast\hat{a}_{s^\prime - k,\uparrow}^\ast\hat{a}_{-s,\downarrow}^\ast\hat{a}_{s+k,\downarrow}^\ast, \hat{a}_{r+p,\uparrow}\hat{a}_{-r,\uparrow}\hat{a}_{r^\prime - p, \downarrow}\hat{a}_{-r^\prime, \downarrow}] 
\\ &= \hat{a}^\ast_{s^\prime-k, \uparrow}\hat{a}_{s+k,\downarrow}^\ast \hat{a}_{r+p,\uparrow}\hat{a}_{r^\prime -p,\downarrow}  
\left(\delta_{s^\prime, r}\hat{a}^\ast_{-s,\downarrow}\hat{a}_{-r^\prime, \downarrow}    +\delta_{s,r^\prime}\hat{a}_{-s^\prime, \uparrow}^\ast\hat{a}_{-r,\uparrow} - \delta_{s^\prime, r} \delta_{s,r^\prime}  \right) \nonumber
\\
  &\quad  - \left( \delta_{s^\prime - k, r+p} \hat{a}^\ast_{s+k,\downarrow} \hat{a}_{r^\prime - p,\downarrow}   +  \delta_{s+k,r^\prime -p}\hat{a}^\ast_{s^\prime - k, \uparrow}\hat{a}_{r+p,\uparrow}
 - \delta_{s^\prime -k, r+p}\delta_{s+k,r^\prime - p} 
  \right) 
  \hat{a}_{-r,\uparrow}\hat{a}_{-r^\prime, \downarrow}\hat{a}_{-s^\prime, \uparrow}^\ast\hat{a}^\ast_{-s,\downarrow}
 .\label{eq: initial comm Q2 B1}
\end{align}
Correspondingly, from the six terms on the right-hand side we obtain $[\mathbb{Q}_2^\udarrow, B_1] = \sum_{j=1}^6 \mathrm{I}_j$, and we shall now estimate these terms.
The main term comes from normal-ordering of $\mathrm{I}_6$. As in the previous proofs, in our estimates we often use that $\|V\|_1\leq C$, together with the bounds proved in Lemma \ref{lem: bounds b b*}, Lemma \ref{lem: bounds phi} and \eqref{eq: est a v H0} with $n=0$.

We start by considering $\mathrm{I}_1$. It is convenient to rewrite the term partially in configuration space, as 
\begin{equation}\label{46:I1}
  \mathrm{I}_1 = 
  -\frac{1}{L^3}\sum_{r}\hat{v}_\uparrow(r)\int dxdydz\, e^{ir\cdot (x-z)} V(x-y)a^\ast_\downarrow(u_y)a^\ast_\uparrow(u_x)a^\ast_\downarrow(v_y)b_\downarrow(\widetilde{\varphi}_{z})a_\uparrow(u_z),
\end{equation}
where we used again the operator $b_\sigma$ introduced in \eqref{eq: def b phi}. Using  that $0\leq \hat{v}_\uparrow(r) \leq 1$ and the Cauchy--Schwarz inequality, we get
\begin{align*}
  &|\langle \psi, \mathrm{I}_1 \psi\rangle| 
  \\
  &\leq \left(\frac{1}{L^3}\sum_{r} \left\| \int dxdy \,V(x-y)e^{-ir\cdot x} a_\downarrow(v_y)a_\uparrow(u_x)a_\downarrow(u_y) \psi\right\|^2 \right)^{\frac{1}{2}}\left(\frac{1}{L^3}\sum_{r} \left\| \int dz\, e^{-ir\cdot z}b_\downarrow(\widetilde{\varphi}_{z})a_\uparrow(u_z) \psi\right\|^2 \right)^{\frac{1}{2}}.
\end{align*}
We can further bound 
\begin{align*}
 &\frac{1}{L^3}\sum_{r} \left\| \int dxdy \,V(x-y)e^{-ir\cdot x} a_\downarrow(v_y)a_\uparrow(u_x)a_\downarrow(u_y) \psi\right\|^2  \\
  &\leq \int dxdydy^\prime V(x-y)V(x-y^\prime) \|a_\downarrow(v_y)a_\uparrow(u_x)a_\downarrow(u_y)\psi\|\|a_\downarrow(v_{y^\prime})a_\uparrow(u_x)a_\downarrow(u_{y^\prime})\psi\|
  \\
  &\leq C\rho \int dxdydy^\prime\, V(x-y)V(x-y^\prime) \|a_\uparrow(u_x)a_\downarrow(u_y)\psi\|^2 
  \leq C\rho\langle \psi, \mathbb{Q}_4\psi\rangle,
\end{align*}
where we used \eqref{eq: est a v H0} and $\|V\|_1 \leq C$.
Using  Lemma \ref{lem: bounds b b*} 
we can also estimate
$$
\frac{1}{L^3}\sum_{r} \left\| \int dz\, e^{-ir\cdot z}b_\downarrow(\widetilde{\varphi}_{z})a_\uparrow(u_z) \psi\right\|^2 
  = \int dz \, \left\| b_\downarrow(\widetilde{\varphi}_{z}) a_\uparrow(u_z)\psi\right\|^2 
 \leq C\rho^{\frac{2}{3} +\gamma}\langle \psi, \mathcal{N}\psi\rangle. 
$$
In combination, we thus obtain 
\begin{equation}\label{eq: est Ia Q2}
  |\langle \psi, \mathrm{I}_1\psi\rangle| \leq 
  C\rho^{\frac{5}{6} +\frac{\gamma}{2}}\|\mathbb{Q}_4^{\frac{1}{2}}\psi\|\|\mathcal{N}^{\frac{1}{2}}\psi\|. 
\end{equation}
The term $\mathrm{I}_2$ in  can be estimated similarly, with the same outcome, and we omit the details.

Next we consider $\mathrm{I}_3$. Similarly as for the other terms, we rewrite it in configuration space as
\begin{equation}\label{46:I3}
  \langle \psi, \mathrm{I}_3\psi\rangle 
= \int dxdydz^\prime\, V(x-y)v_\downarrow(y;z^\prime)\left\langle  a_\uparrow(u_x)a_\downarrow(u_y)\psi, \left(\int dz\, \widetilde{\varphi}(z-z^\prime)v_\uparrow(x;z) a_\uparrow(u_z)\right) a_\downarrow(u_{z^\prime})\psi \right\rangle .
\end{equation}
By the Cauchy--Schwarz's inequality, we have 
\begin{multline*}
  |\langle \psi, \mathrm{I}_3\psi\rangle| \leq \left(\int dxdydz^\prime\, V(x-y) |v_\downarrow(y;z^\prime)|^2 \|a_\uparrow(u_x)a_\downarrow(u_y)\psi\|^2\right)^{\frac{1}{2}}
  \\
  \times\left(\int dxdydz^\prime\, V(x-y)\left\|\int dz\, \widetilde{\varphi}(z-z^\prime)v_\uparrow(x;z) a_\uparrow(u_z)\right\|^2 \|a_\downarrow(u_{z^\prime})\psi\|^2\right)^{\frac{1}{2}}.
\end{multline*}
The first factor is bounded by $\|v_\downarrow\|_2  \|\mathbb{Q}_4^{\frac{1}{2}}\psi\|$. To bound the second, we proceed as in \eqref{eq: est conv L2}, i.e., using \eqref{eq: est a v H0}  and $0\leq \hat{u}_\sigma\leq 1$, to obtain 
\begin{align*}
  |\langle \psi, \mathrm{I}_3\psi\rangle| &\leq \|v_\downarrow\|_2  \|\mathbb{Q}_4^{\frac{1}{2}}\psi\| \left(\int dxdydz^\prime dw\, V(x-y)  |\widetilde{\varphi}(z^\prime -w)|^2|v_\uparrow(x; w)|^2\|a_\downarrow(u_{z^\prime})\psi\|^2\right)^{\frac{1}{2}}
  \\
&  \leq \|V\|_1^{\frac{1}{2}}\|v_\uparrow\|_2\|v_\downarrow\|_2\|\widetilde{\varphi}\|_2 \|\mathbb{Q}_4^{\frac{1}{2}}\psi\|\|\mathcal{N}^{\frac{1}{2}}\psi\| \leq 
  C\rho^{\frac{5}{6} +\frac{\gamma}{2}}\|\mathbb{Q}_4^{\frac{1}{2}}\psi\|\|\mathcal{N}^{\frac{1}{2}}\psi\|. 
\end{align*}
Here we have used  Lemma~\ref{lem: bounds phi} in the last step. 

The next term we consider is $\mathrm{I}_4$. Before estimating it, we note that from the constraint $|p| \geq 4\rho^{1/3 -\gamma}$ and $|r| \leq k_F^\uparrow$, we get that $|r+p|> 3\rho^{1/3 - \gamma}$. Hence $\hat{u}_\uparrow(r+p) = 1$. Rewriting $\mathrm{I}_4$ in configuration space, we then get
\begin{equation}\label{46:I4}
  \mathrm{I}_4 = \int dxdydz\, V(x-y)\widetilde{\varphi}(y-z)a^\ast_\downarrow(u_x)a_\downarrow(u_z)a_\uparrow(v_y)a_\downarrow(v_z)a^\ast_\downarrow(v_x) a^\ast_\uparrow(v_y).
\end{equation}
Hence, using again Lemma \ref{lem: bounds phi} and the fact that $a_\downarrow(u_z)$ commutes with $a_\uparrow(v_y)a_\downarrow(v_z)a^\ast_\downarrow(v_x) a^\ast_\uparrow(v_y)$, 
\begin{eqnarray}\nonumber
  |\langle \psi, \mathrm{I}_4 \psi\rangle | &\leq& C\rho^2 \int dxdydz\, V(x-y) |\widetilde{\varphi}(y-z)| \|a_\downarrow(u_x)\psi\|\|a_\downarrow(u_z)\psi\| 
  \\ &\leq& C\rho^{2}\|V\|_1\|\widetilde{\varphi}\|_1 \langle \psi, \mathcal{N}\psi\rangle 
  \leq C\rho^{\frac{4}{3} +2\gamma}\langle \psi, \mathcal{N}\psi\rangle.  \label{eq: est II Q2}
\end{eqnarray}
The term $\mathrm{I}_5$  can be estimated in the same way, obtaining the same bound.

%

Finally we consider the term $\mathrm{I}_6$,  from which we extract the leading contribution. We start by normal-ordering the terms, using
\begin{equation}\label{noI6}
  \hat{a}_{-r,\uparrow}\hat{a}_{-r^\prime, \downarrow}\hat{a}_{-s^\prime, \uparrow}^\ast\hat{a}^\ast_{-s,\downarrow} = -\left( \hat{a}_{-s^\prime, \uparrow}^\ast  \hat{a}_{-r,\uparrow} - \delta_{r,s^\prime} \right) \left( \hat{a}^\ast_{-s,\downarrow} \hat{a}_{-r^\prime, \downarrow} - \delta_{s,r^\prime} \right)
\end{equation}
This way, we obtain
\begin{align}\nonumber
\mathrm{I}_6 &= -  \rho_\uparrow \rho_\downarrow \sum_{p} \hat V(p)\hat{\varphi}(p)\widehat{\chi}_>(p) 
\\ \nonumber & \quad +   \frac 1{L^3}\sum_{p,r} \hat V(p)\hat{\varphi}(p)\widehat{\chi}_>(p)  \left( \rho_\downarrow  \hat{v}_\uparrow(r)  \hat{a}_{r, \uparrow}^\ast  \hat{a}_{r,\uparrow} + \rho_\uparrow  \hat{v}_\downarrow(r)   \hat{a}^\ast_{r,\downarrow} \hat{a}_{r, \downarrow}  \right)
\\ & \quad +    \frac 1{L^6}\sum_{k,p,r,r'} \hat V(k)\hat{\varphi}(p)\widehat{\chi}_>(p) \hat{v}_\uparrow(k+r+p)  \hat{v}_\downarrow(r'-p-k)  \hat{v}_\uparrow(r) \hat{v}_\downarrow(r')  \hat{a}_{-k-r-p, \uparrow}^\ast   \hat{a}^\ast_{p+k-r',\downarrow} \hat{a}_{-r,\uparrow} \hat{a}_{-r^\prime, \downarrow} \label{eq: aaa*a* in}
\end{align}
where we have used again that $\hat{u}_\uparrow(r+p) =\hat{u}_\downarrow(r'-p) =1$ in all the summands. 
The main contribution is the constant first term.  Since
\[
  \left|\frac{1}{L^3}  \sum_p \hat{V}(p)\hat{\varphi}(p)\widehat{\chi}_>(p) \right| = \left| \int_\Lambda dx\, V(x)\widetilde{\varphi}(x) \right| \leq \|V\|_{L^1(\Lambda)} \|\widetilde{\varphi}\|_{L^\infty(\Lambda)} \leq C,
\]
the term in the second line is bounded by $C \rho \mathcal{N}$. 
To bound the term in the last line in \eqref{eq: aaa*a* in}, we rewrite it in configuration space as
\begin{equation}\label{eq: def IIIb Q2}
-  \int dxdy\, V(x-y)\widetilde{\varphi}(x-y) a^\ast_\downarrow(v_x)a^\ast_\uparrow(v_y) a_\uparrow(v_y)a_\downarrow(v_x) .
\end{equation}
Hence also this term is bounded by $C\rho \mathcal{N}$, using that $\| V \widetilde\varphi\|_1 \leq \|V\|_1 \|\widetilde \varphi\|_\infty \leq C$ because of \eqref{eq: L2 norms phi}. 
%
Combining all the bounds 
the proof of Proposition \ref{pro: Q2} is complete. 
\end{proof}

\subsection{Renormalization of $\mathbb{Q}_2^\udarrow$ by scattering equation}\label{sec: scat eq T1}
In this subsection we will combine the ``$bb + b^\ast b^\ast$ operators'' from Propositions~\ref{pro: H0} and~\ref{pro: Q4}  with $\mathbb{Q}_2^\udarrow$ in order to obtain a renormalized interaction. Recall the definitions of $\mathbb{T}_1$ and $\mathbb{T}_2$ in \eqref{eq: def T1} and in \eqref{eq: def T2}, as well as the ones of  ${\mathbb{Q}}_{2;<}$ and  $\widetilde{\mathbb{Q}}_{2;<}$  in \eqref{eq: definition Q2<} and \eqref{def:q2tl}. Note also that, as all these operators, $\mathbb{Q}_2^\udarrow$ in \eqref{q2ud} is quadratic in the $b$'s, namely,
\[
  \mathbb{Q}_2^\udarrow 
  = \frac{1}{L^3}\sum_{p}\hat{V}(p)b_{p,\uparrow}b_{-p,\downarrow} + \mathrm{h.c}.
\]

\begin{proposition}[Scattering equation cancellation]\label{prop: scatt canc} Let $\gamma\in (0,1/3)$.  Under the same assumptions of Theorem \ref{thm: main up bd},
\begin{equation}\label{eq: canc scatt eq}
\mathbb{T}_1 + \mathbb{T}_2 + \mathbb{Q}_2^\udarrow = {\mathbb{Q}}_{2;<} + \widetilde{\mathbb{Q}}_{2;<} +  \mathcal{E}_{\mathrm{scatt}},
\end{equation}
where $\mathcal{E}_{\mathrm{scatt}}$ is such that for any normalized $\psi \in \mathcal{F}_{\mathrm{f}}$
\begin{equation}\label{eq: error scatt eq canc}
  |\langle \psi, \mathcal{E}_{\mathrm{scatt}} \psi\rangle| \leq  C\left( \rho^{\frac{5}{3} -2\gamma} +  o(1)_{L\to \infty} \right) \langle \psi, \mathcal{N}\psi\rangle + CL^{\frac{3}{2}}\left( \rho^{\frac{13}{6} -\frac{7}{2}\gamma} + o(1)_{L\to \infty} \right) \|\mathcal{N}^{\frac{1}{2}}\psi\|.
\end{equation}
\end{proposition}


\begin{proof} Using  $\widehat{\chi}_> + \widehat{\chi}_< = 1$, a simple computation yields
$$
\mathcal{E}_{\mathrm{scatt}} =  \frac{1}{L^3}\sum_{p} \hat h(p) b_{p,\uparrow}b_{-p,\downarrow} + \mathrm{h.c}
$$
with
\begin{equation}\label{def:h}
\hat h(p) = \hat V(p) -  (\hat{V}\ast_{_{\sum}} \hat\varphi) (p) -8\pi a \widehat\chi_<(p) - 2 p^2 \hat\varphi(p) \widehat\chi_>(p).
\end{equation}
In the following we will show that, as a consequence of the zero-energy scattering equation \eqref{eq: zero energy scatt eq-intro}, $\hat h$ is suitably small (for large $L$ and small $\rho$). We shall write $\hat h=\hat h_1+\hat h_2$, with 
\begin{equation}\label{def:h1}
\hat h_1(p) = \left( \hat V(p) -  (\hat{V}\ast_{_{\sum}} \hat\varphi) (p) - 2 p^2\hat\varphi(p) \right) \widehat\chi_>(p).
\end{equation}
According to \eqref{eq: zero energy scatt eq-intro}, $\hat V(p) - 2 p^2\hat\varphi(p) =(\mathcal{F} (V_\infty \varphi_\infty)  (p) $, hence $\hat h_1$ takes into account the finite-size effect of the periodic box, and vanishes in the thermodynamic limit.
In Lemma \ref{lem: fL}, we shall in fact show that the periodic function $h_1$ with Fourier coefficients $\hat h_1$ satisfies $\|h_1\|_{L^1(\Lambda)} \to 0$ and $\|h_1\|_{L^2(\Lambda)}\to 0$ as $L\to \infty$. 

In configuration space we can write
$$
 \mathrm{I}:= \frac{1}{L^3}\sum_{p} \hat h_1 (p) b_{p,\uparrow}b_{-p,\downarrow} =   \int dxdy\, h_1(x-y)   a_\uparrow(u_x)a_\uparrow(v_x)a_\downarrow(u_y)a_\downarrow(v_y)  = \int dy\,  b_\uparrow((h_1)_y)  a_\downarrow(u_y)a_\downarrow(v_y),
$$
where we used the notation $b_\sigma$ introduced in \eqref{eq: def b phi}. 
With the aid of Lemma \ref{lem: b* phi tilde vs b phi tilde} and the Cauchy--Schwarz inequality, we can estimate 
  \begin{align}\label{eq: est diff convolutions}
\left| \langle  \psi,  \mathrm{I} \psi \rangle \right | 
&\leq \|\hat{v}_\downarrow\|_2\left(\int dy\, \|b_\uparrow((h_1)_y)\psi\|\|a_\downarrow(u_y)\psi\| + CL^{\frac{3}{2}}\rho^{\frac{1}{2}} \|h_1\|_2\|\mathcal{N}^{\frac{1}{2}}\psi\| \right)\nonumber
\\
&\leq \|\hat{v}_\downarrow\|_2\|\hat{v}_\uparrow\|_2\int dxdy\, |h_1(x-y)| \|a_\uparrow(u_x)\psi\|\|a_\downarrow(u_y)\psi\|+ C L^{\frac{3}{2}}\rho \|h_1\|_2\|\mathcal{N}^{\frac{1}{2}}\psi\|\nonumber
\\
&\leq C\rho \|h_1\|_1 \langle \psi, \mathcal{N}\psi \rangle + CL^{\frac{3}{2}}\rho\|h_1\|_2\|\mathcal{N}^{\frac{1}{2}}\psi\|.  
 \end{align}
 This term is thus  $o(L^3)_{L\to \infty}$ in the thermodynamic limit.

We now consider the remaining term
\begin{equation}\label{def:h2}
 \mathrm{II}:= \frac{1}{L^3}\sum_{p} \hat h_2 (p) b_{p,\uparrow}b_{-p,\downarrow} \ , \qquad \hat h_2(p) =  \left( \hat V(p) -  (\hat{V}\ast_{_{\sum}} \hat\varphi) (p) - 8\pi a  \right) \widehat\chi_<(p).
\end{equation}
To start, we note that 
\begin{align*}
   \frac{1}{L^3}\sum_{p}\hat{h}_2(p) b_{p,\uparrow} b_{-p,\downarrow}
  = \frac{1}{L^3}\sum_{p,r,r^\prime}\hat h_2(p)  \hat{u}_\uparrow^<(r+p)\hat{u}_\downarrow^<(r^\prime -p)\hat{v}_\uparrow(r)\hat{v}_{\downarrow}(r^\prime) \hat{a}_{r+p,\uparrow}\hat{a}_{-r,\uparrow}\hat{a}_{r^\prime -p, \downarrow}\hat{a}_{-r^\prime,\downarrow},
\end{align*}
i.e., we can replace  $\hat{u}_\uparrow(r+p)$ and $\hat{u}_\downarrow(r^\prime-p)$ by $\hat{u}^<_\uparrow(r+p)$ and $\hat{u}^<_\downarrow(r^\prime - p)$, with $\hat{u}^<_\sigma$ defined  in \eqref{eq: u<},  due to the support properties of $\widehat\chi_<$ and $\hat v_\sigma$. 
Using also that 
\[
  \sum_{r}\hat{u}_\sigma^<(r+p)\hat{v}_\sigma(r)\hat{a}_{r+p,\sigma}\hat{a}_{-r,\sigma} = \int dx\, e^{-ip\cdot x}a_\sigma(u_x^<)a_\sigma(v_x),
\]
we can write 
$$
\langle  \psi,  \mathrm{II} \psi \rangle
  = \frac{1}{L^3}\sum_p\hat{h}_2(p)\left \langle \left(\int dx\, e^{-ip\cdot x}a^\ast_\uparrow({u}_x^<)a_\uparrow^\ast(v_x)\psi\right), \left( \int dy \, e^{-ip\cdot y} a_\downarrow(v_y)a_\downarrow({u}_y^<)\psi\right)\right\rangle
$$
and hence, by the Cauchy--Schwarz inequality, 
$$
\left| \langle  \psi,  \mathrm{II} \psi \rangle \right| 
\leq \| \hat h_2\|_\infty \left( \int dx\, \| a_\uparrow^\ast ({u}_x^<)a_\uparrow^\ast (v_x)\psi\|^2 \right)^{1/2} \left( \int dy\, \|a_\downarrow(v_y)a_\downarrow({u}^<_y)\psi\|^2\right)^{1/2}.
$$
The last factor can simply be bounded by $\rho^{\frac 12} \langle \psi, \mathcal{N} \psi\rangle^{1/2}$. In the first factor, we first normal and write $a_\uparrow ({u}_x^<) a_\uparrow^\ast ({u}_x^<) = - a_\uparrow^\ast ({u}_x^<)a_\uparrow ({u}_x^<) + \|u^<\|_2^2$, leading to the bound $\rho^{\frac 12} \langle \psi, \mathcal{N} \psi\rangle^{1/2} +\rho^{\frac 12}  \| u^<\|_2 L^{3/2}$. Since $\|u^<\|_2^2 \leq C \rho^{1-3\gamma}$, this yields
$$
\left| \langle  \psi,  \mathrm{II} \psi \rangle \right| 
\leq C \| \hat h_2\|_\infty \left( \rho  \langle \psi, \mathcal{N} \psi\rangle + \rho^{ \frac 32(1-\gamma)} L^{3/2} \langle \psi, \mathcal{N} \psi\rangle^{1/2}\right).
$$
We are left with bounding $\sup_p |\hat h_2(p)|$. Note that $8\pi a = \int_{\R^3} V_\infty ( 1 - \varphi_\infty)= \hat V(0) - \int_{\R^3} V_\infty \varphi_\infty$. Since $V_\infty$ is assumed to have compact support, $| \hat V(p) - \hat V(0)| \leq C |p|^2$. Using $V(x)=V(-x)$ and also \eqref{eq: phiinfty - phi}, we see that
\begin{equation}\label{ca5}
(\hat{V}\ast_{_{\sum}} \hat\varphi) (p) = \int_\Lambda dx\, V(x) \varphi(x) e^{ip\cdot x} = \int_{\R^3} V_\infty(x) \varphi_\infty(x) e^{ip\cdot x} + o(1)_{L\to \infty}
\end{equation}
and hence $|(\hat{V}\ast_{_{\sum}} \hat\varphi) (p) -  \int_{\R^3} V_\infty \varphi_\infty| \leq C |p|^2+ o(1)_{L\to \infty}$. Here note that $\varphi_\infty$ is bounded by Lemma~\ref{lem: prop of phi infty}. 
%
%
Altogether, this shows that $|\hat{h}_2(p)| \leq C|p|^2|\widehat{\chi}_<(p)| + o(1)_{L\to \infty}\leq C\rho^{2/3 - 2\gamma}+ o(1)_{L\to \infty}$, and thus completes the proof.
\end{proof}

\begin{proposition}[Propagation estimates for $\mathbb{Q}_{2;<}$ and $\widetilde{\mathbb{Q}}_{2;<}$]\label{pro: propagation of Q2< and Q2*<}
Let $\lambda \in [0,1]$ and $\gamma\in (0,1/3)$. Under the same assumptions as in Theorem \ref{thm: main up bd}
\begin{align}
 T_{1;\lambda}^\ast (\mathbb{Q}_{2;<} + \widetilde{\mathbb{Q}}_{2; <}) T_{1;\lambda} & = {\mathbb{Q}}_{2;<} + \widetilde{\mathbb{Q}}_{2; <} -2\lambda\rho_\uparrow \rho_\downarrow\sum_{p\in \Lambda^*} \left(8\pi a\widehat{\chi}_<(p) + \hat{V}\ast_{_{\sum}} (\widehat{\chi}_<\hat{\varphi})(p)\right)\hat{\varphi}(p)\widehat{\chi}_>(p) \\ & \quad + \int_0^\lambda d\lambda' \, T^\ast_{1;\lambda'}  \mathcal{E}_{{\mathbb{Q}}_{2;<} + \widetilde{\mathbb{Q}}_{2;<}} T_{1;\lambda'}, \label{eq: prop Q2<}
\end{align}
where the error term $\mathcal{E}_{{\mathbb{Q}}_{2;<} + \widetilde{\mathbb{Q}}_{2;<}}$ is such that for any $\psi \in \mathcal{F}_{\mathrm{f}}$  
\[
  |\langle \psi, \mathcal{E}_{{\mathbb{Q}}_{2;<} + \widetilde{\mathbb{Q}}_{2;<}}\psi \rangle|\leq C\rho^{\frac{4}{3} - \gamma}\langle \psi, \mathcal{N} \psi\rangle + C\rho^{\frac{7}{6} -\frac{\gamma}{2}}\|\mathbb{Q}_4^{\frac{1}{2}} \psi\|\|\mathcal{N}^{\frac{1}{2}}\psi\|. 
\]
Furthermore,  for any normalized $\psi \in \mathcal{F}_{\mathrm{f}}$, we have
\begin{equation}\label{eq: est Q2*<}
  |\langle \psi, \widetilde{\mathbb{Q}}_{2;<}\psi\rangle|  \leq CL^{3/2}\rho^{\frac{4}{3} -\gamma}  \| \mathbb{Q}_4^{\frac 12}\psi\|,
\end{equation}
\begin{equation}\label{eq: est Q2< prop est}
 |\langle \psi, {\mathbb{Q}}_{2;<}\psi\rangle|  \leq CL^{\frac{3}{2}}\rho^{\frac{3}{2}(1-\gamma)}\|\mathcal{N}^{\frac{1}{2}}\psi\|.
\end{equation}
\end{proposition}


\begin{proof}[Proof of Proposition \ref{pro: propagation of Q2< and Q2*<}]
  By Duhamel's formula, 
\begin{equation}\label{eq: Duhamel Q2< + tildeQ2<}
    T_{1;\lambda}^\ast ({\mathbb{Q}}_{2;<} + \widetilde{\mathbb{Q}}_{2;<})  T_{1;\lambda} 
    = {\mathbb{Q}}_{2;<} + \widetilde{\mathbb{Q}}_{2;<}   + \int_0^\lambda d\lambda^\prime\,  T_{1;\lambda^\prime}^\ast [{\mathbb{Q}}_{2;<} + \widetilde{\mathbb{Q}}_{2;<}, B_1 - B_1^\ast] T_{1;\lambda^\prime}.
\end{equation}
We start by computing $[{\mathbb{Q}}_{2;<}, B_1 -B_1^\ast]$. The terms  have the same structure as in Proposition \ref{pro: Q2}, the only difference is that $\hat V(k)$ has to be replaced by $8\pi a \widehat \chi_<(k)$ everywhere. We again write them as 
\[
  [{\mathbb{Q}}_{2;<}, B_1 -B_1^\ast] = \sum_{j=1}^6 \mathrm{I}_j + \mathrm{h.c.}.
\]
As in Proposition \ref{pro: Q2}, putting $\mathrm{I}_6$ in normal order, we get the constant term. All the other terms are error terms, which we estimate in expectation in a state  $\psi \in \mathcal{F}_{\mathrm{f}}$. In some error terms it  will again be convenient to replace $\hat{u}_\sigma$ by $\hat{u}^<_\sigma$, with $\hat{u}^<_\sigma$ defined  in \eqref{eq: u<}, which is  possible by the compact support of $\widehat \chi_<$. In the estimates below, we often use the bounds in Lemma \ref{lem: bounds b b*}, Lemma \ref{lem: bounds phi} and Lemma \ref{lem:cut-off-chi} together with \eqref{bound:a} and \eqref{eq: est a v H0}. 
%

For $\mathrm{I}_1$, we have, analogously to \eqref{46:I1},
\begin{equation}\label{48:I1}
  \mathrm{I}_1 = 
  -\frac{8\pi a}{L^3}\sum_{r}\hat{v}_\uparrow(r)\int dxdydz\, e^{ir\cdot (x-z)} \chi_<(x-y)a^\ast_\downarrow(u^<_y)a^\ast_\uparrow(u^<_x)a^\ast_\downarrow(v_y)b_\downarrow(\widetilde{\varphi}_{z})a_\uparrow(u_z).
\end{equation}
%
Using $0\leq v_\uparrow(r)\leq 1$ and the Cauchy--Schwarz inequality, we can bound 
\begin{align}\label{eq: est I1 Q2< T1}
 |\langle \psi, \mathrm{I}_1\psi\rangle | &\leq 
  8\pi a \left( \int dx  \left\| \int dy\, \chi_<(x-y)  a_\downarrow(v_y)a_\uparrow(u_x^<)a_\downarrow(u_y^<)\psi\right\|^2 \right)^{1/2}
  \left(  \int dz \left\|  b_\downarrow(\widetilde{\varphi}_z) a_\uparrow(u_z)\psi\right\|^2\right)^{1/2} \nonumber
  \\
  &\leq  8\pi a \|\hat{v}_\downarrow\|_2 \|\hat{u}^<_\downarrow\|_2 \|\chi_<\|_1  \rho^{\frac{1}{3} +\frac\gamma 2} \langle \psi, \mathcal{N}\psi\rangle   \leq  C\rho^{\frac{4}{3} -\gamma}\langle \psi, \mathcal{N}\psi\rangle.
\end{align}
The  term $\mathrm{I}_2$ can be estimated in the same way, with the same result.
%

 Next we consider $\mathrm{I}_3$. In analogy with \eqref{46:I3}, it is given by 
$$
  \langle \psi, \mathrm{I}_3\psi\rangle 
  = 8\pi a\int dxdydz^\prime\, \chi_<(x-y) v_\downarrow(y;z^\prime)\left\langle a_\uparrow(u^<_x)a_\downarrow(u^<_y)\psi, \left(\int dz\, \widetilde{\varphi}(z-z^\prime) v_\uparrow(x;z)a_\uparrow (u_z)\right)a_\downarrow(u_{z^\prime})\psi\right\rangle.
$$
Using  \eqref{eq: est conv L2} and the Cauchy--Schwarz inequality, we obtain the bound
\begin{multline}\label{eq: est Ib Q2<}
  |\langle \psi, \mathrm{I}_3\psi\rangle| \leq 8\pi a\left(\int dxdydz^\prime |\chi_<(x-y)|  |v_\downarrow(y;z^\prime)|^2 \|a_\uparrow(u_x^<)a_\downarrow(u_y^<)\psi\|^2\right)^{\frac{1}{2}}
  \\
  \times\left(\int dxdydz^\prime dw\, |\chi_<(x-y)|  |\widetilde{\varphi}(z^\prime -w)|^2|v_\uparrow(x; w)|^2\|a_\downarrow(u_{z^\prime})\psi\|^2\right)^{\frac{1}{2}}
  \\
  \leq 8\pi a \|\chi_<\|_1\|v_\uparrow\|_2\|v_\downarrow\|_2  \|u^<_\uparrow\|_2\|\widetilde{\varphi}\|_2\langle \psi, \mathcal{N}\psi\rangle \leq C\rho^{\frac{4}{3} - \gamma}\langle \psi, \mathcal{N}\psi\rangle,
\end{multline}
where we used the bounds mentioned at the beginning of the proof.

The term $I_4$ can be written, similarly to \eqref{46:I4}, as
$$
  \mathrm{I}_4 = 8\pi a \int dxdydz\, \chi_<(x-y)\widetilde{\varphi}(y-z)a^\ast_\downarrow(u_x)a_\downarrow(u_z)a_\uparrow(v_y)a_\downarrow(v_z)a^\ast_\downarrow(v_x) a^\ast_\uparrow(v_y).
$$
%
Hence
\begin{equation}\label{eq: est I4 Q2< T1}
  |\langle \psi, \mathrm{I}_4\psi\rangle| 
  \leq C\rho^2 \|\chi_<\|_1\|\widetilde{\varphi}\|_1 \langle \psi, \mathcal{N}\psi\rangle = C\rho^{\frac{4}{3} +2\gamma}\langle \psi, \mathcal{N}\psi\rangle, 
\end{equation}
where used that $a_\downarrow(u_x)$ commutes with $ a_\uparrow(v_y)a_\downarrow(v_z)a^\ast_\downarrow(v_x) a^\ast_\uparrow(v_y)$, as well as Lemmas~\ref{lem: bounds phi} and~\ref{lem:cut-off-chi} in the last step. The term $\mathrm{I}_5$ can be estimated in the same way, with the same result.

Finally we consider the last term $\mathrm{I}_6$ coming from the commutator $[\mathbb{Q}_{2;<}, B_1 - B_1]$, which equals \eqref{eq: aaa*a* in} with $\hat V$ replaced by $8\pi a \widehat\chi_<$. The term in the first line is the desired constant. For the term in the second line, we use 
\begin{equation}\label{eq: phi tilde <}
  \left|\frac{1}{L^3}  \sum_p \widehat\chi_<(p) \hat{\varphi}(p) \widehat{\chi}_>(p)\right| \leq \frac{C}{L^3}\sum_{p\neq 0} \frac{\widehat\chi_<(p)}{|p|^2} \leq C\rho^{\frac{1}{3} -\gamma},
\end{equation}
which follows from the fact that $p^2 \hat\varphi(p)$ is bounded, as a consequence of \eqref{eq: zero energy scatt eq-intro}. Hence the expectation value of this term is bounded by $C\rho^{\frac{4}{3} - \gamma}\langle \psi, \mathcal{N}\psi\rangle$. In the last term corresponding to the third line of \eqref{eq: aaa*a* in}, 
we observe that the summation over $p$ is restricted to $|p| \leq 7\rho^{1/3-\gamma}$, due to the constraints on $k, r,r^\prime$ and $r+p+k, r^\prime - p-k$ in the summands.
In other words, we can add for free the cut-off function $\widetilde{\chi}$ given by 
\[
  \widetilde{\chi}(p)= \begin{cases} 1 &\mbox{if}\,\,\, |p| \leq 7 \rho^{\frac{1}{3} - \gamma}, \\ 0 &\mbox{if}\,\,\, |p| > 7\rho^{\frac{1}{3} - \gamma},\end{cases}
\]
arriving at
%
\[
  \frac{8\pi a}{L^6}\sum_{k,p,r,r^\prime}\widehat{\chi}_<(k)\hat{\varphi}(p)\widehat{\chi}_>(p)\widetilde{\chi}(p)\hat{v}_\uparrow(r)\hat{v}_\downarrow(r^\prime)\hat{v}_\uparrow(k+r+p)\hat{v}_\downarrow(r^\prime - p -k) \hat{a}_{-k-r-p, \uparrow}^\ast   \hat{a}^\ast_{p+k-r',\downarrow} \hat{a}_{-r,\uparrow} \hat{a}_{-r^\prime, \downarrow} .
\]
Introducing the function $\widetilde{\varphi}_<$ with Fourier coefficients $\widehat{\chi}_>(p)\hat{\varphi}(p)\widetilde{\chi}(p)$ and rewriting the term in configuration space as in \eqref{eq: def IIIb Q2}, we obtain the bound $C \rho \|\chi_< \widetilde\varphi_<\|_1 \mathcal{N}$. We have $\|\chi_< \widetilde\varphi_<\|_1 \leq \| \chi_< \|_1 \| \widetilde\varphi_<\|_\infty \leq C \rho^{\frac 13 -\gamma}$ using the same bound as in \eqref{eq: phi tilde <}. Altogether, we have thus shown that
$$
  [{\mathbb{Q}}_{2;<}, B_1 - B_1^\ast] =   - 16\pi a  \rho_\uparrow \rho_\downarrow \sum_{p} \widehat\chi_<(p)\hat{\varphi}(p)\widehat{\chi}_>(p) 
   + \mathcal{E}_{{\mathbb{Q}}_{2;<}},
$$
with $\mathcal{E}_{{\mathbb{Q}}_{2;<}}$ such that 
$$
  |\langle \psi, \mathcal{E}_{{\mathbb{Q}}_{2;<}}\psi\rangle|\leq  C\rho^{\frac{4}{3} -\gamma}\langle \psi, \mathcal{N}\psi\rangle.
$$


We now consider $[\widetilde{\mathbb{Q}}_{2;<}, B_1 - B_1^\ast]$.  Again all the terms in the commutator are of the same type as the ones in $[{\mathbb{Q}}_{2}, B_1 - B_1^\ast]$ in Proposition \ref{pro: Q2}, now with $\hat{V}$ replaced by $\hat{V}\ast_{_{\sum}} (\widehat{\chi}_< \hat{\varphi})$. They can then be treated exactly in the same way as the corresponding error terms in Proposition \ref{pro: Q2}, 
using that 
\begin{equation}\label{eq: def phi<}
\sup_{x\in\Lambda} \left|   \frac{1}{L^3}\sum_{p}\hat{\varphi}(p)\widehat{\chi}_<(p) e^{ip\cdot x} \right| \leq C \rho^{\frac 13-\gamma}
\end{equation}
as in \eqref{eq: phi tilde <}. 
%
We thus get that
$$
  [\widetilde{\mathbb{Q}}_{2;<}, B_1 - B_1^\ast] =  -2\rho_\uparrow \rho_\downarrow\sum_{p\in \Lambda^*} \hat{V}\ast_{_{\sum}} (\widehat{\chi}_<\hat{\varphi})(p)\hat{\varphi}(p)\widehat{\chi}_>(p) + \mathcal{E}_{\widetilde{\mathbb{Q}}_{2;<}},
$$
with $\mathcal{E}_{\widetilde{\mathbb{Q}}_{2;<}}$ such that 
$$
  |\langle \psi, \mathcal{E}_{\widetilde{\mathbb{Q}}_{2;<}}\psi\rangle|\leq C\rho^{\frac{7}{6} -\frac{\gamma}{2}}\|\mathbb{Q}_{4}^{\frac{1}{2}}\psi\|\|\mathcal{N}^{\frac{1}{2}}\psi\|  + C\rho^{\frac{4}{3} -\gamma}\langle \psi, \mathcal{N}\psi\rangle.
$$
This proves \eqref{eq: prop Q2<}.

Finally we discuss the bounds in \eqref{eq: est Q2*<} and \eqref{eq: est Q2< prop est}. Rewriting $\widetilde{\mathbb{Q}}_{2;<}$ in configuration space, we have
\[
  \widetilde{\mathbb{Q}}_{2;<} = \int dxdy \, V(x-y)\varphi_<(x-y)a_\uparrow(u_x)a_\uparrow(v_x)a_\downarrow(u_y)a_\downarrow(v_y) + \mathrm{h.c.},
\]
where $\varphi_<$ is the function with Fourier coefficients $\hat\varphi(p)\widehat\chi_<(p)$.   The bound in  \eqref{eq: est Q2*<} is then a consequence of $\|\varphi_<\|_\infty \leq C\rho^{1/3 - \gamma}$ as shown in \eqref{eq: def phi<}. Indeed, via the Cauchy--Schwarz inequality, we get 
\[
  |\langle\psi, \widetilde{\mathbb{Q}}_{2;<}\psi\rangle| \leq C\rho^{\frac{4}{3} - \gamma}\int dxdy\, V(x-y) \|a_\uparrow(u_x) a_\downarrow(u_y) \psi\| \leq  CL^{3/2}\rho^{\frac{4}{3} -\gamma}  \| \mathbb{Q}_4^{\frac 12}\psi\| .
\]
Similarly, we can write $ {\mathbb{Q}}_{2;<}$ in configuration space as
$$
 {\mathbb{Q}}_{2;<}  
 = {8\pi a} \int dx dy \,\, {\chi}_<(x-y) a_{\uparrow}({u}_{x}^< ) a_{\downarrow} ({u}_y^<) a_{\downarrow}( {v}_y)   a_{\uparrow} ( {v}_x) + \mathrm{h.c.} 
 $$
where we again could replace  $\hat{u}_\sigma$ by $\hat{u}^<_\sigma$ defined in \eqref{eq: u<}. 
We then find
\begin{equation}\label{eq:T2-H0-err2}
  |\langle \psi, {\mathbb{Q}}_{2;<} \psi \rangle |
  \leq 16\pi a  \|u^<_\uparrow\|_2 \|v_\uparrow\|_2 \|v_\downarrow\|_2  
  \int dxdy\,|\chi_<(x-y)|\|a_\downarrow(u^<_y) \psi\|
   \leq  CL^{\frac{3}{2}}\rho^{\frac{3}{2}(1-\gamma)} \|\mathcal{N}^{\frac{1}{2}}\psi\| .
\end{equation}
This completes the proof of Proposition \ref{pro: propagation of Q2< and Q2*<}.
\end{proof}


\subsection{Propagation of the estimates for $\mathbb{H}_0$ and $\mathbb{Q}_4$}\label{sec: prop est H0 Q4}

In this subsection we obtain propagation estimates for $\mathbb{H}_0$ and $\mathbb{Q}_4$, which will be useful to estimate some of the error terms.

\begin{proposition}[Propagation of the estimates for $\mathbb{H}_0$ and $\mathbb{Q}_4$]\label{pro: propagation estimates} Let $\lambda \in [0,1]$, $\gamma \in (0,1/6)$, $\delta \in (0,8\gamma]$ with $2\gamma + \frac{\delta}{16}\leq \frac{1}{3}$. Under the same assumptions as in Theorem \ref{thm: main up bd}
\[
  \langle T_{1;\lambda}T_2\Omega, (\mathbb{H}_0 + \mathbb{Q}_4)T_{1;\lambda}T_2\Omega \rangle \leq C \langle T_2\Omega, (\mathbb{H}_0 + \mathbb{Q}_4)  T_2\Omega \rangle + CL^3\rho^2 + o(L^3)_{L\to \infty}. 
\]
\end{proposition}

Note that since $\mathbb{H}_0$ and $\mathbb{Q}_4$ are nonnegative, the above estimate gives two separate bounds for $\langle T_{1;\lambda}T_2\Omega, \mathbb{H}_0 T_{1;\lambda}T_2\Omega \rangle $ and $\langle T_{1;\lambda}T_2\Omega,  \mathbb{Q}_4 T_{1;\lambda}T_2\Omega \rangle$. The term $\langle T_2\Omega, (\mathbb{H}_0 + \mathbb{Q}_4) T_2\Omega \rangle$ on the right-hand side is actually small compared to  $L^3\rho^2$, as will be proved later in Propositions~\ref{pro: prop est H0 T2} and~\ref{pro: prop est Q4 T2}.

\begin{proof} Let  $\xi_\lambda=T_{1;\lambda}T_2\Omega$.  We prove the propagation estimate by Gr\"onwall's Lemma. By Propositions~\ref{pro: H0} and~\ref{pro: Q4}, we have
\begin{equation}\label{eq: der h0 q4 prop est }
 \partial_\lambda \langle \xi_\lambda, (\mathbb{H}_0 + \mathbb{Q}_4) \xi_\lambda\rangle 
 =   \langle\xi_\lambda, (\mathbb{T}_1 + \mathbb{T}_2+  \mathcal{E}_{\mathbb{H}_0}+ \mathcal{E}_{\mathbb{Q}_4})\xi_\lambda\rangle, 
\end{equation}
where the four terms on the right-hand side are defined in \eqref{eq: def T1}, \eqref{eq: def T2}, \eqref{eq: error kin energy} and \eqref{eq: conj Q4 T1}, respectively.  Eq.~\eqref{eq: est err Q4 after  T1} in Proposition \ref{pro: Q4} gives a bound on $\mathcal{E}_{\mathbb{Q}_4}$. 
For the term  $\mathcal{E}_{\mathbb{H}_0} $ in  \eqref{eq: error kin energy}, we first consider 
the terms involving $k\cdot s$ in \eqref{def:E1} and write, as in the proof of Proposition~\ref{pro: H0}, $k=(k+s) - s$.
Correspondingly, writing the resulting expression in configuration space and recalling the definition of $b_\sigma$ in \eqref{eq: def b phi}, we find 
\begin{equation}\label{a2t}
 -\frac{2}{L^3}\sum_{k,s,s^\prime \in \Lambda^*} k\cdot s \hat{\varphi}(k)\widehat{\chi}_>(k)\, b_{k,s,\uparrow}b_{-k,s^\prime, \downarrow} 
 =  2\sum_{\ell=1}^3\int dx\,  b_\downarrow(\widetilde\varphi_x) \left( a_\uparrow(\partial_\ell v_x)a_\uparrow(\partial_\ell \nu^>_x) - a_\uparrow(\partial_\ell^2 v_x) a_\uparrow(u_x)\right)    
\end{equation}
where in the first term we again replaced $u_\uparrow$ by $\nu^>$ defined in \eqref{eq: def nu >}, which is possible due the presence of the cutoff $\chi_>$. 
Using Lemma \ref{lem: bounds b b*} together with the bounds in \eqref{eq: est a v H0}, we can estimate via the Cauchy--Schwarz inequality
\begin{equation}\label{ato}
\int\, dx\,  \|b_\downarrow(\widetilde{\varphi}_x)a_\uparrow(\partial_j^2 v_x)a_\uparrow(u_x)\xi_\lambda\|
\leq CL^{\frac{3}{2}}\rho^{\frac{3}{2} +\frac{\gamma}{2}}\|\mathcal{N}^{\frac{1}{2}}\xi_\lambda\| .
\end{equation}
For the first term on the right-hand side of \eqref{a2t}, we can proceed as in \eqref{eq: est partial u H0} to bound
\begin{equation}\label{eq: est Err kin prop est}
  \int\, dx\,  \|b_\downarrow(\widetilde{\varphi}_x)a_\uparrow(\partial_\ell v_x)a_\uparrow(\partial_\ell \nu^>_x)\xi_\lambda\| 
  \leq CL^{\frac{3}{2}}\rho^{\frac{7}{6}  + \frac{\gamma}{2}}   \langle \xi_\lambda, \mathbb{H}_0 \xi_\lambda\rangle^{1/2} .
\end{equation}
For the terms  involving $k\cdot s'$  in $\mathcal{E}_{\mathbb{H}_0} $ in  \eqref{eq: error kin energy} we can proceed in the same way, obtaining the same bound.

We now consider the expectation of $\mathbb{T}_1 + \mathbb{T}_2$ in \eqref{eq: der h0 q4 prop est }. From Proposition \ref{prop: scatt canc}, we know that 
$\mathbb{T}_1 + \mathbb{T}_2 = - \mathbb{Q}_2^\udarrow + {\mathbb{Q}}_{2;<} + \widetilde{\mathbb{Q}}_{2;<} +  \mathcal{E}_{\mathrm{scatt}}$,
%
with  $\mathcal{E}_{\mathrm{scatt}}$ bounded as in \eqref{eq: error scatt eq canc}, 
and Eqs.~\eqref{eq: est Q2*<} and~\eqref{eq: est Q2< prop est}  from Proposition \ref{pro: propagation of Q2< and Q2*<} give a bound on $\mathbb{Q}_{2;<} + \widetilde{\mathbb{Q}}_{2;<}$. 
It thus remains to give a bound on $ \mathbb{Q}_2^\udarrow$. By the Cauchy--Schwarz inequality and \eqref{eq: est a v H0}, 
$$
  |\langle \xi_\lambda, \mathbb{Q}_2^\udarrow \xi_\lambda \rangle| \leq C\rho\int dxdy\, V(x-y) \|a_\uparrow(u_x)a_\downarrow(u_y)\xi_\lambda\| \leq CL^{\frac{3}{2}}\rho \|\mathbb{Q}_4^{\frac{1}{2}}\xi_\lambda\| \leq CL^3\rho^2 + \langle \xi_\lambda, \mathbb{Q}_4 \xi_\lambda\rangle
$$
where we used $V \ge 0$ and $\|V\|_{L^1(\Lambda)} \le C$. Therefore,  
\begin{align} \nonumber
  |\langle \xi_\lambda, \left(\mathbb{T}_1 + \mathbb{T}_2\right)\xi_\lambda\rangle | & \leq CL^3 \rho^{2} + C\langle \xi_\lambda, \mathbb{Q}_4 \xi_\lambda\rangle + 
  C \left( \rho^{\frac{5}{3} -2\gamma} +  o(1)_{L\to \infty} \right)
 \langle \xi_\lambda, \mathcal{N}\xi_\lambda\rangle \\ & \quad + CL^{\frac{3}{2}} \left( \rho^{\frac 32(1-\gamma)} + \rho^{\frac{13}{6} -\frac{7}{2}\gamma}+  o(1)_{L\to \infty} \right)\|\mathcal{N}^{\frac{1}{2}}\xi_\lambda\| .  \label{eq: est T1+T2 prop}
\end{align}
Inserting the bounds \eqref{eq: est err Q4 after  T1}, \eqref{ato}, \eqref{eq: est Err kin prop est} and \eqref{eq: est T1+T2 prop} in \eqref{eq: der h0 q4 prop est }, we find  
\begin{align*}
  \left| \partial_\lambda \langle \xi_\lambda, (\mathbb{H}_0 + \mathbb{Q}_4)\xi_\lambda\rangle\right| 
   & \leq  C\langle \xi_\lambda, \left(\mathbb{H}_0 + \mathbb{Q}_4 \right) \xi_\lambda\rangle  + CL^3\rho^2 +   C \left( \rho^{\frac{5}{3} -2\gamma} +  o(1)_{L\to \infty} \right)
\langle \xi_\lambda, \mathcal{N}\xi_\lambda\rangle \\ & \quad + CL^{\frac{3}{2}}\left(\rho^{\frac{3}{2}(1-\gamma)}  + \rho^{\frac{13}{6} -\frac{7}{2}\gamma}+ o(1)_{L\to \infty} \right)\|\mathcal{N}^{\frac{1}{2}}\xi_\lambda\|  . 
\end{align*}
Using the estimate for the number operator in \eqref{eq: est number operator T1 final} together with the constraints on the parameters  $\gamma <1/3$, $\delta \leq  8\gamma$ and $ 2\gamma + \frac{\delta}{16} \leq \frac{1}{3}$, we conclude that 
\[
\left|  \partial_\lambda \langle \xi_\lambda, (\mathbb{H}_0 + \mathbb{Q}_4) \xi_\lambda\rangle\right| \leq  C \langle \xi_\lambda, (\mathbb{H}_0 + \mathbb{Q}_4)\xi_\lambda\rangle + CL^3\rho^2 + o(L^3)_{L\to \infty}. 
\]
The desired result then follows from Gr\"onwall's Lemma.
\end{proof}


\subsection{Conclusion of Proposition \ref{lem:conjugation-T1}}\label{sec: T1}\label{sec: conj undert T1}

Now we have all the necessary ingredients  in order to give the proof of Proposition \ref{lem:conjugation-T1}. 
Recall the definition of $\mathcal{H}^{\mathrm{eff}}_{\mathrm{corr}}= \mathbb{H}_0 + \mathbb{Q}_2^\udarrow + \mathbb{Q}_4$ in \eqref{eq: def eff corr en strategy}.  
We start by writing
\begin{equation}\label{eq: prop h0 up bd}
  \langle T_1T_2\Omega, \left( \mathbb{H}_0  + \mathbb{Q}_4 \right) 
  T_1T_2\Omega\rangle = \langle T_2\Omega,  \left( \mathbb{H}_0  + \mathbb{Q}_4 \right)
  T_2\Omega\rangle + \int_0^1 d\lambda\, \partial_\lambda \langle T_{1;\lambda}T_2\Omega,  \left( \mathbb{H}_0  + \mathbb{Q}_4 \right)
  T_{1;\lambda}T_2\Omega\rangle . 
\end{equation}
From Proposition \ref{pro: H0} and Remark \ref{rem: Duhamel err kin energy}, 
\begin{multline}\label{eq: duhamel H0 trial state}
  \int_0^1 d\lambda\, \partial_\lambda \langle T_{1;\lambda}T_2\Omega, \mathbb{H}_0 T_{1;\lambda}T_2\Omega\rangle 
  \\
  = \int_0^{1} d\lambda\, \langle T_{1;\lambda}T_2\Omega, \mathbb{T}_1 T_{1;\lambda}T_2\Omega\rangle + \langle T_2\Omega, \mathcal{E}_{\mathbb{H}_0} T_2\Omega\rangle  + \int_0^1 d\lambda  \, (1-\lambda)   \partial_{\lambda}\langle  T_{1;\lambda}T_2\Omega, \mathcal{E}_{\mathbb{H}_0}T_{1;\lambda}T_2\Omega\rangle,
\end{multline}
where $ \mathbb{T}_1$ is defined in \eqref{eq: def T1}, 
and the last error term can be controlled by combining the bound \eqref{eq: est err H0} from Proposition \ref{pro: H0} with Propositions~\ref{pro: est number op} and~\ref{pro: propagation estimates} as
\begin{align}\nonumber
  \left| \partial_{\lambda}\langle  T_{1;\lambda}T_2\Omega, \mathcal{E}_{\mathbb{H}_0}T_{1;\lambda}T_2\Omega\rangle\right| &\leq C\rho^{1+\gamma}\|\mathbb{H}_0^{\frac{1}{2}}T_{1;\lambda}T_2\Omega\|\|\mathcal{N}^{\frac{1}{2}}T_{1;\lambda}T_2\Omega\|   + C\rho^{\frac 43+\gamma} \langle T_{1;\lambda}T_2\Omega, \mathcal{N} T_{1;\lambda}T_2\Omega\rangle \nonumber
  \\
    &\leq C L^{3/2} \rho^{ \frac {17}6 + \frac \gamma 2 - \frac \delta{16} }\| ( \mathbb{H}_0 + \mathbb{Q}_4)^{\frac{1}{2}}T_2\Omega\|  + C L^3 \left( \rho^{\frac {17}6 + \frac \gamma 2 - \frac \delta{16}}     +  \rho^{3 -\frac \delta 8 }\right) + o(L^3)_{L\to \infty}. 
   \label{eq: est err H0 up bd}
\end{align}
From Proposition \ref{pro: Q4} we get 
\begin{equation}\label{eq: prop Q4 up bd}
  \int_0^1 d\lambda\, \partial_\lambda \langle T_{1;\lambda}T_2\Omega, \mathbb{Q}_4 T_{1;\lambda}T_2\Omega\rangle = 
  \int_0^1 d\lambda\, \langle T_{1;\lambda}T_2\Omega, \left( \mathbb{T}_2 + \mathcal{E}_{\mathbb{Q}_{4}} \right)  T_{1;\lambda}T_2\Omega\rangle 
\end{equation}
where $\mathbb{T}_2$ is defined in \eqref{eq: def T2}
and the error term  $\mathcal{E}_{\mathbb{Q}_{4}}$ can be controlled, again with Proposition~\ref{pro: est number op} and~\ref{pro: propagation estimates}, as
\begin{align}\label{eq: est err Q4 after  T1 XX}
  |\langle T_{1;\lambda}T_2\Omega,  \mathcal{E}_{\mathbb{Q}_{4}}  T_{1;\lambda}T_2\Omega\rangle| &\leq  C\rho^{\frac{5}{6} +\frac{\gamma}{2}}\|\mathcal{N}^{\frac{1}{2}}T_{1;\lambda}T_2\Omega\|\|\mathbb{Q}^{\frac{1}{2}}_4 T_{1;\lambda}T_2\Omega\| \nonumber\\
  &\le C \rho^{\frac 53 - \frac \delta{16}} L^{3/2} \left( \|(\mathbb{H}_0+\mathbb{Q}_4)^{\frac{1}{2}} T_2\Omega\| + L^{3/2} \rho + o(L^{3/2})_{L\to \infty} \right) . 
\end{align}
The terms $\mathbb{T}_1$ and $\mathbb{T}_2$ can be combined as in Proposition \ref{prop: scatt canc} to $\mathbb{T}_1 + \mathbb{T}_2 = - \mathbb{Q}_2^\udarrow + {\mathbb{Q}}_{2;<} + \widetilde{\mathbb{Q}}_{2;<} +  \mathcal{E}_{\mathrm{scatt}}$
%
with $\mathbb{Q}_{2;<}$ and $\widetilde{\mathbb{Q}}_{2;<}$ defined in \eqref{eq: definition Q2<} and \eqref{def:q2tl}, and $\mathcal{E}_{\mathrm{scatt}}$ satisfying
\begin{equation}\label{eq: est err scatt eq up bd}
 | \langle T_{1;\lambda}T_2\Omega,\mathcal{E}_{\mathrm{scatt}}T_{1;\lambda}T_2\Omega\rangle|  
 \leq CL^3\rho^{3 -4\gamma - \frac{\delta}{16}} + o(L^3)_{L\to \infty},
\end{equation}
where we used the estimate in \eqref{eq: error scatt eq canc} together with Proposition \ref{pro: est number op}.
Therefore, combining \eqref{eq: prop h0 up bd}--\eqref{eq: est err scatt eq up bd}, we obtain 
\begin{align}\nonumber
  \langle T_1 T_2\Omega, \mathcal{H}^{\mathrm{eff}}_{\mathrm{corr}}T_1T_2\Omega\rangle 
 &= \langle T_2\Omega, (\mathbb{H}_0 + \mathbb{Q}_4 + \mathcal{E}_{\mathbb{H}_0})T_2\Omega\rangle 
 + \int_0^1 d\lambda\, \langle T_{1;\lambda}T_2\Omega, ({\mathbb{Q}}_{2;<} + \widetilde{\mathbb{Q}}_{2;<}) T_{1;\lambda}T_2\Omega \rangle 
  \\
  &  \quad + \langle T_1 T_2\Omega, \mathbb{Q}_2^\udarrow T_1 T_2\Omega \rangle  -\int_0^1 d\lambda\, \langle T_{1;\lambda}T_2\Omega, \mathbb{Q}_2^\udarrow T_{1;\lambda}T_2\Omega\rangle+ \mathcal{E}  
\label{eq: up db after prop}
\end{align}
with 
\begin{equation}\label{eq: err after prop T1 up bd}
|\mathcal{E}| \leq    C \rho^{\frac 53 - \frac \delta{16}} L^{3/2}  \|(\mathbb{H}_0+\mathbb{Q}_4)^{\frac{1}{2}} T_2\Omega\|  + C L^{3} \rho^{\frac 83 - \frac \delta{16}} +CL^3\rho^{3 -4\gamma - \frac{\delta}{16}}   +  o(L^3)_{L\to \infty}.  
\end{equation}
%

The terms involving $ \mathbb{Q}_2^\udarrow $ in the second line on the right-hand of \eqref{eq: up db after prop} above can be written as 
$$
 \langle T_1 T_2\Omega, \mathbb{Q}_2^\udarrow T_1 T_2\Omega \rangle  -\int_0^1 d\lambda\, \langle T_{1;\lambda}T_2\Omega, \mathbb{Q}_2^\udarrow T_{1;\lambda}T_2\Omega\rangle =
 \int_0^1d\lambda\, \lambda \,\partial_\lambda \langle  T_{1;\lambda}T_2\Omega, \mathbb{Q}_2^\udarrow T_{1;\lambda}T_2\Omega\rangle ,
$$
which, according to  Proposition \ref{pro: Q2}, equals
\begin{equation}\label{eq: prop Q2 up bd}
 -\rho_\uparrow \rho_\downarrow \sum_{p\in\Lambda^*} \hat{V}(p)\hat{\varphi}(p)\widehat{\chi}_>(p) 
  + \int_0^1 d\lambda\, \lambda\, \langle T_{1;\lambda}T_2\Omega,\mathcal{E}_{\mathbb{Q}_2^\udarrow}T_{1;\lambda}T_2\Omega\rangle 
\end{equation}
with 
\begin{equation}\label{eq: est err Q2 up bd}
  |\langle T_{1;\lambda}T_2\Omega,\mathcal{E}_{\mathbb{Q}_2^\udarrow}T_{1;\lambda}T_2\Omega\rangle| \leq CL^3\rho^{\frac{8}{3} - \gamma -\frac{\delta}{8}} + C \rho^{\frac 53 - \frac \delta{16}} L^{3/2}  \|(\mathbb{H}_0+\mathbb{Q}_4)^{\frac{1}{2}} T_2\Omega\|    + o(L^{3})_{L\to \infty} 
\end{equation}
for all $\lambda \in [0,1]$, where we used again Proposition \ref{pro: est number op} and Proposition \ref{pro: propagation estimates}.
%
%
Next we consider the contribution coming from $\mathbb{Q}_{2;<} + \widetilde{\mathbb{Q}}_{2;<}$ in the first line on the right-hand side of \eqref{eq: up db after prop}. 
By Proposition \ref{pro: propagation of Q2< and Q2*<},  
\begin{align*}
  \int_0^1 d\lambda\, \langle T_{1;\lambda}T_2\Omega, ({\mathbb{Q}}_{2;<}  &+ \widetilde{\mathbb{Q}}_{2;<})T_{1;\lambda}T_2\Omega\rangle  =  \langle T_2\Omega, ({\mathbb{Q}}_{2;<} + \widetilde{\mathbb{Q}}_{2;<})T_2\Omega\rangle  - 8\pi a \rho_\uparrow\rho_\downarrow \sum_{p\in\Lambda^*}\widehat{\chi}_<(p)\hat{\varphi}(p)\widehat{\chi}_>(p) 
  \\
   &-  \rho_\uparrow\rho_\downarrow\sum_{p \in\Lambda^*} \hat{V}\ast_{_{\sum}}(\widehat{\chi}_<\hat{\varphi})(p)\hat{\varphi}(p)\widehat{\chi}_>(p)
  + \int_0^1 d\lambda \,(1-\lambda) \langle  T_{1;\lambda}T_2\Omega, \mathcal{E}_{{\mathbb{Q}}_{2;<} + \widetilde{\mathbb{Q}}_{2;<}} T_{1;\lambda}T_2\Omega\rangle
\end{align*}
with 
$$
  |\langle  T_{1;\lambda}T_2\Omega, \mathcal{E}_{{\mathbb{Q}}_{2;<} + \widetilde{\mathbb{Q}}_{2;<}}T_{1;\lambda}T_2\Omega\rangle| \leq C L^3 \rho^{3-2\gamma-\frac \delta 8} + C L^{3/2} \rho^{2 -  \gamma  -\frac \delta{16}} 
    \|(\mathbb{H}_0+\mathbb{Q}_4)^{\frac{1}{2}} T_2\Omega\|  + o(L^{3})_{L\to \infty} ,
$$  
%
where we used again the bounds in Propositions~\ref{pro: est number op} and~\ref{pro: propagation estimates}.   We can further bound  the term  $\widetilde{\mathbb{Q}}_{2;<}$ on the right hand side above using \eqref{eq: est Q2*<}.
Therefore, combining all the estimates, we proved that 
\begin{align}\nonumber
  \langle T_1 T_2\Omega, \mathcal{H}^{\mathrm{eff}}_{\mathrm{corr}} T_1T_2\Omega\rangle &\leq - \rho_\uparrow \rho_\downarrow \sum_{p\in\Lambda^*} \left(\hat{V}(p) + \hat{V}\ast_{_{\sum}} (\widehat{\chi}_<\hat{\varphi})(p) + 8\pi a \widehat{\chi}_<(p) \right)\hat{\varphi}(p)\widehat{\chi}_>(p) 
  \\ & \quad  +\langle T_2\Omega, \left(\mathbb{H}_0 + \mathbb{Q}_4 + {\mathbb{Q}}_{2;<}+ \mathcal{E}_{\mathbb{H}_0}   \right) T_2\Omega\rangle 
  \nonumber\\ \nonumber
  & \quad + C \rho^{\frac 53 - \frac \delta{16}} L^{3/2}  \|(\mathbb{H}_0+\mathbb{Q}_4)^{\frac{1}{2}} T_2\Omega\|    + CL^{3/2}\rho^{\frac{4}{3}  -\gamma} \|  \mathbb{Q}_4^{\frac 12} T_2\Omega\| \\ & \quad + CL^3\rho^{\frac{8}{3} - \gamma -\frac{\delta}{8}}   +CL^3\rho^{3 -4\gamma - \frac{\delta}{16}}  +  o(L^3)_{L\to \infty}.  
  \label{eq: conjugation eff corr}
\end{align}

To complete the proof of Proposition~\ref{lem:conjugation-T1}, we need the following Lemma, evaluating the constant contribution in the first line on the right-hand side of \eqref{eq: conjugation eff corr}.

\begin{lemma}[Constant contribution to $8\pi a \rho_\uparrow\rho_\downarrow$]\label{lem: constant contribution after T1} Let $V$ be as in Assumption \ref{asu: potential V}. Let $\widehat{\chi}_<, \widehat{\chi}_>$ be defined as in \eqref{eq: def chi< chi>}, with $0<\gamma<1/3$, and let ${\varphi}$ be as in Definition~\eqref{def:scattering-phi-eta}. Then
\begin{align}\nonumber
&\frac 1{L^3} \sum_{p\in\Lambda^*} \left(\hat{V}(p) + \hat{V}\ast_{_{\sum}} (\widehat{\chi}_<\hat{\varphi})(p) + 8\pi a \widehat{\chi}_<(p) \right)\hat{\varphi}(p)\widehat{\chi}_>(p)
  \\ &= \hat V(0) - 8\pi a  - \frac {1}{L^3}  \sum_{{0\ne p \in\Lambda^*}} \frac{(8\pi a)^2(\widehat{\chi}_<(p))^2}{2|p|^2} + \mathcal{E}_{\mathrm{const}}
\label{eq: constant terms after T1}
\end{align}
with 
$|\mathcal{E}_{\mathrm{const}}| \leq C  \rho^{\frac{2}{3} - 2\gamma} + o(1)_{L\to \infty}$.
\end{lemma}

\begin{proof}
Using  $\widehat{\chi}_< + \widehat{\chi}_>=1$, we can rewrite the left-hand side of \eqref{eq: constant terms after T1} (multiplied by $L^3$) as
\begin{equation}\label{eq: constant term 8pia}
    \sum_p  \hat{V}(p)\hat{\varphi}(p)  - \sum_p 
   \hat h_2(p) \hat{\varphi}(p)    
    -\sum_p \hat{V}\ast_{_{\sum}} (\widehat{\chi}_<\hat{\varphi})(p)\hat{\varphi}(p)\widehat{\chi}_<(p) - {8\pi a}\sum_p \hat{\varphi}(p)(\widehat{\chi}_<(p))^2
\end{equation}
where we used the definition of $\hat h_2(p) = ( \hat V(p) -  (\hat{V}\ast_{_{\sum}} \hat\varphi) (p) - 8\pi a  ) \widehat\chi_<(p)$ in \eqref{def:h2}. As in \eqref{ca5}, we have $L^{-3} \sum_p \hat V(p) \hat \varphi(p) = \int_{\R^3} V_\infty \varphi_\infty + o(1)_{L\to \infty} = \hat V(0) - 8\pi a + o(1)_{L\to \infty}$. 
At the end of the proof of Proposition~\ref{prop: scatt canc}, we also showed that $|\hat{h}_2(p)| \leq C|p|^2|\widehat{\chi}_<(p)| + o(1)_{L\to \infty}$. 
The zero-energy scattering equation  \eqref{eq: zero energy scatt eq-intro} 
implies that $p^2 \hat\varphi(p)$ is bounded, hence
$$
\left| \frac 1{L^3} \sum_p \hat h_2(p) \hat{\varphi}(p) \right| \leq  \frac{C}{L^3}\sum_p\widehat{\chi}_<(p)  + o(1)_{L\to \infty} \frac 1{L^3} \sum_{0\neq |p|\leq 5 \rho^{1/3-\gamma}} \frac 1{p^2}\leq  C\rho^{1 - 3\gamma} + o(1)_{L\to \infty} .
$$
Using that $\|\hat{V}\|_\infty \leq C$, the third term in \eqref{eq: constant term 8pia} is bounded as 
\begin{equation}\label{eq: est term I}
  \left|\frac{1}{L^6}\sum_{q,p}\hat{V}(p-q)\widehat{\chi}_<(q)\hat{\varphi}(q)\widehat{\chi}_<(p)\hat{\varphi}(p)\right| \leq C\left(\frac{1}{L^3}\sum_{p\neq 0} \frac{\widehat{\chi}_<(p)}{|p|^2}\right)^2 \leq C\rho^{\frac{2}{3} - 2\gamma}.
\end{equation}
%
%
Finally, an argument as at the end of the proof of Proposition~\ref{prop: scatt canc} shows at $|p^2 \hat \varphi(p) - 4\pi a | \leq Cp^2$, hence
$$
\sum_p \hat{\varphi}(p)(\widehat{\chi}_<(p))^2 = 4\pi a \sum_{p\neq 0}\frac{(\widehat{\chi}_<(p))^2}{2|p|^2} + \mathcal{E}
$$
with $| \mathcal{E}| \leq C \sum_p \chi_<(p)^2 \leq C L^3 \rho^{1-3\gamma}$. This completes the proof.
\end{proof}

Proposition~\ref{lem:conjugation-T1} follows from combining \eqref{eq: conjugation eff corr} with \eqref{eq: constant terms after T1} and a simple Cauchy--Schwarz inequality.

\section{Second quasi-bosonic Bogoliubov transformation $T_2$}\label{sec: T2}

In this section we complete the second step discussed in Section~\ref{sec:sub-strategy}. After the conjugation with $T_1$, the new effective correlation operator is 
$ \mathbb{H}_0 + \mathbb{Q}_{2;<}$. By conjugating this operator  with $T_2$, we will extract the full constant of order $\rho^{7/3}$. The main result of this section is a  rigorous version of  \eqref{eq: secon step appr part II}, plus estimates on the conjugations of $\mathbb{H}_0$, $\mathbb{Q}_4$ and $\mathcal{E}_{\mathbb{H}_0}$, which are needed to control the error terms in \eqref{eq: final T1 up bd} from Proposition \ref{lem:conjugation-T1}. 

Recall that $\mathbb{H}_0$ and $\mathbb{Q}_4$ are defined in \eqref{eq: def H-corr}, while $\mathbb{Q}_{2;<}$ and $\mathcal{E}_{\mathbb{H}_0}$ are  given in \eqref{eq: definition Q2<} and \eqref{eq: error kin energy}, respectively.

\begin{proposition}[Conjugation by $T_2$] \label{prop:T2} Let $0 <\gamma < \frac 16$, $0 < \delta \leq 8\gamma$ with $ 2\gamma + \frac{\delta}{16} \leq \frac{1}{3} $. Then 
\begin{align}\label{eq:main-T2}
  \langle T_2\Omega, (\mathbb{H}_0 + \mathbb{Q}_{2;<}) T_2\Omega\rangle &=- \frac{(8\pi a)^2}{L^6}\hspace{-0.1cm}\sum_{p,r,r^\prime \in\Lambda^*}  \frac{\hat{u}_\uparrow(r+p) \hat{u}_\downarrow(r^\prime - p)\hat{v}_\uparrow(r) \hat{v}_\downarrow(r^\prime)}{|r+p|^2 - |r|^2 + |r^\prime - p|^2 - |r^\prime|^2 + 2\varepsilon}   (\widehat{\chi}_<(p))^2 + \mathcal{E}_{T_2}
  \end{align}
with 
\[
  |\mathcal{E}_{T_2}|\leq CL^3 \left( \rho^{\frac{7}{3} -\gamma + \frac{7}{8}\delta} + \rho^{3  - 3\gamma -\frac{3}{16}\delta} \right). 
\]
Moreover, 
\begin{align}
  \langle T_{2;\lambda}\Omega, \mathbb{H}_0 T_{2;\lambda}\Omega\rangle &\leq CL^3\rho^{\frac{7}{3} -2\gamma -\frac{\delta}{16}},\label{eq:H0 T2-intro}\\
   \langle T_{2;\lambda}\Omega, \mathbb{Q}_4 T_{2;\lambda}\Omega\rangle &\leq CL^3\rho^{\frac{8}{3} - 2\gamma},\label{eq:Q4 T2-intro}\\
     |\langle T_2\Omega, \mathcal{E}_{\mathbb{H}_0} T_2\Omega\rangle |  &\leq CL^3 \left( \rho^{\frac{7}{3} +\gamma} +  \rho^{ 3  - 2\gamma -  \frac \delta 4} 
     \right). \label{eq:cH0 T2-intro}
\end{align}
\end{proposition}

We will study the individual terms  $\langle T_{2;\lambda} \Omega, \mathbb{H}_0  T_{2;\lambda} \Omega\rangle$, $\langle T_{2;\lambda} \Omega, \mathbb{Q}_4 T_{2;\lambda} \Omega\rangle$, $\langle T_2 \Omega,  \mathcal{E}_{\mathbb{H}_0} T_2 \Omega\rangle$ and $\langle T_{2;\lambda} \Omega,  \mathbb{Q}_{2;<} T_{2;\lambda} \Omega\rangle$ in Sections~\ref{sec: prop est H0}--\ref{sec: main T2 Q2}, and then complete the proof of Proposition \ref{prop:T2} in Section \ref{sec:conc-T2}. In the proof, it will be very helpful to use the configuration space representation of $B_2$ in \eqref{def:B2}.

\subsection{Conjugation of $\mathbb{H}_0$}\label{sec: prop est H0}
In this subsection we investigate $ T_2^\ast \mathbb{H}_0T_2$, and in particular validate \eqref{tov}. 
%
Moreover, we also prove a rough estimate for $\langle T_2\Omega, \mathbb{H}_0T_2\Omega \rangle$, which is helpful to control the corresponding error term  in Proposition \ref{lem:conjugation-T1}.

\begin{proposition}[Conjugation of $\mathbb{H}_0$ by $T_2$]\label{pro: prop est H0 T2} Let $\lambda\in [0,1]$, $0 <\gamma < \frac 16$, $0<\delta \leq 8\gamma$ with $2\gamma + \frac{\delta}{16} \leq \frac{1}{3}$. Under the same assumptions as in Theorem \ref{thm: main up bd},
\begin{align}\label{eq: conj H0 T2}
\partial_\lambda T_{2;\lambda}^\ast \mathbb{H}_0T_{2;\lambda} = - T^\ast_{2;\lambda}({\mathbb{Q}}_{2;<} + \mathfrak{e}_{\mathbb{H}_0} )T_{2;\lambda}   
 \end{align}
where $\mathfrak{e}_{\mathbb{H}_0}$ is such that for any normalized $\psi\in \mathcal{F}_{\mathrm{f}}$
\begin{equation}\label{eq: est err kin T2 psi}
  |\langle \psi, \mathfrak{e}_{\mathbb{H}_0} \psi\rangle| \leq CL^{\frac{3}{2}}\rho^{\frac{3}{2} - \frac{\gamma}{2} + \frac{15}{16}\delta}\|\mathcal{N}^{\frac{1}{2}}\psi\| .
 \end{equation}
Moreover, for any normalized $\psi\in \mathcal{F}_{\mathrm{f}}$
\begin{equation}\label{eq: est err kin T2 psi2} 
  \langle T_{2;\lambda}\psi, \mathbb{H}_0 T_{2;\lambda}\psi\rangle \leq \langle\psi , \mathbb{H}_0 \psi\rangle + CL^{\frac{3}{2}}\rho^{\frac{3}{2}(1-\gamma)}
  \|\mathcal{N}^{\frac{1}{2}}\psi\|.
\end{equation}
\end{proposition}

Applying the bounds to $\psi = T_{2;\lambda}\Omega$ and $\psi = \Omega$, respectively, we obtain from \eqref{eq: est number operator T1 final} (with $\lambda=0$)
\begin{equation}\label{eq: prop est H0 T2}
 |\langle T_{2;\lambda}\Omega, \mathfrak{e}_{\mathbb{H}_0} T_{2;\lambda}\Omega\rangle|\leq CL^3\rho^{\frac{7}{3} -\gamma + \frac{7}{8}\delta }, \qquad  \langle T_{2;\lambda}\Omega, \mathbb{H}_0 T_{2;\lambda}\Omega\rangle \leq CL^3\rho^{\frac{7}{3} -2\gamma -\frac{\delta}{16}}  
 .
\end{equation}

\begin{proof}
We have 
\begin{equation}\label{eq: Duhamel H0 T2}
 \partial_\lambda T_{2;\lambda}^\ast \mathbb{H}_0 T_{2;\lambda} =  T_{2;\lambda}^\ast  [\mathbb{H}_0, B_2 - B_2^\ast]T_{2;\lambda^\prime}.
\end{equation}
An explicit computation, using 
%
that for $s+k \notin\mathcal{B}_F^{\uparrow} ,s^\prime-k \notin\mathcal{B}_F^{\downarrow}$,  $s\in \mathcal{B}_F^{\uparrow}$, $s^\prime\in \mathcal{B}_F^{\downarrow}$,  
\[
  -\big(|s+k|^2 - |s|^2 + |s^\prime - k|^2 - |s^\prime|^2\big)\hat{\eta}^\varepsilon_{s,s^\prime}(k) = - 8\pi a  + 2\varepsilon \hat{\eta}^{\varepsilon}_{s,s^\prime}(k)
\]
by the definition of $\hat{\eta}^\varepsilon_{s,s^\prime}$ in \eqref{eq: def eta epsilon}, gives 
$[\mathbb{H}_0, B_2 -B_2^\ast] =  - {\mathbb{Q}}_{2;<}  + 2 \eps (B_2 + B_2^*)$. 
%
We have already estimated the expectation value of $B_2 + B_2^*$  in the proof of Proposition~\ref{pro: est number op}. 
Recalling that $\varepsilon = \rho^{\frac{2}{3} + \delta}$ and applying this bound, we directly obtain \eqref{eq: est err kin T2 psi}. 
%
To prove \eqref{eq: est err kin T2 psi2} we use Gr\"onwall's Lemma together with  the bound  \eqref{eq: est Q2< prop est} on ${\mathbb{Q}}_{2;<} $. 
\end{proof}

\subsection{Conjugation of $\mathbb{Q}_4$}\label{sec: prop Q4 T2}

In this subsection we conjugate $\mathbb{Q}_4$ with respect to $T_{2;\lambda}$. We will show that   $\langle T_{2;\lambda}\Omega, \mathbb{Q}_4 T_{2;\lambda}\Omega\rangle$ does not contribute to the energy to order $\rho^{7/3}$. 

\begin{proposition}[Conjugation of $\mathbb{Q}_4$ by $T_2$]\label{pro: prop est Q4 T2} Let $\lambda\in [0,1]$, $0< \gamma < \frac 16$, $0< \delta\leq 8\gamma$ with $ 2\gamma + \frac{\delta}{16}\leq \frac{1}{3}$. Under the same assumptions as in Theorem \ref{thm: main up bd}
\begin{equation}\label{t2q4}
  \langle T_{2;\lambda}\Omega, \mathbb{Q}_4 T_{2;\lambda}\Omega\rangle \leq CL^3\rho^{\frac{8}{3} - 2\gamma}.
\end{equation}
\end{proposition}

\begin{proof}
For the proof we find it convenient to split the operator $\mathbb{Q}_4$
 according to the support of the momenta of each of the $\hat{u}_\sigma$. To make this precise, we introduce smooth functions $\hat{\zeta}^<, \hat{\zeta}^>: \R^3\to [0,1]$ with $\hat\zeta^< + \hat \zeta^> = 1$  such that
\begin{equation}\label{eq: def zeta Q4}
  \hat\zeta^<(k) := \begin{cases} 1 &\mbox{if}\,\,\, |k| \leq 5 \max\{k_F^\uparrow, k_F^\downarrow\}, \\
                                    0 &\mbox{if}\,\,\, |k| > 6 \max\{k_F^\uparrow, k_F^\downarrow\}.
                    \end{cases} 
\end{equation}
We denote $\zeta^{<}:\Lambda\to \R$ the periodic function with Fourier coefficients $\hat{\zeta}^< (k)$ 
for $k\in (2\pi/L)\mathbb{Z}^3$, and also $ \hat{u}^<_\sigma(k) = \hat{u}_\sigma(k) \hat\zeta^<_\uparrow(k)$, 
$ \hat{u}^>_\sigma(k) = \hat{u}_\sigma(k) \hat\zeta^>_\uparrow(k)$. Accordingly, we split
\[
  \mathbb{Q}_4 = \mathbb{Q}_4^> + \mathbb{Q}_4^<,
\]
where 
$$
  \mathbb{Q}_4^> 
 = \sum_{\sigma,\sigma^\prime} \int dxdy\,  V(x-y) a^\ast_\sigma(u_x^>)a^\ast_{\sigma^\prime}(u_y^>)a_{\sigma^\prime}(u_y^>)a_\sigma(u_x^>)
$$
and  ${\mathbb{Q}}_4^<$ contains all the other possible terms. Specifically, 
$
{\mathbb{Q}}_4^< = \sum_{j=0}^3 {\mathbb{Q}}_{4;j}^<  
$, 
where ${\mathbb{Q}}_{4;j}^<$ contains all the terms that have  $j$ functions $u^>_\sigma$ and $4-j$ functions $u^<_\sigma$. 

Let $\xi_\lambda =  T_{2;\lambda} \Omega$. Using that $\|a_\sigma(u_x^<)\| \leq C\rho^{1/2}$, it is easy to see that 
\begin{equation}\label{eq: est Q4<2 and Q4<3}
  |\langle \xi_\lambda , ({\mathbb{Q}}_{4;0}^< + {\mathbb{Q}}_{4;1}^< + \mathbb{Q}_{4;2}^{<,a}) \xi_\lambda \rangle| \leq C\rho\langle \xi_\lambda, \mathcal{N} \xi_\lambda\rangle
\end{equation}
where $\mathbb{Q}_{4;2}^{<,a}$ denotes the part of $\mathbb{Q}_{4;2}^{<}$ with one $a_\sigma(u^>)$ and one $a^\ast_\sigma(u^>)$. The other contributions to $\mathbb{Q}_4^<$ can be bounded via the  Cauchy--Schwarz inequality as
\begin{equation}\label{eq: est Q4<1}
  | \langle \xi_\lambda, ( {\mathbb{Q}}_{4;2}^{<,b} +  {\mathbb{Q}}_{4;3}^<) \xi_\lambda \rangle| \leq C\rho^{\frac{1}{2}}\|\mathcal{N}^{\frac{1}{2}}\xi_\lambda \|\|(\mathbb{Q}_4^>)^{\frac{1}{2}}\xi_\lambda\|.
\end{equation}
In particular, from \eqref{eq: est Q4<2 and Q4<3} and \eqref{eq: est Q4<1} we  find that 
\begin{equation}\label{eq: est Q4 tilde}
  |\langle \xi_\lambda, {\mathbb{Q}}_4^< \xi_\lambda \rangle| \leq 
  C\rho\langle \xi_\lambda, \mathcal{N}\xi_\lambda\rangle + C\langle\xi_\lambda, \mathbb{Q}_4^> \xi_\lambda\rangle.
\end{equation}

We are thus left with  estimating ${\mathbb{Q}}_4^>$, and we will use Gr\"onwall's Lemma. We start by computing 
\[
 \partial_\lambda T_{2;\lambda}^\ast \mathbb{Q}_4^>T_{2;\lambda} =   T_{2;\lambda}^\ast [\mathbb{Q}_4^>, B_2 - B_2^\ast]T_{2;\lambda}. 
\]
The terms we find from the commutator above are of the same type as the ones in Proposition \ref{pro: Q4}. 
Explicitly, 
$  [\mathbb{Q}_4^>, B_2-B_2^\ast] = \mathrm{I}_1 + \mathrm{I}_2 + \mathrm{I}_3 + \mathrm{h.c.}$ 
with
\begin{multline}\label{eq: def I1 Q4 T2}
   \mathrm{I}_1 =  \frac{8\pi a}{L^6}\sum_{\sigma}\sum_{k,p,r,r^\prime, s} \hat{V}(k)\widehat{\chi}_<(p) \hat{\eta}_{r,r^\prime}^{\varepsilon}(p) \hat{u}_\uparrow^>(r+p)\hat{u}_\downarrow(r^\prime - p)\hat{u}_\uparrow^>(r+p-k)\hat{u}_{\sigma}^>(s -k)\hat{u}_{\sigma}^>(s)\hat{v}_\uparrow(r)\hat{v}_\downarrow(r^\prime)\times 
    \\
    \times \hat{a}^\ast_{s - k, \sigma}\hat{a}_{r^\prime - p, \downarrow}\hat{a}_{-r,\uparrow}\hat{a}_{-r^\prime, \downarrow}\hat{a}_{s, \sigma}\hat{a}_{r+p-k,\uparrow},
\end{multline}
\begin{multline*}
   \mathrm{I}_2 =  \frac{1}{L^6}\sum_{\sigma}\sum_{k,p,r,r^\prime, s} \hat{V}(k)\widehat{\chi}_<(p) \hat{\eta}_{r,r^\prime}^{\varepsilon}(p) \hat{u}_\uparrow(r+p)\hat{u}_\downarrow^>(r^\prime - p)\hat{u}_\downarrow^>(r^\prime -p-k)\hat{u}_{\sigma}^>(s -k)\hat{u}_{\sigma}^>(s)\hat{v}_\uparrow(r)\hat{v}_\downarrow(r^\prime)\times 
    \\
    \times \hat{a}^\ast_{s - k, \sigma}\hat{a}_{r + p, \uparrow}\hat{a}_{-r^\prime,\downarrow}\hat{a}_{-r, \uparrow}\hat{a}_{s, \sigma}\hat{a}_{r^\prime -p-k,\downarrow},
\end{multline*}
\begin{multline}\label{eq: I3 Q4 T2}
  \mathrm{I}_3 = - \frac{1}{L^6}\sum_{k,q, r,r^\prime}\hat{V}(k)\widehat{\chi}_<(p) \hat{\eta}_{r,r^\prime}^{\varepsilon}(p) \hat{u}_\uparrow^>(r+p)\hat{u}_\downarrow^>(r^\prime - p) \hat{u}_\uparrow^>(r+p+k)\hat{u}_\downarrow^>(r^\prime -p-k)\hat{v}_\uparrow(r)\hat{v}_\downarrow(r^\prime)\times
  \\
   \times \hat{a}_{r+p+k, \uparrow}\hat{a}_{-r,\uparrow}\hat{a}_{r^\prime -p-k, \downarrow}\hat{a}_{-r^\prime, \downarrow}.
 \end{multline} 
 In what follows we will 
 often use that $\|V\|_1 \leq C$, the bounds in Lemma \ref{lem:cut-off-chi} and Lemma \ref{lem: zeta t L1} and that 
 \begin{equation}\label{eq: est avt Q4}
  \|a_\sigma(v^t_x)\| = \|{v}^t_\sigma\|_2 \leq \rho^{\frac{1}{2}}e^{t(k_F^\sigma)^2}
 \end{equation}
 with $v^t_\sigma$ defined  in \eqref{eq: def ut vt}.
 
We start by considering $\mathrm{I}_1$. 
Since  $\hat{u}^>_\uparrow(r+p)= \hat{\zeta}^>(r+p) \hat{u}_\uparrow (r+p)= \hat\zeta^>(r+p)$, the definition of $\hat{\zeta}^>$ implies that $|r+p| \geq 5\max\{k_F^\uparrow,k_F^\downarrow\}$, and hence 
$|p-r'| \geq 3 \max\{k_F^\downarrow,k_F^\uparrow\}$ since $|r| \leq k_F^\uparrow$ and $|r'|\leq k_F^\downarrow$. 
Moreover, also $|r^\prime - p| \leq 6\rho^{1/3 - \gamma}$ and hence we can replace $\hat{u}_\downarrow(r^\prime - p)$ in \eqref{eq: def I1 Q4 T2}  by
$\hat \eta(r'-p)$ with $\hat \eta$ the characteristic function of the set $3\max\{k_F^\uparrow,k_F^\downarrow\} \leq |k| \leq 6\rho^{1/3 - \gamma}$. Proceeding then as in Remark 
\ref{rem: T1 and T2 configuration space} we can rewrite $I_1$ in configuration space as
\begin{equation}\label{eq: def It Q4>}
   {\mathrm{I}}_1  
= 8\pi a \sum_\sigma\int_0^{\infty} dt\, e^{-2t\varepsilon}\int dxdydzdz^\prime\, V(x-y)\chi_<(z-z^\prime)\zeta^{t}(z;y)\,  a^\ast_{\sigma}(u_y^>) a_\uparrow(v_z^t)a_{\downarrow}(\eta_{z^\prime}^t)a_\downarrow(v_{z^\prime}^t)a_\downarrow (u_x^>)a_{\sigma}(u_y^>)
\end{equation}
with 
\begin{equation}\label{eq: def u>t}
 \zeta^{t} (z;y) := \frac{1}{L^3}\sum_k \zeta^> (k) e^{-t|k|^2} e^{ik\cdot (z-y)},\qquad \eta^t(x;y) = \frac{1}{L^3}\sum_k \hat{\eta}(k) e^{-t|k|^2} e^{ik\cdot (x-y)}.
\end{equation}
We shall write the right-hand side of \eqref{eq: def It Q4>} as 
$8\pi a \int_0^{\infty} dt \, e^{-2t\varepsilon}\, \mathrm{I}_1^t$. With the aid of \eqref{eq: est avt Q4}, Lemma \ref{lem: zeta t L1} and the Cauchy--Schwarz inequality, we find
\begin{align*}
  |\langle \xi_\lambda, {\mathrm{I}}^t_1\xi_\lambda\rangle| &\leq C\sum_\sigma\int dxdydzdz^\prime\, V(x-y)|\chi_<(z-z^\prime)||\zeta^{t}(z;y)| \|\hat{v}^t_\downarrow\|_2\|\hat{\eta}^t\|_2\|\hat{v}_{\uparrow}^t\|_2\| a_{\sigma}(u_y^>)\xi_{\lambda}\| \| a_\downarrow (u_x^>)a_{\sigma}(u_y^>)\xi_{\lambda}\|
  \\
  &\leq  C\|\chi_<\|_1 \|\zeta^{t}\|_1 \|\hat{v}^t_\uparrow\|_2\|\hat{v}^t_\downarrow\|_2 \|\hat{\eta}^t\|_2 \sum_\sigma \int dxdy\, V(x-y)\|a_{\sigma}(u_y^>)\xi_{\lambda}\|\|a_\downarrow(u_x^>)a_{\sigma}(u_y^>)\xi_{\lambda}\|
  \\
  &\leq  C\rho^{\frac{1}{2}}e^{t(k_F^\uparrow)^2}\|\hat{v}^t_\downarrow\|_2\|\hat{\eta}^t\|_2 \|\mathcal{N}^{\frac{1}{2}}\xi_{\lambda}\|\|(\mathbb{Q}_4^>)^{\frac{1}{2}}\xi_{\lambda}\|. 
\end{align*}
For the integration over $t$, we can use the 
Cauchy--Schwarz's inequality to get 
\[
  \int_0^\infty dt\, e^{-2t\varepsilon} e^{t(k_F^\uparrow)^2} \|\hat{\eta}^t\|_2 \|\hat{v}_\downarrow^t\|_2 \leq \left( \int_0^\infty dt\, e^{-2t\varepsilon} e^{2t(k_F^\uparrow)^2} e^{2t(k_F^\downarrow)^2}\|\hat{\eta}_\downarrow^t\|_2^2 \right)^{\frac{1}{2}}\left(\int_0^\infty dt\, e^{-2t\varepsilon} e^{-2t(k_F^\downarrow)^2}  \|\hat{v}_\downarrow^t\|_2^2\right)^{\frac{1}{2}}.
\]
The second  integral above is bounded   in \eqref{fac2}. The first equals
\begin{equation}\label{int:eta}
\frac 12  \frac{1}{L^3}\sum_{3\max\{k_F^\uparrow,k_F^\downarrow\} \leq |k| \leq 6\rho^{1/3 - \gamma}} \frac{1}{ |k|^2 - (k_F^\downarrow)^2 - (k_F^\uparrow)^2   +\varepsilon} \leq 
C\rho^{\frac{1}{3} -\gamma}.
\end{equation}
%
In combination, we have thus shown that 
\begin{equation}\label{eq: est I1 Q4 T2}
  |\langle \xi_\lambda, {\mathrm{I}}_1 \xi_\lambda\rangle| \leq  C\rho^{\frac{5}{6}  -\frac{\gamma}{2} - \frac{\delta}{16}}\|\mathcal{N}^{\frac{1}{2}}\xi_\lambda\| \|(\mathbb{Q}_4^>)^{\frac{1}{2}}\xi_\lambda\|\leq  C \rho^{\frac{5}{3} - \gamma - \frac{\delta}{8}}\langle \xi_\lambda, \mathcal{N}\xi_\lambda \rangle + C\langle \xi_\lambda, \mathbb{Q}_4^> \xi_\lambda \rangle. 
\end{equation}
The term $\mathrm{I}_2$ can be bounded in the same way.

To conclude the analysis of $\mathbb{Q}_4^>$, we now consider $\mathrm{I}_3$ in \eqref{eq: I3 Q4 T2}. Arguing as above, 
we can replace $\hat{u}^>_\uparrow(r+p)\hat{u}^>_\downarrow(r^\prime - p)$ 
by $\hat{\eta}(r+p)\hat{\eta}(r^\prime-p)$. 
In configuration space, we then obtain
\[
 {\mathrm{I}}_3  = 8\pi a\int_0^{\infty}dt\, e^{-2t\varepsilon}\int\, dxdydzdz^\prime\, V(x-y)\chi_<(z-z^\prime) \eta^{t}(z;x) \eta^{t}(z^\prime;y)\, a_\uparrow(u_x^>)a_\uparrow(v^t_z)a_\downarrow(u_y^>)a_\downarrow(v^t_{z^\prime})  
\]
with $\eta^t$ as in \eqref{eq: def u>t}.
Writing $\mathrm{I}_3 = \int_0^{\infty}dt\, e^{-2t\varepsilon} \, \mathrm{I}_3^t$ and using \eqref{eq: est avt Q4}, we can  estimate 
$$
  |\langle \xi_\lambda, \mathrm{I}_3^t\xi_\lambda\rangle|  
 \leq C\rho e^{t(k_F^\uparrow)^2}e^{t(k_F^\downarrow)^2} \int dxdydzdz^\prime V(x-y)|\chi_<(z-z^\prime)| |\eta^t(z;x)||\eta^t(z^\prime;y)| \|a_\uparrow(u^>_x)a_\downarrow(u^>_y)\xi_{\lambda}\| .
$$
With
\[
  \left|\int dzdz^\prime \chi_<(z-z^\prime) \eta^{t}(z;x)\eta^{t}(z^\prime;y)\right|\leq C\|\chi_<\|_1 \|\eta^{t}\|_2^2\leq C\|\eta^{t}\|_2^2
\]
and the Cauchy--Schwarz inequality  we then get
\[
  |\langle \xi_\lambda, \mathrm{I}_3\xi_\lambda\rangle| 
  \leq CL^{\frac{3}{2}}\rho \left(\int_0^{\infty} dt\, e^{-2t\varepsilon} e^{t(k_F^\uparrow)^2} e^{t(k_F^\downarrow)^2}  \|\hat{\eta}^{t}\|_2^2 \right) \|(\mathbb{Q}_4^>)^{\frac{1}{2}}\xi_{\lambda}\|.
\]
The integral is again bounded as in \eqref{int:eta}. Inserting this bound, we obtain
%
\begin{equation}\label{eq: est term II}
  |\langle \xi_\lambda, \mathrm{I}_3 \xi_\lambda \rangle| \leq CL^{\frac{3}{2}}\rho^{ \frac{4}{3} - \gamma}\|(\mathbb{Q}_4^>)^{\frac{1}{2}}\xi_{\lambda}\| \leq CL^3 \rho^{ \frac{8}{3} - 2\gamma} + \langle \xi_\lambda, \mathbb{Q}_4^> \xi_\lambda \rangle.
\end{equation}
Combining \eqref{eq: est I1 Q4 T2} and \eqref{eq: est term II}, we get
\begin{align*}
  |\partial_\lambda\langle T_{2;\lambda}\Omega, \mathbb{Q}_4^> T_{2;\lambda}\Omega\rangle|  &\leq CL^3\rho^{\frac{8}{3} - 2\gamma} + C\rho^{\frac{5}{3} -\gamma - \frac{\delta}{8}}\langle T_{2;\lambda}\Omega, \mathcal{N} T_{2;\lambda}\Omega\rangle  + C\langle T_{2;\lambda}\Omega, \mathbb{Q}_4^>  T_{2;\lambda}\Omega \rangle
\\
&\leq 
CL^3 \rho^{\frac{8}{3} -2\gamma} + C\langle T_{2;\lambda}\Omega, \mathbb{Q}_4^> T_{2;\lambda}\Omega \rangle,
\end{align*}
where we used the bound on the number operator $\mathcal{N}$ from  Proposition \ref{pro: est number op} (Eq.~\eqref{eq: est number operator T1 final} for $\lambda=0$) and the constraints $0<\gamma<1/3$ and $\delta \leq 8\gamma$ in the last estimate. Therefore, by Gr\"onwall's Lemma, for any $\lambda\in [0,1]$, we have,
\[
  |\langle T_{2;\lambda}\Omega, \mathbb{Q}_4^> T_{2;\lambda}\Omega\rangle | \leq  CL^3\rho^{\frac{8}{3} - 2\gamma}.
\]
Inserting this estimate in \eqref{eq: est Q4 tilde} and using again the bound for the number operator from Proposition \ref{pro: est number op}, we  find \eqref{t2q4}.  
\end{proof}


\subsection{Conjugation of $\mathcal{E}_{\mathbb{H}_0}$}\label{sec: error kinetic T2}
The purpose of this subsection is to show that $  \langle T_2\Omega, \mathcal{E}_{\mathbb{H}_0} T_2\Omega\rangle = L^3 o(\rho^{\frac{7}{3}}) $, with $\mathcal{E}_{\mathbb{H}_0}$  defined in \eqref{eq: error kin energy} 
the kinetic error term coming from Proposition \ref{pro: H0} from the conjugation of $\mathbb{H}_0$ by $T_1$. 

\begin{proposition}[Final estimate for $\mathcal{E}_{\mathbb{H}_0}$]\label{pro: estimate EH0 fin} Let $0< \gamma < \frac 16$, $0<\delta \leq 8\gamma$ with $ 2\gamma + \frac{\delta}{16} \leq \frac{1}{3} $. Under the same assumptions of Theorem \ref{thm: main up bd}
\begin{equation}\label{eq: est err kin T2Omega}
  |\langle T_2\Omega, \mathcal{E}_{\mathbb{H}_0} T_2\Omega\rangle |   \leq CL^3 \left( \rho^{\frac{7}{3} +\gamma} + \rho^{3 - 2\gamma-  \frac{\delta}{4} }  
   \right).
\end{equation}
\end{proposition}
\begin{proof}
Again we use Gr\"onwall's Lemma, and compute   
\[
   \partial_\lambda T_{2;\lambda}^\ast\mathcal{E}_{\mathbb{H}_0} T_{2;\lambda} = T_{2;\lambda} [\mathcal{E}_{\mathbb{H}_0}, B_2 - B_2^\ast] T_{2;\lambda}.
\]
The terms in $[\mathcal{E}_{\mathbb{H}_0}, B_2 - B_2^\ast]$ are of the same type as those in Proposition \ref{pro: H0}, see \eqref{eq: def Ia}--\eqref{eq: term III H0 T1}.  The only difference is that each of them has $\hat{\eta}^{\varepsilon}_{r,r^\prime}(p) \widehat{\chi}_<(p)$ in place of $\hat\varphi(p)\widehat{\chi}_>(p)$, which distinguishes $B_1$ from $B_2$. As in the proof of Proposition \ref{pro: H0}, it is enough to prove the bound in \eqref{eq: est err kin T2Omega} with $ \mathcal{E}_{\mathbb{H}_0}$ replaced by $ \mathcal{E}_{\mathbb{H}_0;1}$, defined in \ref{def:E1}. We shall again write $[\mathcal{E}_{\mathbb{H}_0;1}, B_2] = \sum_{j=1}^6 \mathrm{I}_j$ and consider the terms individually. 
We will estimate the expectation value of all the terms in the state $\xi_\lambda = T_{2;\lambda}\Omega$. In the estimates we will often use the bounds in Lemma \ref{lem: bounds phi}, Lemma \ref{lem: bounds b b*}, Lemma \ref{lem:cut-off-chi} as well as \eqref{eq: est a v H0} and \eqref{eq: est avt Q4}. 
Similarly to \eqref{eq: from aut to Number op}, we shall also need that  
\begin{equation}\label{eq: ut numb op H0 T2}
  \int dx \, \|a_\sigma( v^t_{x})\xi_\lambda\|^2 \leq  e^{2 t(k_F^\sigma)^2}\langle \xi_\lambda, \mathcal{N}\xi_\lambda\rangle .
\end{equation} 
Moreover, we will often use the constraints on $p,r,r^\prime$ to replace  $\hat{u}_\uparrow(r+p)$ and $ \hat{u}_\downarrow(r^\prime-p)$ by $\hat{u}^<_\uparrow(r+ p)$ and $\hat{u}^<_\downarrow(r^\prime - p)$, respectively, with $\hat{u}^<_\sigma$ defined as in \eqref{eq: u<}. 
Similarly as in Proposition \ref{pro: H0}, it is convenient to split $\mathrm{I}_1$, $\mathrm{I}_2$ and $\mathrm{I}_3$ into a sum of two terms each.

Proceeding as in the proof of  Proposition \ref{pro: H0}, taking into account Remark~\ref{rem: T1 and T2 configuration space}, we have for the analogue of \eqref{eq: def Ia} 
\begin{equation}\label{spliti}
  \mathrm{I}_{1} = 8\pi a\int_0^{\infty}dt\, e^{-2t\varepsilon} \left( \mathrm{I}_{1;a}^t + \mathrm{I}_{1;b}^t \right)
\end{equation}
with 
$$
\mathrm{I}_{1;a}^t  = \ \int dxdydzdz^\prime\, \widetilde{\varphi}(x-y) \chi_<(z-z^\prime)v^t_\downarrow(y;z^\prime)\Delta v^t_\uparrow(x;z)\, a^\ast_\uparrow({u}_x)a^\ast_\downarrow(u_y)a_\downarrow(u^t_{z^\prime})a_\uparrow(u^t_z)
$$
and
$$
 \mathrm{I}_{1;b}^t = \sum_{\ell=1}^3 \int dxdydzdz^\prime\, \widetilde{\varphi}(x-y) \chi_<(z-z^\prime) v^t_\downarrow(y;z^\prime)\partial_\ell v^t_\uparrow(x;z) a^\ast_\uparrow(\partial_\ell {\nu}_x^>)a^\ast_\downarrow(u_y)a_\downarrow(u^t_{z^\prime})a_\uparrow(u^t_z)
$$
with $\nu^>$ defined in \eqref{eq: def nu >}.
With the aid of Cauchy--Schwarz inequality, together with bounds as in  \eqref{eq: est conv L2} and \eqref{eq: est conv L2 partial ell}, using that $0\leq \hat{u}_\downarrow^t\leq e^{-t(k_F^\downarrow)^2}$ and $0\leq \hat{u}_\uparrow\leq 1$, we obtain for the first term
\begin{align*}
  |\langle \xi_\lambda, \mathrm{I}_{1;a}^t\xi_\lambda\rangle| & \leq  \left( \int dydz\, \left \|\int dx\, \widetilde{\varphi}(x-y) \Delta v_\uparrow^t(x;z) a_\uparrow(u_x)\right\|^2 \|a_\downarrow(u_y)\xi_\lambda\|^2\right)^{\frac{1}{2}} 
  \\ & \qquad 
  \times \left(\int dydz\,   \left\|\int dz^\prime\, \chi_<(z-z^\prime) v_\downarrow^t(y;z^\prime) a_\downarrow(u_{z^\prime}^t)\right\|^2\| a_\uparrow(u_{z}^t)\xi_\lambda\|^2   \right)^{\frac{1}{2}}
  \\
 &  \leq e^{-t(k_F^\downarrow)^2} \left( \int dydzdz^\prime\, | \widetilde{\varphi}(z^\prime-y)|^2|\Delta v_\uparrow^t(z;z^\prime)|^2  \|a_\downarrow(u_y)\xi_\lambda\|^2\right)^{\frac{1}{2}} \\ & \qquad \times \left(\int dydzdz^\prime\, |\chi_<(z-z^\prime)|^2 |v^t_\downarrow(y;z^\prime)|^2 \|a_\uparrow(u^t_z)\xi_\lambda\|^2 \right)^{\frac{1}{2}}.
\end{align*}
%
To bound the terms, we shall use Lemma~\ref{lem: bounds phi} and  $\|\chi_<\|_2 \leq C\rho^{1/2 -(3/2)\gamma}$, as well as \eqref{eq: from aut to Number op}. 
This way, we obtain
%
\[
  |\langle \xi_\lambda ,\mathrm{I}_{1;a}^t\xi_\lambda\rangle|  \leq   e^{-t\sum_{\sigma}(k_F^\sigma)^2 }\|\widetilde{\varphi}\|_2\|\chi_<\|_2\|\Delta v_\uparrow^t\|_2\|v_\downarrow^t\|_2 \langle \xi_\lambda, \mathcal{N}\xi_\lambda\rangle\leq C\rho^{1-\gamma} e^{-t\sum_{\sigma}(k_F^\sigma)^2 }\|{v}_\uparrow^t\|_2\|{v}_\downarrow^t\|_2 \langle \xi_\lambda, \mathcal{N}\xi_\lambda\rangle.
\]
The Cauchy--Schwarz inequality and \eqref{fac2} for the integral with respect to $t$ then yields
\begin{equation}\label{eq: est Ia1 T2}
\int_0^\infty dt\, e^{-2t \varepsilon} |\langle \xi_\lambda, \mathrm{I}_{1;a}^t\xi_\lambda\rangle| \leq 
C\rho^{\frac{4}{3}-\gamma - \frac{\delta}{8}}\langle \xi_\lambda, \mathcal{N}\xi_\lambda\rangle.
\end{equation}
The term $\mathrm{I}_{1;b}$ can be bounded very similarly, using in addition  
%
%
\eqref{eq: est partial u H0}. We obtain
\begin{equation}\label{eq: est Ia2 T2}
\int_0^\infty dt\, e^{-2t \varepsilon}  |\langle \xi_\lambda, \mathrm{I}_{1;b}^t\xi_\lambda\rangle|  
  \leq C\rho^{1 - \gamma - \frac{\delta}{8}} \|\mathbb{H}_0^{\frac{1}{2}}\xi_\lambda\|\|\mathcal{N}^{\frac{1}{2}}\xi_\lambda\|.
\end{equation}

Next we consider the term $\mathrm{I}_2$. Proceeding as above, we find again \eqref{spliti} with $I_1$ replaced by $I_2$ everywhere, and 
%
$$
  \mathrm{I}_{2;a}^t
= \frac{1}{L^3} \sum_r |r|^2 \hat{v}_\uparrow^t(r) \left(\int dx\, e^{ir\cdot x} a_\uparrow^\ast (u_x) b_\downarrow^\ast (\widetilde{\varphi}_x) \right)\left(\int dz\, e^{-ir\cdot z} b_\downarrow^t({\chi_{<}}_{z})a_\uparrow(u_z^t)\right) ,
$$
where we  used the notation in \eqref{eq: def b phi} and introduced in addition 
\begin{equation}\label{eq: def bt chi <}
  b_\sigma^t({\chi_{<}}_{z}) := \int dz^\prime\, \chi_<(z-z^\prime)a_\sigma(u^t_{z^\prime})a_\sigma(v^t_{z^\prime}).
\end{equation}  
Using $|r|^2\hat{v}^t_\uparrow(r) \leq C\rho^{2/3}e^{t(k_F^\uparrow)^2}$ and the Cauchy--Schwarz inequality, we obtain
$$
  |\langle \xi_\lambda, \mathrm{I}_{2;a}^t\xi_\lambda\rangle|
  \leq C\rho^{\frac{2}{3}} e^{t(k_F^\uparrow)^2}\sqrt{ \int dx \left\| b_\downarrow(\widetilde{\varphi}_x)a_\uparrow(u_x) \xi_\lambda\right\|^2}\sqrt{ \int dz  \left\| b_\downarrow^t({\chi_<}_z)a_\uparrow(u_z^t)\xi_\lambda\right\|^2}.
$$ 
Using the bound from Lemma \ref{lem: bounds b b*} for $\| b_\downarrow(\widetilde{\varphi}_x)\|$, we find
$$
   \int dx\, \| b_\downarrow(\widetilde{\varphi}_x) a_\uparrow(u_x)\xi_\lambda\|^2 \leq C\rho^{\frac{2}{3} +\gamma}\langle \xi_\lambda,\mathcal{N}\xi_\lambda\rangle.
$$
Similarly, from
\begin{equation}\label{eq: bound bt}
  \|b^t_\downarrow({\chi_<}_z)\|\leq \int dz^\prime |\chi_<(z-z^\prime)|\|a_\downarrow(u^t_{z^\prime})\| \|a_\downarrow(v^t_{z^\prime})\| \leq \|\chi_<\|_1\|{u}^t_\downarrow\|_2 \|{v}^t_\downarrow\|_2 \leq C\|{u}^t_\downarrow\|_2 \|{v}^t_\downarrow\|_2
\end{equation}
together with \eqref{eq: ut numb op H0 T2}, we find that 
\[
\int dz \left\| b_\downarrow^t({\chi_<}_z) a_\uparrow(u_z^t)\xi_\lambda\right\|^2 \leq  Ce^{-2t(k_F^\uparrow)^2} \|{u}^t_\downarrow\|_2^2 \| {v}^t_\downarrow\|_2^2\langle \xi_\lambda, \mathcal{N}\xi_\lambda\rangle.
\]
Using  \eqref{eq: est int t uv final} for the integral with respect to $t$, we obtain
$$ 
\int_0^\infty  dt\, e^{-2t \varepsilon} |\langle \xi_\lambda , \mathrm{I}_{2;a}^t\xi_\lambda\rangle  |\leq C\rho^{ 1 + \frac{\gamma}{2}} \left(\int_0^{\infty} dt\, e^{-2t\varepsilon}\|\hat{u}^t_\downarrow\|_2 \|\hat{v}^t_\downarrow\|_2\right) \langle \xi_\lambda, \mathcal{N}\xi_\lambda\rangle = C\rho^{\frac{4}{3} -\frac{\delta}{16}}\langle \xi_\lambda, \mathcal{N}\xi_\lambda\rangle.
$$
The estimate for $\mathrm{I}_{2;b}^t$ is similar. It is given by
$$
 \mathrm{I}_{2;b}^t  = \frac{1}{L^3}  \sum_{\ell=1}^3\sum_r r_\ell \hat{v}_\uparrow^t(r) \int dx dz  \, e^{ir\cdot (x-z)}   a_\uparrow^\ast(\partial_\ell \nu^>_x) 
 b_\downarrow^\ast (\widetilde{\varphi}_x)
 b_\downarrow^t({\chi_{<}}_{z})
 a_\uparrow(u_z^t).
$$
%
%
Using  \eqref{eq: est partial u H0} and proceeding as above, we obtain
\[
 \int_0^\infty  dt\, e^{-2t \varepsilon} |\langle \xi_\lambda, \mathrm{I}_{2;b}^t\xi_\lambda\rangle| 
  \leq C\rho^{1-\frac{\delta}{16}}\|\mathcal{N}^{\frac{1}{2}}\xi_\lambda\| \|\mathbb{H}_0^{\frac{1}{2}}\xi_\lambda\|.
\]
%
The term $I_3$ can be estimated in the same way as $I_2$, obtaining the same final bound. We omit the details. 

We now consider the  term $\mathrm{I}_4$. 
In contrast to the corresponding one in Proposition \ref{pro: H0} in \eqref{eq: term II H0 T1}, to estimate this term we rewrite it in normal order using that 
\begin{align}\nonumber
  &\hat{a}_{s - k,\downarrow}^\ast\hat{a}_{r^\prime - p,\downarrow}\hat{a}_{-r,\uparrow}\hat{a}_{-r^\prime, \downarrow}\hat{a}_{-s, \downarrow}^\ast\hat{a}_{-r-p+k,\uparrow}^\ast = \hat{a}_{s-k,\downarrow}^\ast\hat{a}^\ast_{-s,\downarrow}\hat{a}^\ast_{-r-p+k, \uparrow}\hat{a}_{-r,\uparrow}\hat{a}_{-r^\prime, \downarrow}\hat{a}_{r^\prime - p, \downarrow}
  \\
  &  +\delta_{p,k}\hat{a}^\ast_{s-k,\downarrow}\hat{a}_{-s, \downarrow}^\ast \hat{a}_{-r^\prime,\downarrow}\hat{a}_{r^\prime - p,\downarrow} - \delta_{r^\prime, s}\hat{a}^\ast_{s-k,\downarrow}\hat{a}_{-r-p+k, \uparrow}^\ast\hat{a}_{-r,\uparrow}\hat{a}_{r^\prime - p, \downarrow} + \delta_{r^\prime, s}\delta_{p,k} \hat{a}_{s-k, \downarrow}^\ast \hat{a}_{r^\prime -p, \downarrow}. \label{I4no}
\end{align}
Correspondingly, we have $\mathrm{I}_4 = 8\pi a \int_0^\infty dt\, e^{-2t\varepsilon} (\mathrm{I}_{4;a}^t 
+ \mathrm{I}_{4;b}^t + \mathrm{I}_{4;c}^t + \mathrm{I}_{4;d}^t)$.
The first term equals 
$$
  \mathrm{I}_{4;a}^t = \sum_{\ell=1}^3 \frac{1}{L^3}\sum_k\hat{u}_\uparrow^t(k)  \int dydz \, e^{ik\cdot (z-y)}   
  a^\ast_{\uparrow}(\partial_\ell v_y) b^\ast_\downarrow(\partial_\ell\widetilde{\varphi}_y) 
  b^t_\downarrow({\chi_<}_{z})
  a_\uparrow(v^t_{z})
$$
where we used again the notations introduced in \eqref{eq: def b phi} and \eqref{eq: def bt chi <}.
Since $0\leq \hat{u}_\uparrow ^t(k) \leq e^{-t(k_F^\uparrow)^2}$,  the Cauchy--Schwarz inequality yields
$$
 |\langle\xi_\lambda,\mathrm{I}_{4;a}^t\xi_\lambda\rangle| \leq  e^{-t(k_F^\uparrow)^2} \sum_{\ell=1}^3 \sqrt{  \int dy \left\|  b_\downarrow(\partial_\ell\widetilde{\varphi}_y) a_\uparrow(\partial_\ell v_y)\xi_\lambda\right\|^2}\sqrt{ \int dz \left\| b^t_\downarrow({\chi_<}_{z})a_\uparrow(v^t_z)\xi_\lambda\right\|^2}. 
$$  
The last factor can be bounded with the aid of \eqref{eq: bound bt}
and \eqref{eq: ut numb op H0 T2}. 
For the first factor, we use 
\begin{align*}
 \left\|  b_\downarrow(\partial_\ell\widetilde{\varphi}_y) a_\uparrow(\partial_\ell v_y)\xi_\lambda\right\| & \leq \int dx\, | \partial_\ell \widetilde \varphi(x-y)| |\| a_\downarrow(u_x) a_\downarrow(v_x) a_\uparrow(\partial_\ell v_y)\xi_\lambda\| \\ &\leq  \| v_\downarrow\|_2 \| \| \partial_\ell v_\uparrow\|_2 \int dx\, | \partial_\ell \widetilde \varphi(x-y)| |\| a_\downarrow(u_x) \xi_\lambda\|
\end{align*}
%
%
and hence 
\begin{align*}
&\int dy \left\|  b_\downarrow(\partial_\ell\widetilde{\varphi}_y) a_\uparrow(\partial_\ell v_y)\xi_\lambda\right\|^2
  \\
 &\leq  C\rho^{\frac{8}{3}} \int dy dx dx^\prime\,  |\partial_\ell\widetilde{\varphi}(y-x)||\partial_\ell\widetilde{\varphi}(y-x^\prime)|\| a_\downarrow(u_x) \xi_\lambda\|\| a_\downarrow(u_{x^\prime}) \xi_\lambda\|
 \leq C\rho^{\frac{8}{3}} \|\partial_\ell\widetilde{\varphi}\|_1^2 \langle \xi_\lambda, \mathcal{N}\xi_\lambda\rangle .
\end{align*}
%
In combination with 
 Lemma \ref{lem: bounds phi} and \eqref{eq: est int t uv final}, this yields
\begin{equation}\label{I4at}
  \int_0^\infty  dt\, e^{-2t \varepsilon} |\langle \xi_\lambda, \mathrm{I}_{4;a}^t\xi_\lambda\rangle| 
  \leq C\rho^{\frac{4}{3} + \frac{\gamma}{2} - \frac{\delta}{16}}\langle \xi_\lambda, \mathcal{N}\xi_\lambda\rangle.
\end{equation}
Next we consider $ \mathrm{I}_{4;b}^t  $, which is given by 
 \[
 \mathrm{I}_{4;b}^t  = \sum_{\ell=1}^3 \int dydz\, 
 u^{t}_\uparrow(y;z)\partial_\ell v^t_\uparrow(y;z)
 b^\ast_\downarrow(\partial_\ell\widetilde{\varphi}_y) 
  b_\downarrow^t({\chi_{<}}_{z}) 
 \]
 and can be bounded as 
\begin{multline*}
   |\langle \xi_\lambda, \mathrm{I}_{4;b}^t\xi_\lambda\rangle| \\ \leq 
   \| v_\downarrow \|_2 \| v_\downarrow^t\|_2\sum_{\ell=1}^3\int dxdydzdz^\prime\, |\partial_\ell\widetilde{\varphi}(x-y)||\chi_<(z-z^\prime)||u^{t}_\uparrow(y;z)||\partial_\ell v^t_\uparrow(y;z)| \|a_\downarrow(u_x)\xi_\lambda\|  \|a_\downarrow(u^t_{z^\prime})\xi_\lambda\|.
\end{multline*}
With  \eqref{eq: from aut to Number op}, 
$\|\partial_\ell^n v_\sigma^t\|_2 \leq C \rho^{n/3}\| v_\sigma^t\|_2$, 
$\| v_\sigma^t\|_2 \leq C \rho^{1/2} e^{t (k_F^\sigma)^2}$,
Lemma~\ref{lem: bounds phi} and Cauchy--Schwarz we thus obtain
$$
   |\langle \xi_\lambda, \mathrm{I}_{4;b}^t\xi_\lambda\rangle|  \leq e^{-t (k_F^\downarrow)^2} \sum_{\ell=1}^3
   \| v_\downarrow \|_2 \| v_\downarrow^t\|_2 \|\partial_\ell\widetilde{\varphi}\|_1\| \chi_<\|_1  \|u^{t}_\uparrow\|_2 \|\partial_\ell v^t_\uparrow\|_2 
   \langle \xi_\lambda, \mathcal{N}\xi_\lambda\rangle
   \leq C   \rho^{1+\gamma}  \| v_\uparrow^t\|_2   \|u^{t}_\uparrow\|_2   
   \langle \xi_\lambda, \mathcal{N}\xi_\lambda\rangle.
   $$  
In combination with \eqref{eq: est int t uv final} this yields the same bound as in \eqref{I4at} after integration over $t$. The same applies to $\mathrm{I}_{4;c}^t$ and we omit the details.
The last  term in $\mathrm{I}_4$ is
$$
\mathrm{I}_{4;d}^t = \frac{1}{L^6}\sum_{\ell=1}^3\sum_{p,r,r^\prime}p_\ell\hat{\varphi}(p)\widehat{\chi}_>(p)\widehat{\chi}_<(p) \hat{u}_\uparrow^t(r+p) r_\ell\hat{v}_\uparrow^t(r)\hat{u}_\downarrow^t(r^\prime) \hat{v}_\downarrow^t(r^\prime+p)  \hat{a}_{r^\prime,\downarrow}^\ast \hat{a}_{r^\prime , \downarrow}.
$$
Since $| p_\ell\hat{\varphi}(p)\widehat{\chi}_>(p)\widehat{\chi}_<(p) | \leq \| \partial_\ell \widetilde\varphi\|_1 \| \chi_<\|_1 $ and 
$$
\frac{1}{L^6}\sum_{p,r}  \hat{u}_\uparrow^t(r+p) |r_\ell| \hat{v}_\uparrow^t(r)\hat{u}_\downarrow^t(r^\prime) \hat{v}_\downarrow^t(r^\prime+p)  \leq C \rho^{4/3} \|u_\uparrow^t\|_2 \|v_\uparrow^t\|_2 
$$
for all $r'\in \Lambda^*$, we can again obtain a bound of the same kind. Altogether, we have thus shown that
\begin{equation}\label{eq: est I4 H0 T2}
  |\langle \xi_\lambda, \mathrm{I}_4 \xi_\lambda\rangle| \leq C\rho^{\frac{4}{3} +\frac{\gamma}{2} - \frac{\delta}{16}}\langle \xi_\lambda, \mathcal{N}\xi_\lambda\rangle.
\end{equation}
The term  $\mathrm{I}_5$ can be estimated very similarly, obtaining the same result.

To conclude, we consider the last error term $I_6$, which corresponds to \eqref{eq: term III H0 T1} in Proposition \ref{pro: H0}. We shall again start by writing it in normal order, using
$$
\hat{a}_{-r,\uparrow}\hat{a}_{-r^\prime, \downarrow}\hat{a}_{-r^\prime + p-k, \downarrow}^\ast\hat{a}_{-r-p+k,\uparrow}^\ast = \delta_{p,k} \left( 1 -  \hat{a}^\ast_{-r,\uparrow}\hat{a}_{-r,\uparrow}
 - \hat{a}_{-r^\prime, \downarrow}^\ast \hat{a}_{-r^\prime, \downarrow} \right) + \hat{a}_{-r^\prime + p-k, \downarrow}^\ast \hat{a}^\ast_{-r-p+k,\uparrow}\hat{a}_{-r,\uparrow}\hat{a}_{-r^\prime, \downarrow}
$$
and correspondingly obtain $4$ terms, which we write as  $\mathrm{I}_6 =  \mathrm{I}_{6;a} +  8\pi a \int_0^\infty dt\, e^{-2t\varepsilon} ( \mathrm{I}_{6;b}^t + \mathrm{I}_{6;c}^t + \mathrm{I}_{6;d}^t)$. The term $\mathrm{I}_{6;a}$ is a constant, given by 
\[
  \mathrm{I}_{6;a} = \frac{1}{L^6}\sum_{p,r,r^\prime} \widehat{\chi}_<(p)\hat{\eta}^\varepsilon_{r,r^\prime}(p) \widehat{\chi}_>(p)   \hat{\varphi}(p) p \cdot r \hat{v}_\uparrow(r) \hat{v}_\downarrow(r^\prime).
\]
To arrive at this expression, we  used that support properties of the various functions imply that
$\hat u_\uparrow (r+p)  \hat u_\downarrow (r'-p) = 1$. 
If we replace $\hat{\eta}^\varepsilon_{r,r^\prime}(p)$ by $1/(2|p|^2)$, the resulting sum over $p$ is zero by symmetry. Hence we can also write
\[
  \mathrm{I}_{6;a} = \frac{1}{L^6}\sum_{p,r,r^\prime} \widehat{\chi}_<(p)\left( \hat{\eta}^\varepsilon_{r,r^\prime}(p) - \frac 1{2 |p|^2}\right) \widehat{\chi}_>(p)   \hat{\varphi}(p) p \cdot r \hat{v}_\uparrow(r) \hat{v}_\downarrow(r^\prime).
\]
On the support of $\widehat{\chi}_<(p)\widehat{\chi}_>(p)$, i.e., for $4\rho^{1/3 - \gamma}\leq |p| \leq 5\rho^{1/3 - \gamma}$, one readily checks that $| \hat{\eta}^\varepsilon_{r,r^\prime}(p) - 1/(2 |p|^2) | \leq C \rho^{1/3} /|p|^3$ 
%
for $r\in \mathcal{B}_F^\uparrow$ and $r^\prime \in \mathcal{B}_F^\downarrow$ (and $\epsilon \lesssim \rho^{2/3}$). Moreover, as already used earlier,  $|\hat\varphi(p)| \leq C/ |p|^2$  as a consequence of \eqref{eq: zero energy scatt eq-intro}. Hence
\begin{equation}\label{eq: est I6a H0 T2}
  |  \mathrm{I}_{6;a}|  \leq C \rho^{\frac 13}  \frac{1}{L^6}\sum_{p,r,r^\prime} \widehat{\chi}_<(p) \frac 1{|p|^4} \widehat{\chi}_>(p)     |r|    \hat{v}_\uparrow(r) \hat{v}_\downarrow(r^\prime) 
  \leq C L^3 \rho^{\frac 83} \int_{4\rho^{1/3 - \gamma}\leq |p| \leq 5\rho^{1/3 - \gamma}} \frac{dp}{|p|^4} = C L^3 \rho^{7/3+\gamma}.
\end{equation}
The term  $\mathrm{I}_{6;b}^t$ 
is given by 
\[
  \mathrm{I}_{6;b}^t = -\frac{1}{L^6} \sum_{p, r,r^\prime} p\cdot r  \hat{\varphi}(p)\widehat{\chi}_>(p)\widehat{\chi}_<(p) \hat{u}^t_\uparrow(r+p)\hat{u}^t_\downarrow(r^\prime - p) \hat{v}^t_\uparrow(r)\hat{v}^t_\downarrow(r^\prime)  \hat{a}^\ast_{-r,\uparrow}\hat{a}_{-r,\uparrow}.
\]
Proceeding as with $\mathrm{I}_{4;d}$ above, we obtain
$$
|\langle \xi_\lambda, \mathrm{I}_{6;b}^t \xi_\lambda\rangle| \leq C \rho^{1+\gamma}  \| u_\uparrow^t\|_2 \|u_\downarrow^t\|_2 e^{t (k_F^\uparrow)^2}  e^{t (k_F^\downarrow)^2}  \langle \xi_\lambda, \mathcal{N}\xi_\lambda\rangle .
$$
The integral over $t$ can then be bounded with the aid of \eqref{fac1}, with the result that 
\begin{equation}\label{reho}
\int_0^\infty dt\, e^{-2t \varepsilon} |\langle \xi_\lambda, \mathrm{I}_{6;b}^t \xi_\lambda\rangle| \leq C \rho^{\frac 4 3}    \langle \xi_\lambda, \mathcal{N}\xi_\lambda\rangle .  
\end{equation}
The same bounds holds for $\mathrm{I}_{6;c}^t$. 
%
%

The last term to study is $\mathrm{I}_{6;d}^t$, which in configuration space can be written as 
\[
 \mathrm{I}_{6;d}^t  = \sum_{\ell=1}^3  \int dxdydzdz^\prime\, \partial_\ell\widetilde{\varphi}(x-y)\chi_<(z-z^\prime)u^{t}_\uparrow(y;z)u^{t}_\downarrow(x;z^\prime) a^\ast_\downarrow(v_x)a^\ast_{\uparrow}(\partial_\ell v_y)a_\downarrow(v^t_{z^\prime})a_\uparrow(v^t_{z})
\]
Using  \eqref{eq: est a v H0}, \eqref{eq: est avt Q4} 
 and \eqref{eq: ut numb op H0 T2} together with the Cauchy--Schwarz inequality 
we can bound
\begin{align*}
  &|\langle\xi_\lambda, \mathrm{I}_{6;d}\xi_\lambda\rangle| \\ &  \leq  \sum_{\ell =1}^3 \| \partial_\ell v_\uparrow\|_2  \| v_\downarrow^t\|_2 \int  dxdydzdz^\prime  \, |\partial_\ell \widetilde{\varphi}(x-y)||\chi_<(z-z^\prime)||u^{t}_\uparrow(y;z)||u^t_\downarrow(x;z^\prime)|  
   \| a_\downarrow(v_x)\xi_\lambda\|\|a_\uparrow(v^t_{z})\xi_\lambda\|
  \\
  &  \leq  C e^{t (k_F^\uparrow)^2} \sum_{\ell =1}^3 \| \partial_\ell v_\uparrow\|_2  \| v_\downarrow^t\|_2   \|\partial_\ell \widetilde{\varphi}\|_1 \|\chi_<\|_1   \, \|\hat{u}_\uparrow^t\|_2 \|\hat{u}_\downarrow^t\|_2  \langle \xi_\lambda, \mathcal{N}\xi_\lambda\rangle
  \\
  & \leq C\rho^{1+\gamma} e^{t(k_F^\downarrow)^2} e^{t(k_F^\uparrow)^2} \|\hat{u}_\uparrow^t\|_2 \|\hat{u}_\downarrow^t\|_2  \langle\xi_\lambda, \mathcal{N}\xi_\lambda\rangle,
\end{align*}
where we also used the bounds from Lemma \ref{lem: bounds phi} and Lemma \ref{lem:cut-off-chi} in the last step. Again we can bound the integral over $t$ using \eqref{fac1} to conclude that \eqref{reho} also holds for $\mathrm{I}_{6;d}^t$.
%
%
Combining all the bounds above, 
we obtain 
\begin{equation}\label{eq: est I6 H0 T2}
|\langle \xi_\lambda, \mathrm{I}_6\xi_\lambda\rangle| \leq CL^3\rho^{\frac{7}{3} + \gamma} + C\rho^{\frac{4}{3}}\langle \xi_\lambda, \mathcal{N}\xi_\lambda\rangle .
\end{equation}
Combining all these bounds 
with \eqref{eq: est number operator T1 final} (for $\lambda=0$) and the second bound in \eqref{eq: prop est H0 T2}, 
%
 and using Gr\"onwall's Lemma, we obtain \eqref{eq: est err kin T2Omega}.
\end{proof}



\subsection{Conjugation of $\mathbb{Q}_{2;<}$} \label{sec: main T2 Q2}

In the next proposition we finally take into account the conjugation of ${\mathbb{Q}}_{2;<}$ by $T_2$. The contribution from $\langle T_2 \Omega, \mathbb{Q}_{2;<} T_2\Omega\rangle$, when combined with the main contribution of $\langle T_2 \Omega, \mathbb{H}_0 T_2\Omega\rangle$ in Proposition \ref{pro: prop est H0 T2}, will give the correct energy of order $\rho^{7/3}$. 

\begin{proposition}[Conjugation of $\mathbb{Q}_{2;<}$ by $T_2$]\label{pro: T2 Q2<} Let $0<\gamma<\frac 16$, $\delta \leq 8\gamma$, $ 2\gamma + \frac{\delta}{16} \leq \frac{1}{3}$. 
Under the same assumptions as in Theorem \ref{thm: main up bd}
\begin{equation}\label{eq: T2 Q2<}
 \partial_\lambda T_{2;\lambda}^\ast {\mathbb{Q}}_{2;<}T_{2;\lambda} = -\frac{16 \pi a}{L^6}\sum_{p,r,r^\prime\in \frac{2\pi}{L}\mathbb{Z}^3} (\widehat{\chi}_<(p))^2\hat{\eta}^{\varepsilon}_{r,r^\prime}(p)\hat{u}_\uparrow(r+p)\hat{u}_\downarrow(r^\prime - p) \hat{v}_\uparrow(r)\hat{v}_\downarrow(r^\prime)  + T^\ast_{2;\lambda}\mathcal{E}_{{\mathbb{Q}}_{2;<}}T_{2;\lambda},
\end{equation}
where $\mathcal{E}_{{\mathbb{Q}}_{2;<}}$ is such that for any $\psi \in\mathcal{F}_{\mathrm{f}}$ 
\begin{equation}\label{eq: conj Q2< T2 psi}
  |\langle \psi, \mathcal{E}_{{\mathbb{Q}}_{2;<}} \psi\rangle| \leq C\rho^{\frac{4}{3} - 2\gamma - \frac{\delta}{16}}\langle \psi, \mathcal{N} \psi\rangle.
\end{equation}
\end{proposition}

Combining \eqref{eq: conj Q2< T2 psi}  for  $\psi = T_{2;\lambda} \Omega$ with \eqref{eq: est number operator T1 final} (for $\lambda=0$), we get for any $\lambda\in [0,1]$ 
\begin{equation}\label{eq: conj Q2< T2 Om}
   |\langle T_{2;\lambda}\Omega, \mathcal{E}_{{\mathbb{Q}}_{2;<}}T_{2;\lambda}\Omega\rangle| \leq CL^3\rho^{3(1 -\gamma)  -\frac{3}{16}\delta}.
\end{equation}

\begin{proof}
Since $\partial_\lambda  T^\ast_{2;\lambda} {\mathbb{Q}}_{2;<} T_{2;\lambda} =  T_{2;\lambda}^\ast [{\mathbb{Q}}_{2;<}, B_2 - B_2^\ast]T_{2;\lambda}$, we have to compute the commutator  $[{\mathbb{Q}}_{2;<}, B_2 - B_2^\ast]$. 
The structure of the various  terms  is the same as those in the commutator in \eqref{eq: derivative lambda Q2} in Proposition \ref{pro: Q2}. 
Following the same strategy, we have 
$$  
[{\mathbb{Q}}_{2;<}, B_2 - B_2^\ast] = (8\pi a)^2 \sum_{j=1}^6 \int_0^\infty dt\, e^{-2t\varepsilon} \,\mathrm{I}_j^t + \mathrm{h.c.}.
$$
The leading order constant term comes from normal ordering $I_6^t$, all the other terms are error terms, which we now estimate individually.
In all the error terms, 
we can replace each $\hat{u}_\sigma$ by $\hat{u}^<_\sigma$ defined in \eqref{eq: u<} as a consequence of the  constraints on $p,k,r,r^\prime$, which will be convenient for the estimates. 

The first term to consider is the analogue of \eqref{46:I1}, 
\[
 \mathrm{I}_1^t  =  \frac{1}{L^3}\sum_k \hat{v}^t_\uparrow(k)  \int dx dz\, e^{ik\cdot (x-z)} a^\ast_\uparrow(u_x)    b^{<\ast}_\downarrow({\chi_<}_{x}) b_\downarrow^t({\chi_<}_{z})a_\uparrow(u_z^{t}) 
\]
where we used the definition \eqref{eq: def bt chi <} and introduced 
\begin{equation}\label{eq:  def b< bt operators Q2<}
  b_\sigma^<({\chi_<}_y) = \int dx\, \chi_<(x-y)a_\sigma(u_x^<)a_\sigma(v_x). 
\end{equation}
By the Cauchy--Schwarz inequality and $0\leq \hat{v}^t_\uparrow(k)\leq e^{t (k_F^\uparrow)^2}$, 
$$
  |\langle \psi, \mathrm{I}_1^t \psi\rangle| 
\leq  e^{ t(k_F^\uparrow)^2}
\left(\int dx\,\left\|b_\downarrow^<({\chi_<}_{x})a_\uparrow(u_x) \psi\right\|^2\right)^{1/2} \left(\int dz\, \left\|b_\downarrow^t({\chi_<}_{z})a_\uparrow(u_z^{t}) \psi\right\|^2\right)^{1/2}.
$$
From the bounds \eqref{eq: bound bt} and $\|b^<_\downarrow({\chi_<}_{y})\| \leq \|\chi_<\|_1\|{v}_\downarrow\|_2 \| {u}^<_\downarrow\|_2 \leq C\rho^{1-\frac{3}{2}\gamma}$, we find with \eqref{eq: from aut to Number op} that
%
\begin{equation}\label{eq: est I1 Q2< T2}
  |\langle \psi, \mathrm{I}_1^t \psi\rangle| \leq C\rho^{1-\frac{3}{2}\gamma}  \|{u}^{t}_\downarrow\|_2 \|{v}^t_\downarrow\|_2 \langle \psi, \mathcal{N}\psi\rangle .
\end{equation}
The integral over $t$ is bounded  in \eqref{eq: est int t uv final}, resulting in
\begin{equation}\label{eq: est I2 Q2< T2}
 \int_0^\infty dt\, e^{-2t\varepsilon}  |\langle \psi, \mathrm{I}_1^t \psi\rangle| \leq C\rho^{\frac{4}{3} -2\gamma -\frac{\delta}{16}}\langle \psi,\mathcal{N}\psi\rangle.
\end{equation}
The estimate of  $\mathrm{I}_2$ works in the same way, with the same result.

Next we consider the analogue of \eqref{46:I3},  given by 
$$
  \mathrm{I}_3^t = \int dxdydzdz^\prime\, \chi_<(x-y) \chi_<(z-z^\prime)v^t_\downarrow(y;z^\prime)v^t_\uparrow(x;z) 
    a^\ast_\uparrow(u_x^<)a^\ast_\downarrow(u_y^<)a_\downarrow(u^t_{z^\prime})a_\uparrow(u^t_z).
 $$
Using  the Cauchy--Schwarz inequality together with the bound  \eqref{eq: from aut to Number op}, we can estimate 
\begin{align*}
  |\langle \psi, \mathrm{I}_3\psi\rangle| &\leq \|{u}^t_\downarrow\|_2\|{u}^<_\downarrow\|_2\int dxdydzdz^\prime\, |\chi_<(x-y)| |\chi_<(z-z^\prime)||v^t_\downarrow(y;z^\prime)||v^t_\uparrow(x;z)| 
 \|a_\uparrow(u_x^<)\psi\| \|a_\uparrow(u^t_z)\psi\|
  \\
   & \leq  e^{-t (k_F^\uparrow)^2}\|{u}^t_\downarrow\|_2\|{u}^<_\downarrow\|_2 \|\chi_<\|_1^2 \|v^t_\downarrow\|_2 \|v^t_\uparrow\|_2  \langle \psi, \mathcal{N}\psi\rangle
\leq C\rho^{1 - \frac{3}{2}\gamma} \|{u}_\downarrow^t\|_2 \|{v}_\downarrow^t\|_2 \langle \psi, \mathcal{N}\psi\rangle,
\end{align*}
where we also used Lemma \ref{lem:cut-off-chi} in the last step. Applying again \eqref{eq: est int t uv final} for the integration over $t$ shows that \eqref{eq: est I2 Q2< T2} also holds for $\mathrm{I}_3^t$.  

%
%
The term $\mathrm{I}_4^t$ is the analogue of \eqref{46:I4}. Similarly as for the term $\mathrm{I}_4$ in the proof of Proposition \ref{pro: estimate EH0 fin},  it is convenient to write $\mathrm{I}_4^t$ in normal order. Using \eqref{I4no}, we can write $\mathrm{I}_4^t = \mathrm{I}_{4;a}^t + \mathrm{I}_{4;b}^t + \mathrm{I}_{4;c}^t + \mathrm{I}_{4;d}^t$, and the individual terms can be 
analyzed in a similar way as the  corresponding terms in Proposition \ref{pro: estimate EH0 fin}. 
The first term equals
  \[
  \mathrm{I}_{4;a}^t  = \frac{1}{L^3} \sum_k \hat{u}_\uparrow^t(k) \int dy dz \, e^{ik\cdot (y-z)}  a_\uparrow^\ast (v_y) b_{\downarrow}^{<\ast}({\chi_<}_y) b_\downarrow^t({\chi_<}_{z})a_\uparrow(v^t_{z}) 
  \]
where we used again the operators introduced in \eqref{eq: def bt chi <} and \eqref{eq:  def b< bt operators Q2<}. Using $0\leq \hat{u}_\uparrow^t(k) \leq e^{-t(k_F^\uparrow)^2}$
and the Cauchy--Schwarz inequality, we find 
\[
    |\langle \psi, \mathrm{I}_{4;a}\psi\rangle|\leq e^{-t(k_F^\uparrow)^2}\left(\int dy \left\|b_\downarrow^<({\chi_<}_{y})a_\uparrow(v_y)\psi\right\|^2\right)^{\frac{1}{2}}\left(\int dz \left\|  b_\downarrow^t({\chi_<}_{z})a_\uparrow(v_z^t)\psi\right\|^2\right)^{\frac{1}{2}}.
\]
We can further bound $\|b_\downarrow^<({\chi_<}_{y})a_\uparrow(v_y)\psi\| \leq \|v_\uparrow\|_2 \|v_\downarrow\|_2 \int dx\, |\chi_<(x-y)| \|a_\downarrow(u^<_x) \psi\|$ and  similarly for $\|b_\downarrow^t({\chi_<}_{z})a_\uparrow(v_z^t)\psi\|$, with the result that 
\[
    |\langle \psi, \mathrm{I}_{4;a}\psi\rangle|\leq \rho \|\chi_<\|_1^2 e^{-t(k_F^\uparrow)^2-t(k_F^\downarrow)^2 } 
     \|v^t_\uparrow\|_2 \|v^t_\downarrow\|_2 \langle \psi, \mathcal{N}\psi\rangle .
\]
where we also used again \eqref{eq: from aut to Number op}. 
Applying now \eqref{fac2} for the integration over $t$, we arrive at
  \[
\int_0^\infty dt\, e^{-2t \varepsilon}   |\langle \psi, \mathrm{I}_{4;a}^t\psi\rangle| 
\leq C\rho^{\frac{4}{3} -\frac{\delta}{8}}\langle \psi, \mathcal{N}\psi\rangle. 
  \]
 The term $\mathrm{I}_{4;b}^t$ is given by 
   \[
    \mathrm{I}_{4;b}^t = \int dxdydzdz^\prime\, \chi_<(x-y)\chi_<(z-z^\prime)u^{t}_\uparrow(z;y) v^t_\downarrow(z^\prime;x) a^\ast_\downarrow(u_x^<)a^\ast_\uparrow(v_y)a_\downarrow(u^{t}_{z^\prime})a_\uparrow(v^{t}_z).
  \]
Similarly as above, 
we can bound it with the aid of the  Cauchy--Schwarz inequality as
  \begin{align*}
|\langle \psi, \mathrm{I}_{4;b}\psi\rangle| &\leq  \|{v}^t_\uparrow\|_2 \|v_\downarrow\|_2   \int dxdydzdz^\prime\, |\chi_<(x-y)||\chi_<(z-z^\prime)| |u^{t}_\uparrow(z;y)| |v^{t}_\downarrow(z^\prime; x)|\|a_\downarrow(u^<_x)\psi\|  \|a_\downarrow(u^t_{z^\prime})\psi\|
\\
  &\leq  e^{ -t(k_F^\downarrow)^2} \|{v}^t_\uparrow\|_2 \|v_\downarrow\|_2    \|{v}_\downarrow^t\|_2 \|{u}_\uparrow^t\|_2 \|\chi_<\|_1^2 \langle \psi, \mathcal{N}\psi\rangle
\leq C\rho  \|{u}^{t}_\uparrow\|_2\|{v}^t_\uparrow\|_2 \langle \psi, \mathcal{N}\psi\rangle.
  \end{align*}
With \eqref{eq: est int t uv final} to bound the integral over $t$ we therefore obtain 
\begin{equation}\label{vabo}
\int_0^\infty dt\, e^{-2t \varepsilon}  |\langle \psi, \mathrm{I}_{4;b}^t\psi\rangle| \leq  C\rho^{\frac{4}{3} - \frac{\gamma}{2} - \frac{\delta}{16}}\langle \psi, \mathcal{N}\psi\rangle.
\end{equation}
The estimate for the term $\mathrm{I}_{4;c}^t$, which is given by
\[
  \mathrm{I}_{4;c}^t = \int dxdydzdz^\prime\, \chi_<(x-y)\chi_<(z-z^\prime)u^{t}_\uparrow(z;y) v^t_\uparrow(z;y) a^\ast_\downarrow(u_x^<)a^\ast_\downarrow(v_x)a_\downarrow(u^{t}_{z^\prime})a_\downarrow(v^{t}_{z^\prime}),
\]
can be done in a similar way, obtaining the same result. 
Finally, the term 
  \[
    \mathrm{I}_{4;d}^t = \int dxdydzdz^\prime\, \chi_<(x-y)\chi_<(z-z^\prime) u^{t}_\uparrow(z;y)v^t_\downarrow(z^\prime;x)v^t_\uparrow(z;y) a^\ast_\downarrow(u_x^<)a_\downarrow(u^{t}_{z^\prime})
  \]
  can be estimated similarly as $\mathrm{I}_{4;a}^t$ above, using also that $|v^t_\downarrow(z^\prime;x)| \leq \rho e^{t(k_F^\downarrow)^2}$, with the result that 
  \[
    |\langle \psi, \mathrm{I}_{4;d}^t\psi\rangle| \leq \rho \|\chi_<\|_1^2  \|{u}^{t}_\uparrow\|_2\|v^t_\uparrow\|_2  \langle \psi, \mathcal{N}\psi\rangle. 
  \] 
Again \eqref{eq: est int t uv final}  then yields the validity of the bound   \eqref{vabo} also for $\mathrm{I}_{4;d}^t$. Combining the bounds, we conclude that
for $\delta \leq 8\gamma$,
\begin{equation}\label{eq: est I4 Q2 T2}
\int_0^\infty dt\, e^{-2t \varepsilon}  |\langle \psi, \mathrm{I}_4^t\psi\rangle| \leq C\rho^{\frac{4}{3} -\frac{\gamma}{2} -\frac{\delta}{16}}\langle \psi, \mathcal{N}\psi\rangle.
\end{equation}
The estimate for the term $\mathrm{I}_5^t$  can be done in the same way, obtaining the same result.

Finally we analyze $\mathrm{I}_6^t$, which is the analogue of \eqref{eq: aaa*a* in}. As already mentioned at the beginning of the proof, we shall put it in normal order in order to extract the constant contribution. Using \eqref{noI6}, we obtain $\mathrm{I}_6^t = \mathrm{I}_{6;a}^t + \mathrm{I}_{6;b}^t + \mathrm{I}_{6;c}^t$, where the first is the desired constant,  
\[
(8\pi a)^2 \int_0^\infty dt\, e^{-2t\varepsilon}\,  \mathrm{I}_{6;a}^t = -\frac{8 \pi a}{L^6}\sum_{p,r,r^\prime} (\widehat{\chi}_<(p))^2 \hat{\eta}^{\varepsilon}_{r,r^\prime}(p)\hat{u}_\uparrow(r+p)\hat{u}_\downarrow(r^\prime - p) \hat{v}_\uparrow(r)\hat{v}_\downarrow(r^\prime).
\]
The second term equals
\[
  \mathrm{I}_{6;b}^t =   \frac 1{L^6} \sum_{p,r,r'} \widehat \chi_<(p)^2 \hat u_\uparrow^t(r+p) \hat u_\downarrow^t(r'-p) \hat v^t_\uparrow(r) \hat v^t_\downarrow(r') \left( \hat a^\ast_{-r,\uparrow} \hat a_{-r,\uparrow} + \hat a^\ast_{-r',\downarrow} \hat a_{-r',\downarrow} \right)
\]
and is thus bounded by $\rho \|u_\uparrow^t\|_2 \|u_\downarrow^t\|_2  e^{t (k_F^\uparrow)^2 + t (k_F^\downarrow)^2} \mathcal{N}$, using $0\leq \widehat \chi_<(p) \leq 1$. The third term we can write in configuration space as
$$
  \mathrm{I}_{6;c}^t = \int dxdydzdz^\prime\, \chi_<(x-y)\chi_<(z-z^\prime)u^{t}_\uparrow(z;x) u^t_\downarrow(z^\prime;y) a^\ast_\uparrow(v_x)a^\ast_\downarrow(v_y)a_\uparrow(v^{t}_{z}) a_\downarrow(v^{t}_{z'}).
$$
Proceeding similarly as for $\mathrm{I}_{4;b}^t$ above, we can bound it as 
$$
  |\langle \psi,\mathrm{I}_{6;c}^t\psi\rangle| \leq   \|\chi_<\|_1^2  \| v_\downarrow\|_2 \| v_\uparrow^t\|_2  \|u_\uparrow^t\|_2 \|u_\downarrow^t\|_2  e^{t(k_F^\downarrow)^2}  \langle \psi, \mathcal{N}\psi\rangle  \leq C \rho  \|{u}^{t}_\uparrow\|_2 \|{u}^t_\downarrow\|_2 e^{t(k_F^\uparrow)^2 + t(k_F^\downarrow)^2} \langle \psi, \mathcal{N}\psi\rangle .
$$
For the integration over $t$, we can again use Lemma~\ref{lem:t} to conclude that
\begin{equation}\label{eq: est I6b I6c I6d }
\int_0^\infty dt\, e^{-2t\varepsilon}  |\langle \psi, (\mathrm{I}_{6;b}^t + \mathrm{I}_{6;c}^t )\psi\rangle|\leq C\rho^{\frac{4}{3} - \gamma} \langle \psi, \mathcal{N}\psi\rangle.
\end{equation}
Combining all the estimates 
we proved \eqref{eq: conj Q2< T2 psi}. 
\end{proof} 


\subsection{Conclusion of Proposition \ref{prop:T2}} \label{sec:conc-T2}

\begin{proof}[Proof of Proposition \ref{prop:T2}]
We shall now combine the results of Propositions~\ref{pro: prop est H0 T2}--\ref{pro: T2 Q2<} in order to prove Proposition~\ref{prop:T2}. 
The bounds \eqref{eq:H0 T2-intro}, \eqref{eq:Q4 T2-intro} and \eqref{eq:cH0 T2-intro} were already proved in Proposition \ref{pro: prop est H0 T2}, Proposition \ref{pro: prop est Q4 T2} and Proposition \ref{pro: estimate EH0 fin}, respectively. Moreover, combining Proposition \ref{pro: prop est H0 T2} and Proposition \ref{pro: T2 Q2<}  we obtain
\begin{align*}
  \langle T_2\Omega, (\mathbb{H}_0 + \mathbb{Q}_{2;<}) T_2\Omega\rangle 
  &=\int_0^1 d\lambda \, \lambda \, \partial_\lambda \langle T_{2;\lambda}\Omega, {\mathbb{Q}}_{2;<}T_{2;\lambda}\Omega\rangle + \int_0^{1}d\lambda\, \langle T_{2,\lambda}\Omega\, \mathfrak{e}_{\mathbb{H}_0}T_{2;\lambda}\Omega\rangle \\
  &= -\frac{8\pi a}{L^6}\sum_{p,r,r^\prime}(\widehat{\chi}_<(p))^2 \hat{\eta}^{\varepsilon}_{r,r^\prime}(p) \hat{u}_\uparrow(r+p)\hat{u}_{\downarrow}(r^\prime - p)\hat{v}_\uparrow(r)\hat{v}_{\downarrow}(r^\prime) \\ & \quad + \int_0^{1}d\lambda\, \langle T_{2,\lambda}\Omega , (\mathfrak{e}_{\mathbb{H}_0} + \lambda \mathcal{E}_{\mathbb{Q}_{2;<}} )T_{2;\lambda}\Omega\rangle.
\end{align*}
Inserting the definition of $\hat{\eta}_{r,r^\prime}^\varepsilon$ in \eqref{eq: def eta epsilon},  the desired estimate \eqref{eq:main-T2} follows from \eqref{eq: prop est H0 T2} and \eqref{eq: conj Q2< T2 Om}. 
\end{proof}

\section{Conclusion of Theorem \ref{thm: main up bd}} \label{sec: conclusions}

We shall now complete the proof of Theorem \ref{thm: main up bd}. As already explained in
Section~\ref{sec:sub-strategy}, the starting point is the bound 
\[
  E_{L}(N_\uparrow, N_\downarrow) \leq E_{\mathrm{FFG}} +  \langle \Omega,T_2^* T_1^* \mathcal{H}_{\mathrm{corr}} T_1T_2\Omega\rangle 
\]
using the trial state \eqref{eq: def trial state} and Proposition~\ref{pro: fermionic transf}. From Proposition~\ref{lem:eff-Hamiltonian} we further obtain
$$
\langle T_1T_2 \Omega, \mathcal{H}_{\mathrm{corr}} T_1T_2 \Omega\rangle \le
\langle T_1T_2 \Omega, \mathcal{H}_{\mathrm{corr}}^{\mathrm{eff}}   T_1T_2 \Omega\rangle + C L^3\rho^{\frac{8}{3} -\gamma - \frac{\delta}{8}}.
$$
Combining now Propositions~\ref{lem:conjugation-T1} and~\ref{prop:T2}
and taking the thermodynamic limit we find that 
\begin{align}\label{eq: final section T2 before lemmas} 
e(\rho_\uparrow, \rho_\downarrow) &\le \frac{3}{5}(6\pi^2)^{\frac{2}{3}} \left(\rho_\uparrow^{\frac{5}{3}}  + \rho_\downarrow^{\frac{5}{3}}\right)  + 8\pi a\rho_\uparrow\rho_\downarrow    \nonumber
  \\
  &\quad +\frac{(8\pi a)^2}{(2\pi)^9} \int dpdrdr^\prime \left(\frac{1}{2|p|^2} - \frac{\hat{u}_\uparrow(r+p) \hat{u}_\downarrow(r^\prime - p)}{|r+p|^2 - |r|^2 + |r^\prime - p|^2 - |r^\prime|^2 + 2\varepsilon}\right)   (\widehat{\chi}_<(p))^2\hat{v}_\uparrow(r) \hat{v}_\downarrow(r^\prime) \nonumber \\ & \quad + \mathcal{E}_{T_1 + T_2}
\end{align}
with
\[
  |\mathcal{E}_{T_1 + T_2}| \leq C \left( \rho^{\frac{7}{3} + \gamma} + \rho^{\frac{8}{3} - 2\gamma}+ 
  \rho^{\frac{7}{3} -\gamma + \frac{7}{8}\delta } + \rho^{3 - 3\gamma - \frac{3}{16}\delta}\right).
\]
To obtain \eqref{eq: final section T2 before lemmas}, we also used \eqref{eq: FFG energy} and that $(2\pi)^{-3} \int  dr\, \hat v_\sigma(r) = \rho_\sigma$. 
Recall the conditions on the parameters,  $0<\gamma< \frac{1}{3}$, $0<\delta \leq 8\gamma$, and $2\gamma + \frac{\delta}{16}\leq \frac{1}{3}$.
%
By choosing $\gamma=1/9$ and  $16/63\leq \delta \leq 8/9$, we fulfill these conditions 
and find that
$|\mathcal{E}_{T_1 + T_2}| \leq C\rho^{\frac{7}{3}+\frac{1}{9}}$.

To complete the proof of Theorem~\ref{thm: main up bd}, we still have to remove the parameter $2\varepsilon = 2\rho^{2/3 + \delta}$ in \eqref{eq: final section T2 before lemmas}, which was used to have strictly positive lower bound on the relevant denominator, and the cut-off $(\widehat{\chi}_<(p))^2$ which restricts the integration to $|p| \lesssim \rho^{1/3 -\gamma}$. 
This is the content of the following lemma.

\begin{lemma}[Removing the gap and the cut-off]\label{lem: gap constant} 
We have
  \begin{multline}\label{tocomp}
 \int dpdrdr^\prime \left(   \frac{ \hat{u}_\uparrow(r+p)\hat{u}_\downarrow(r^\prime - p)}{|r+p|^2 - |r|^2 + |r^\prime - p|^2 - |r^\prime|^2+ 2\rho^{\frac{2}{3} + \delta}}    - \frac{1}{2|p|^2}\right)(\widehat{\chi}_<(p))^2\hat{v}_\uparrow(r)\hat{v}_\downarrow(r^\prime)\\ 
  =\int dpdrdr^\prime \left(   \frac{\hat{u}_\uparrow(r+p)\hat{u}_\downarrow(r^\prime - p)}{|r+p|^2 - |r|^2 + |r^\prime - p|^2 - |r^\prime|^2}    - \frac{1}{2|p|^2}\right)\hat{v}_\uparrow(r)\hat{v}_\downarrow(r^\prime) + {O}(\rho^{\frac{7}{3} + \delta})+ {O}(\rho^{\frac{7}{3} + \gamma}) .
\end{multline}
\end{lemma}

\begin{proof}  Recall the notation $\lambda_{r,p} = |r+p|^2 - |r|^2$. We split the proof in two parts. We first prove that 
\begin{multline}\label{eq: diff gap}
 \int dpdrdr^\prime  \frac{ \hat{u}_\uparrow(r+p)\hat{u}_\downarrow(r^\prime - p)  \hat{v}_\uparrow(r)\hat{v}_\downarrow(r^\prime)}{\lambda_{r,p} + \lambda_{r^\prime, -p}+ 2\rho^{\frac{2}{3} + \delta} }(\widehat{\chi}_<(p))^2
  \\
   = \int dpdrdr^\prime  \frac{ \hat{u}_\uparrow(r+p)\hat{u}_\downarrow(r^\prime - p)\hat{v}_\uparrow(r)\hat{v}_\downarrow(r^\prime)}{\lambda_{r,p} + \lambda_{r^\prime, -p}} (\widehat{\chi}_<(p))^2 + {O}(\rho^{\frac{7}{3} + \delta}).
  \end{multline}
The difference between the two terms in \eqref{eq: diff gap} is given by 
\begin{align}
  \mathrm{I}:= 2\rho^{\frac 2 3+\delta}\int dpdrdr^\prime    \frac{ \hat{u}_\uparrow(r+p)\hat{u}_\downarrow(r^\prime - p)\hat{v}_\uparrow(r)\hat{v}_\downarrow(r^\prime)}{(\lambda_{r,p} + \lambda_{r^\prime, -p} ) (\lambda_{r,p} + \lambda_{r^\prime, -p} + 2\rho^{\frac{2}{3} + \delta}) }(\widehat{\chi}_<(p))^2  .
 \end{align}
Without loss of generality let us assume that  $k_F^\uparrow \geq k_F^\downarrow$. Rescaling the variables respect to $k_F^\uparrow$, we get
\begin{equation}\label{eq: integral with singularity}
  |\mathrm{I}| \leq C\rho^{\frac{7}{3} +\delta } \int_{\substack{|r|<1<|r+p| \\ |r^\prime|< k_F^\downarrow/k_F^\uparrow < |r^\prime - p|}} dpdrdr^\prime \frac{1 
  }{(\lambda_{r,p} + \lambda_{r^\prime, -p})^2}\leq C\rho^{\frac{7}{3} + \delta } 
\end{equation}
where the boundedness of the integral follows from Lemma \ref{lem:F2} in Appendix \ref{app: integral singularity}. 
To conclude it remains to show that 
$$
\mathrm{II}:= \int dpdrdr^\prime  \left(   \frac{ \hat{u}_\uparrow(r+p)\hat{u}_\downarrow(r^\prime - p)}{\lambda_{r,p} + \lambda_{r^\prime, -p}}    - \frac{1}{2|p|^2}\right)\hat{v}_\uparrow(r)\hat{v}_\downarrow(r^\prime)
 \left(1-\widehat{\chi}_<(p)^2\right) 
 = {O}(\rho^{\frac{7}{3} + \gamma}).
$$
The integration over $p$ runs over $|p|\geq 4\rho^{1/3-\gamma}$ only, hence we have  $\hat{u}_\uparrow(r+p)\hat{u}_\downarrow(r^\prime - p) =1$. Moreover,  using that
\begin{equation}\label{inte}
     \frac{1}{\lambda_{r,p} + \lambda_{r^\prime, -p}}    - \frac{1}{2|p|^2} 
     = \frac{p\cdot (r'-r) }{2|p|^4 } + \frac{\left( p\cdot (r'-r) \right)^2}{|p|^4 (\lambda_{r,p} + \lambda_{r^\prime, -p})} 
\end{equation}
 we can write
 $$
\mathrm{II}= \int dpdrdr^\prime  \frac{\left( p\cdot (r'-r) \right)^2}{|p|^4 (\lambda_{r,p} + \lambda_{r^\prime, -p})} \hat{v}_\uparrow(r)\hat{v}_\downarrow(r^\prime)
 \left(1-\widehat{\chi}_<(p)^2\right) 
$$
%
since the first term in \eqref{inte} integrates to zero by symmetry. On the support of $1-\chi_<(p)^2$, $\lambda_{r,p} + \lambda_{r^\prime, -p}$ is bounded from below by $C p^2$, hence
%
$$
\mathrm{II} \leq  C\rho^{\frac{8}{3}}\int_{|p| > 4\rho^{\frac{1}{3} -\gamma}}  \frac{dp}{|p|^4} = C\rho^{\frac{7}{3} +\gamma}.
$$
 This completes the proof of the lemma. 
\end{proof}


Applying Lemma~\ref{lem: gap constant} to the bound in \eqref{eq: final section T2 before lemmas} completes the proof of  Theorem \ref{thm: main up bd}. To get the explicit form \eqref{eq:def-F} of the Huang--Yang correction  one has to compute the integral on the right-hand side of  \eqref{tocomp}.  This will be carried out in Appendix \ref{app: function g}.

%
%
%
%
%
%
%
%

\bigskip

\noindent\textit{Acknowledgments.} We thank the referees for valuable remarks. This work was partially funded by the Deutsche Forschungsgemeinschaft (DFG, German Research Foundation) via the TRR 352 – Project-ID 470903074. PTN was partially supported by the European Research Council via the ERC Consolidator Grant RAMBAS – Project-Nr. 10104424. 
\appendix
\section{Useful bounds on scattering solution and cut-off functions}\label{app: bounds phi}
In this section we prove  bounds on the solution of the zero-energy scattering equation \eqref{eq: zero energy scatt eq-intro},  and on the cut-off functions  we use in our analysis.

\subsection{Bounds on scattering solution}
Let us start this section by collecting some well-known estimates for the scattering solution $\varphi_\infty$. 

\begin{lemma}[Properties of $\varphi_\infty$]\label{lem: prop of phi infty} Let $V_\infty$ be as in Assumption~\ref{asu: potential V}, 
and $\varphi_\infty$ the solution of \eqref{eq: zero energy scatt eq-intro}. Then
  \begin{equation}\label{eq: point bd phi infty}
    0 \leq \varphi_\infty(x)\leq \min \left\{ 1 , \frac a {|x|} \right\} \quad \forall x\in\mathbb{R}^3 \ , \qquad  
    \varphi_\infty(x) = \frac{a}{|x|}\quad \forall x\in\mathbb{R}^3\setminus (\mathrm{supp}V_\infty) ,
  \end{equation}
  where $a= (8\pi)^{-1} \int_{\mathbb{R}^3} V_\infty (1-\varphi_\infty)$ is the scattering length of  $V_\infty$. 
\end{lemma}
For the proof of Lemma \ref{lem: prop of phi infty} we refer to \cite{LY2d} or \cite{LSSY}.
Next we investigate $\tilde \varphi$ defined in \eqref{eq: def tilde phi}, a periodized version of $\varphi_\infty$ 
 regularized by a momentum cut-off. Recall also the definition of $\varphi$ in Definition~\ref{def:scattering-phi-eta}.

\begin{lemma}[Bounds for $\tilde \varphi$]\label{lem: bounds phi} Let $0<\gamma < 1/3$. Let $V_\infty$ as in Assumption \ref{asu: potential V} and let $\tilde \varphi$ be as in \eqref{eq: def tilde phi}. Then the following bounds hold uniformly for $L$ large,  
  \begin{equation}\label{eq: L2 norms phi}
    \|\tilde \varphi\|_{L^\infty(\Lambda)} \leq C, \qquad \|\tilde \varphi\|_{L^2(\Lambda)} \leq C\rho^{-\frac{1}{6} +\frac{\gamma}{2}},
  \end{equation}
  and 
  \begin{equation}\label{eq: L1 norms phi}
    \|\tilde \varphi\|_{L^1(\Lambda)} \leq C\rho^{-\frac{2}{3} +2\gamma}, \qquad \|\nabla\tilde \varphi\|_{L^1(\Lambda)}\leq C\rho^{-\frac{1}{3} +\gamma}, \qquad \|\Delta\tilde \varphi\|_{L^1(\Lambda)} \leq C.
  \end{equation}
  Furthermore, with $\varphi$ as in Definition~\ref{def:scattering-phi-eta},  for any $R>0$
    \begin{equation}\label{eq: phiinfty - phi}
 \lim_{L\to \infty}   \sup_{|x|\leq R}\left|\varphi_{\infty}(x) - \varphi(x)\right| = 0 .
  \end{equation}
\end{lemma}

\begin{proof}
  We start by proving  \eqref{eq: L2 norms phi}. Let $f_\infty := 1-\varphi_\infty$. By the definition  \eqref{eq: def tilde phi} and the zero-energy scattering equation \eqref{eq: zero energy scatt eq-intro} we have that 
\[
    |\widetilde{\varphi}(x)| \leq  \frac{1}{L^3}\sum_{p} |\hat{\varphi}(p)| \leq  \frac{C}{L^3}\sum_{p\neq 0} \frac{|\mathcal{F}(V_\infty f_\infty)(p)|}{|p|^2} \leq \frac{C}{L^3}\sum_{0\neq |p|\leq 1} \frac{1}{|p|^2}   + \frac{C}{L^3}\sum_{|p|>1} \frac{|\mathcal{F}(V_\infty f_\infty)(p)|}{|p|^2},
\]
where we used that $V_\infty f_\infty \in L^1(\R^3)$ because of Lemma \ref{lem: prop of phi infty}. Assumption~\ref{asu: potential V} implies that also $V_\infty f_\infty \in L^2(\R^3)$ and the function has compact support. Hence, for large enough $L$, the $L^2(\Lambda)$-norm of its periodization equals the $L^2(\R^3)$ norm, and in particular,
  \begin{equation}\label{eq: bound L2 norm Vf}
   \frac{1}{L^3}\sum_{p\in \Lambda^*}|\mathcal{F}(V_\infty f_\infty)(p)|^2 
   =\int_{\mathbb{R}^3}dx \, |V_\infty(x)f_\infty(x)|^2\leq C\|V_\infty\|_{L^2(\mathbb{R}^3)}^2\leq C.
  \end{equation}
  Therefore, by Cauchy--Schwarz and using that $\frac{1}{L^3}\sum_{|p|>1} \frac{1}{|p|^4}\leq C$, we can conclude that $|\widetilde{\varphi}(x)| \leq C$. For the $L^2$-norm, we have
  \[
    \|\widetilde{\varphi}\|_{L^2(\Lambda)}^2 = \frac{1}{L^3}\sum_p |\hat{\varphi}(p)|^2 |\widehat{\chi}_>(p)|^2 \leq \frac{C}{L^3}\sum_{|p|\geq 4\rho^{{1}/{3} -\gamma}} \frac{1}{|p|^4} \leq C\rho^{-\frac 13 +\gamma},
  \]
  where we used again that $|\hat\varphi(p)|\leq C/ p^2$. 
  \begin{remark}\label{rm:varphi} The proof above actually shows that $\varphi$ in Definition~\ref{def:scattering-phi-eta} is uniformly bounded, 
  %
  i.e.,  $\|\varphi\|_{L^\infty(\Lambda)}\le C$. However, $\|\varphi\|_{L^2(\Lambda)}$ diverges as $L\to \infty$. 
  \end{remark}
  
  In the following, we bound the $L^1$ norms of $\widetilde{\varphi}$ and its derivatives. 
With the aid of \eqref{eq: zero energy scatt eq-intro} and \eqref{eq: def chi< chi>} we can write 
  \[
    \hat{\widetilde{\varphi}}(p) = \mathcal{F}(V_\infty f_\infty \ast g),
  \]
  where $g:\mathbb{R}^3\rightarrow \mathbb{C}$ has Fourier transform $\mathcal{F}(g)(p) = \mathcal{F}({\chi}_{>,\infty})(p)/(2|p|^2)$. The function $\widetilde \varphi$ is thus the periodization of $V_\infty f_\infty \ast g$,
  and hence
 $$
 \| \widetilde \varphi \|_{L^1(\Lambda)} \leq \| V_\infty f_\infty \ast g \|_{L^1(\R^3)} \leq  \| V_\infty f_\infty \|_{L^1(\R^3)} \|  g \|_{L^1(\R^3)}.
 $$ 
It is a simple exercise, using the smoothness and support properties of $\chi_{>\infty}$, to check that $\| g\|_{L^1(\R^3)} \leq C \rho^{-2/3+2\gamma}$, yielding the first bound in \eqref{eq: L1 norms phi}. The second bound on $\nabla \widetilde\varphi$ follows in the same way, using $\| \nabla g\|_{L^1(\R^3)} \leq C \rho^{-1/3+\gamma}$. We omit the details. For the last bound in  \eqref{eq: L1 norms phi}, we can directly use that 
 $\Delta\varphi(x) = \sum_{n\in\mathbb{Z}^3}\Delta\varphi_\infty(x + nL)$ with at most one non-vanishing term in the sum, due to the compact support of $V$. This implies that $\Delta\varphi = \Delta {\varphi}_\infty$ on $\Lambda$. Using that $\widehat{\chi}_> = 1 - \widehat{\chi}_<$, it then follows that 
\[
  \Delta\widetilde{\varphi}(x) =   \Delta\varphi (x) - \int_{\Lambda}dy\, \Delta\varphi (y)\chi_<(x-y) .
\]
Hence 
\[
  \|\Delta\widetilde{\varphi}\|_{L^1(\Lambda)} 
  \leq C\|V_\infty\|_{L^1(\mathbb{R}^3)} \left( 1+ \|\chi_<\|_{L^1(\Lambda)}\right) \leq C,
\]
where we used again Assumption~\ref{asu: potential V}, \eqref{eq: zero energy scatt eq-intro},  Lemma~\ref{lem: prop of phi infty}  and the fact that $\|\chi_<\|_{L^1(\Lambda)}\leq C$ by Lemma \ref{lem:cut-off-chi} below.

It remains to prove \eqref{eq: phiinfty - phi}. We have
$$
\varphi(x) - \varphi_\infty(x) = \frac 1{L^3} \sum_{p\neq 0} \frac{\mathcal{F}(V_\infty f_\infty)(p)}{|p|^2}e^{ip\cdot x} - \frac 1{(2\pi)^3} \int_{\R^3}  dp\,\frac{\mathcal{F}(V_\infty f_\infty)(p)}{|p|^2}e^{ip\cdot x} .
$$
The function $p\mapsto \mathcal{F}(V_\infty f_\infty)(p) e^{ip\cdot x}$ is square integrable and smooth, uniformly in $x$ for bounded $x$. Approximating the integral by a Riemann sum, it is then a simple exercise to show that the right-hand side is $O(1/L)$ for large $L$, uniformly in $x$ for bounded $x$. We leave the details to the reader.
\end{proof}



\begin{lemma}[Estimates of the $L^1, L^2$ norms of $h_1$]\label{lem: fL} Let $h_1:\Lambda\to \R$ be the periodic function with the Fourier coefficients defined in \eqref{def:h1}.
Then 
$ \|h_1\|_{L^1(\Lambda)} \rightarrow 0$ and $\|h_2\|_{L^2(\Lambda)} \rightarrow 0$ as $L\rightarrow \infty$.
\end{lemma}

\begin{proof}
Using \eqref{eq: zero energy scatt eq-intro} and $\widehat\chi_< + \widehat\chi_> =1$, we  can write 
$$
  \hat{h}_1(p) = \hat g(p) \left( 1- \widehat\chi_<(p)\right)  \quad \text{with} \ 
\hat g(p) =  \mathcal{F} (V_\infty \varphi_\infty)  (p)  - \widehat {V\varphi}(p).
$$
The function $g:\Lambda\to \R$ with Fourier coefficients $\hat g(p)$ is thus the periodization of $V_\infty (\varphi_\infty - \varphi)$, and because of the compact support of $V$, we have for $x\in \Lambda = [-L/2,L/2]^3$ (and $L$ large enough)
$$
g(x) = V_\infty(x) \left( \varphi_\infty(x) - \varphi(x) \right) .
$$
From Assumption~\ref{asu: potential V} and \ref{eq: phiinfty - phi}, it readily follows that $\| g\|_{L^1(\Lambda)} \to 0$ and $\| g\|_{L^2(\Lambda)} \to 0$ as $L\to \infty$. Now 
$$
h_1(x)  = g(x)  - \int_\Lambda dy\, g(x-y) \chi_<(y),
$$
hence the same holds for $h_1$ since $\|\chi_<\|_{L^1(\Lambda)} < C$, as will be shown in Lemma~\ref{lem:cut-off-chi} below. This completes the proof.
\end{proof}

\subsection{Bounds on the cut-off functions}

The following Lemma is an easy exercise that we leave to the reader. (Compare with \cite[Prop.~4.2]{FGHP}.) 

\begin{lemma}[$L^1$ norm of the cut-off functions $\chi_<$, $\zeta^<$] \label{lem:cut-off-chi} Let $\chi_<:\Lambda \rightarrow \mathbb{R}$ and $\zeta^<: \Lambda \rightarrow \mathbb{R}$ be the cut-off functions with Fourier coefficients introduced in \eqref{eq: def chi< chi>} and \eqref{eq: def zeta Q4}, respectively. Then
\[
  \|\chi_<\|_{L^1(\Lambda)}\le C, \qquad \|\zeta^<\|_{L^1(\Lambda)}\le C
\]
\end{lemma}

The next lemma is somewhat more subtle, as it requires uniformity in $t$.

\begin{lemma}[$L^1$ norm of the function $\zeta^t$]\label{lem: zeta t L1}
Let $\zeta^{t}$ be as in \eqref{eq: def u>t} for $t\in (0,\infty)$. Then
\begin{equation}\label{zetat}
  \|\zeta^t\|_{L^1(\Lambda)} \leq C.
\end{equation}
\end{lemma}
\begin{proof}
We start by writing
\[ 
\zeta^{t} (x) = \frac{1}{L^3}\sum_k  e^{-t|k|^2} e^{ik\cdot x} - \frac{1}{L^3}\sum_k \zeta^< (k) e^{-t|k|^2} e^{ik\cdot x} = \zeta^{t}_{1}(x) + \zeta^{t}_{2}(x),
\]
with $\zeta^<$ as in \eqref{eq: def zeta Q4}. The function $\zeta_1^t$ is the periodization of $\zeta^t_\infty(x):=(2\pi)^{-3}\int_{\R^3} dp\, e^{-tp^2}e^{ip\cdot x}= (4\pi t)^{-3/2} e^{-x^2/(4t)}$, hence $\|\zeta_1^t\|_{L^1(\Lambda)} \leq \|\zeta^t_\infty\|_{L^1(\R^3)} = 1$.
To estimate the $L^1$ norm of $\zeta^t_{2}$ it is enough to notice that 
\[
  \zeta^t_{2}(x) = \int_{\Lambda}dy\, \zeta^t_1(x-y)\zeta^<(y),  
\]
which implies that $\|\zeta^t_{2}\|_{L^1(\Lambda)}\leq \|\zeta^t_1\|_{L^1(\Lambda)}\|\zeta^<\|_{L^1(\Lambda)}\leq C$, using  Lemma \ref{lem:cut-off-chi}. This yields \eqref{zetat}. 
\end{proof}


\section{Evaluation of $F(x)$}\label{app: function g}

In this section we explicitly evaluate the integral in \eqref{inte:F}, and show that it equals 
 $a^2 \rho_\uparrow^{\frac{7}{3}} F({\rho_\downarrow}/{\rho_\uparrow})$ with $F$ given in \eqref{eq:def-F}. In fact, rescaling the integrations over 
 $p,r,r^\prime$ in   \eqref{inte:F} by  $k_{F}^\uparrow = (6\pi^2)^{\frac{1}{3}}\rho_\uparrow^{1/3}$, 
we find
\begin{align}\label{eq:def-F-proof}
F(x)  = \frac 4 \pi  (6\pi^2)^{1/3}  \int_{\R^3} dp \left(   \frac x{ p^2}  - g(x,p)\right)  
\end{align}
and 
\begin{equation}\label{eq: def g}
g(x,p) = \frac 9{8\pi^2} \int_{|k+p|>1>|k|} dk \int_{|q-p|>x^{1/3}>|q|} dq\,\frac 1{(k+p)^2 - k^2 + (q-p)^2 - q^2}    .
\end{equation}
%
%

\begin{lemma}
For every $x>0$, $F(x)$ can be written explicitly as in \eqref{eq:def-F}.
\end{lemma}

\begin{proof} Let us rewrite $g$ as
\begin{align*}
g(x,p) &= \frac 9{8\pi^2} \lim_{\eps\to 0} \int dk\, \chi(k) \int dq\, \chi(x^{-1/3}q) \frac {\chi^c(k+p) \chi^c(x^{-1/3}(q-p))}{(k+p)^2 - k^2 + (q-p)^2 - q^2 + i \eps}     \\ 
&= \frac 9{8\pi^2}  \lim_{\eps\to 0} \Re \int dk\, \chi(k) \int dq\, \chi(x^{-1/3}q)\frac {1 - \chi(k+p) - \chi(x^{-1/3}(q-p))}{(k+p)^2 - k^2 + (q-p)^2 - q^2 + i \eps}   
\end{align*}
where $\chi$ stands for the characteristic function of the unit ball, $\chi^c = 1-\chi$ and we have used that 
$$
 \Re \int dk\, \chi(k) \int dq\, \chi(x^{-1/3}q) \frac {\chi(k+p)  \chi(x^{-1/3}(q-p))}{(k+p)^2 - k^2 + (q-p)^2 - q^2 + i \eps}   =0
$$
because of symmetry. In particular, $g = g_0-g_1-g_2$ with
\begin{align*}
g_0(x,p) &= \frac 9{8\pi^2}  \lim_{\eps\to 0} \Re \int dk\, \chi(k) \int dq\, \chi(x^{-1/3}q)\frac {1}{(k+p)^2 - k^2 + (q-p)^2 - q^2 + i \eps}   , \\
g_1(x,p) &= \frac 9{8\pi^2}  \lim_{\eps\to 0} \Re \int dk\, \chi(k) \int dq\, \chi(x^{-1/3}q)\frac {\chi(k+p)}{(k+p)^2 - k^2 + (q-p)^2 - q^2 + i \eps} , \\
g_2(x,p) &= \frac 9{8\pi^2}  \lim_{\eps\to 0} \Re \int dk\, \chi(k) \int dq\, \chi(x^{-1/3}q)\frac { \chi(x^{-1/3}(q-p))}{(k+p)^2 - k^2 + (q-p)^2 - q^2 + i \eps} .
\end{align*}
We claim that
$$
\int_{\R^3} dp \left(   \frac x{ p^2}  - g_0(x,p)\right)  = 0
$$
which follows from
$$
\lim_{\eps\to 0} \Re \int_{\R^3} dp \left(   \frac 1{ 2p^2}  - \frac {1}{(k+p)^2 - k^2 + (q-p)^2 - q^2 + i \eps} \right)   = 0
$$
for fixed $k$ and $q$. In fact, the Fourier transform of the integrand above is proportional to 
$$
\frac 1{|y|} \left( 1 - e^{iy (k-q)/2} e^{-\sqrt{ - (k-q)^2/4 + i \eps} |y| }\right).
$$
Taking the real part and $\eps\to 0$, this evaluates to $0$ at $y=0$. 

We conclude that 
\begin{align*}
F(x) & = \frac 4 \pi  (6\pi^2)^{1/3}  \int_{\R^3} dp \left(   g_1(x,p) + g_2(x,p) \right)  .
\end{align*}
By simple scaling, $g_2(x,p) = x^{4/3} g_1(x^{-1},x^{-1/3}p)$, hence also
\begin{align}\label{eq:F-symmetric-form}
F(x)  = \frac 4 \pi  (6\pi^2)^{1/3} \left( f(x) + x^{7/3} f(x^{-1})\right) \qquad \text{with} \quad f(x) =  \int_{\R^3}  dp\,  g_1(x,p)   .
\end{align}
We have
\begin{align*}
f(x) & = \frac 9{8\pi^2}  \lim_{\eps\to 0} \Re \int dp \int dk\, \chi(k) \int dq\, \chi(x^{-1/3}q)\frac {\chi(k+p)}{(k+p)^2 - k^2 + (q-p)^2 - q^2 + i \eps}  \\
 & = \frac 9{8\pi^2}  \lim_{\eps\to 0} \Re   \int dk \, \chi(k) \int dq\, \chi(x^{-1/3}q) \int dp\, \frac {\chi(p)}{p^2 - k^2 + (q+k-p)^2 - q^2 + i \eps} dp dq dk .
\end{align*}
Let us compute
\begin{align*}
f'(x) & =  \frac 3{8\pi^2}  \lim_{\eps\to 0} \Re   \int_{\R^3} dk\, \chi(k) \int_{\Ss^2} d\omega  \int_{\R^3} dp\, \frac {\chi(p)}{p^2 - k^2 + (x^{1/3} \omega+k-p)^2 - x^{2/3} + i \eps}  \\ 
& =  \frac 3{8\pi^2} x^{4/3} \lim_{\eps\to 0} \Re   \int_{\R^3} dk\, \chi(x^{1/3} k) \int_{\Ss^2} d\omega \int_{\R^3} dp\, \frac {\chi(x^{1/3} p)}{p^2 - k^2 + (\omega+k-p)^2 - 1+ i \eps} 
\end{align*}
and
\begin{align*}
x^2 \left( x^{-4/3} f'(x)\right)' & =    - \frac 1{8\pi^2}   \lim_{\eps\to 0} \Re   \int_{\Ss^2} d\nu  \int_{\Ss^2} d\omega \int_{\R^3} dp\, \frac {\chi(x^{1/3} p)}{p^2 - x^{-2/3}  + (\omega+x^{-1/3}\nu -p)^2 - 1+ i \eps}  \\ 
& \quad -   \frac 1{8\pi^2}  \lim_{\eps\to 0} \Re   \int_{\R^3} dk\, \int_{\Ss^2} d\omega \int_{\Ss^2}d\nu\, \frac { \chi(x^{1/3} k) }{ x^{-2/3} - k^2 + (\omega+k- x^{-1/3} \nu)^2 - 1+ i \eps} 
\end{align*}
which can further be rewritten as
\begin{align*}
x^{7/3} \left( x^{-4/3} f'(x)\right)' & =    - \frac 1{8\pi^2}   \lim_{\eps\to 0} \Re   \int_{\Ss^2} d\nu  \int_{\Ss^2} d\omega  \int_{\R^3} dp\, \frac {\chi( p)}{p^2 - 1  + (x^{1/3}\omega+\nu -p)^2 - x^{2/3}+ i \eps}  \\ 
& \quad -   \frac 1{8\pi^2}  \lim_{\eps\to 0} \Re   \int_{\Ss^2} d\omega \int_{\Ss^2}d\nu \int_{\R^3}dp\, \frac {\chi( p)}{1- p^2   + (x^{1/3}\omega+\nu -p)^2 - x^{2/3}+ i \eps}  . 
\end{align*}
For the integration over $p$, we use that (for $a^2 \geq 4b$)
$$
\lim_{\eps\to 0} \Re  \int dp\, \frac {\chi(p)}{p^2 + a p + b + i \eps}  = 2\pi -\frac \pi 2 \sqrt{a^2 - 4 b} \ln \left| \frac { b-1 - \sqrt{a^2-4b}}{b-1+ \sqrt{a^2 - 4b} }\right| + \frac \pi 2 \frac{a^2 -2b -2}{|a|} \ln \left| \frac{ b+1-|a|}{b+1+|a|}\right|
$$
as well as
$$
\lim_{\eps\to 0} \Re  \int dp]m \frac {\chi(p)}{ a p + b + i \eps}  =  \frac {2\pi b}{a^2} - \frac \pi {|a|^3} \left( a^2 - b^2\right) \ln \left| \frac{ b -|a|}{b+|a|}\right|. 
$$
Straightforward manipulations then lead to
$$
- 8 \pi x^{7/3} \left( x^{-4/3} f'(x)\right)' =  8\pi^2\left( 2   + \frac 1{x^{1/3}} (x^{2/3}-1) \ln  \frac {|1-x^{1/3}|}{1+x^{1/3}}  \right).
$$

Consequently, we conclude that 
$$
x^{-4/3} f'(x) - \frac{3}{7}A = \frac 15 \left( 9 x^{-4/3} -  2 x^{-2/3} - \left( 3 x^{-5/3} - 5x^{-1} \right)  \ln  \frac {|1- x^{1/3}|}{1+x^{1/3}}  + 2 \ln \frac{x^{2/3}}{|1-x^{2/3}|}  \right)
$$
%
and hence 
$$
f(x) - A x^{7/3} - B = \frac 9{14} x^{1/3} + \frac{99}{70} x - \frac 6{35} x^{5/3} + \frac{3}{20} \left( \frac{15}{7}  -6 x^{2/3} + 5x^{4/3} \right) \ln  \frac {|1- x^{1/3}|}{1+x^{1/3}}  + 24 x^{7/3} \ln \frac{x^{2/3}}{|1-x^{2/3}|}
$$
%
for some constants $A,B\in \R$. Since 
$$ 
0= F(0) = \frac 4\pi (6\pi^2)^{1/3}\left( f(0) + \lim_{x\to 0} x^{7/3} f(x^{-1}) \right) =  \frac 4\pi (6\pi^2)^{1/3}\left( B + A \right)
$$
necessarily $B=-A$. The value of $A$ is irrelevant as it drops out in the calculation of $F$ using \eqref{eq:F-symmetric-form}, and we arrive at the desired formula \eqref{eq:def-F}.
\end{proof}



\section{Estimate of the integral in \eqref{eq: integral with singularity}}\label{app: integral singularity}

In this section we prove the bound claimed in \eqref{eq: integral with singularity}.


\begin{lemma}\label{lem:F2} 
\[
 \sup_{0<x\leq 1}  \int_{\R^3}dp\int _{|r| < 1 < |r+p|} dr \int_{|r^\prime| <x < |r^\prime -p|} dr^\prime\frac{1}{(|r+p|^2 - |r|^2 + |r^\prime - p|^2 - |r^\prime|^2 )^2}  <\infty.  
\]
\end{lemma}

\begin{proof} 
We start by observing that it suffices to integrate only over $|p|\leq 2$. For $|p|\geq 2$, the integrand is bounded by $C/|p|^4$, hence the resulting integral is finite and bounded uniformly in $x$. 
%
%

In the following, let us use the notation $e_{k} := ||k|^2 - 1|$. Using that $|r'-p|\geq x$ and $x^2 - |r'|^2 \geq x(x-|r'|)$, the integral to consider is bounded above by 
\[
  \mathrm{I} :=  \int_{|p|< 2}dp\int _{|r| < 1 < |r+p|} dr \int_{|r^\prime| < x < |r^\prime -p|} dr^\prime\frac{1}{e_{r+p}^2 + e_{r}^2 + x^2 (x - |r^\prime|)^2 }. 
\]
%
Let us split the domain of integration into two domains, depending on whether $e_{r+p}\leq e_r$ or $e_r < e_{r+p}$, and denote the corresponding integrals by $\mathrm{I}_1$ and $\mathrm{I}_2$. 
%
In the first case, we have $1 < |r+p|^2 \leq 2 - r^2$. Dropping the positive term $e_{r+p}^2$ in the denominator, we thus get the upper bound
%
$$
  \mathrm{I}_1 \leq \frac {4\pi} 3 \int_{|r| <1} dr \int_{|r^\prime| < x} dr^\prime \,    \frac{ (2-|r|^2)^{3/2} - 1}{ e_r^2 + x^2 (x-|r'|)^2}
\leq C \int_{|r| <1} dr \int_{|r^\prime| < x} dr^\prime\, \frac{e_r}{e_r^2 + x^2 (x-|r'|)^2}
$$
where we have used that $(2-|r|^2)^{3/2} - 1 \le C(1-r^2) =Ce_r$ for $|r|\le 1$.
We further have, for any $y>0$, 
\begin{align}\label{eq:int-1/1+s^2}
\int_{|r^\prime| < x} dr'\,  \frac{y}{y^2 + x^2 (x- |r'|)^2} = 4\pi   \int_0^x ds\, \frac{ y s^2  }{y^2 + x^2(x-s)^2}\le  4\pi \int_0^\infty ds\, \frac{ y x^2 }{y^2 + x^2 s^2} = 2\pi^2  x.  
\end{align}
Hence $\mathrm{I}_1\leq C$. For the second term $I_2$ we proceed similarly. 
%
If $e_r < e_{r+p}$ we have $(2- |r+p|^2)_+ \le |r|^2 \le 1$, with $(\,\cdot\,)_+$ denoting the positive part.  
We change variables to $q=r+p$. Using $1<|q|\le |r|+|p|\le 3$ and $1- (2- |q|^2)_+^{3/2} \le C( |q|^2-1) = C e_q$ we get
\begin{eqnarray*}
  \mathrm{I}_2 &\leq& \int_{|r| <1} dr \int_{|r^\prime| < x} dr^\prime\int_{1<|q|< 3} dq\,  \ind_{\{  (2- |q|^2)_+ \le |r|^2 \le 1\}}\frac{1}{ e_q^2 + x^2 (x-|r^\prime|)^2}
  \\
  &=& \frac {4\pi}3  \int_{1<|q|< 3} dq \int_{|r^\prime| < x} dr^\prime \, \frac{  ( 1- (2- |q|^2)_+^{3/2}}{e_q^2 + x^2 (x-|r^\prime|)^2} \leq C \int_{|q| <3} dq   \int_{|r^\prime| < x} dr' \, \frac{e_q}{e_q^2 +x^2 (x-|r^\prime|)^2}
\end{eqnarray*}
which is finite due to \eqref{eq:int-1/1+s^2}. 
\end{proof}

\end{document}